\documentclass{article}
\pdfoutput=1

\usepackage[utf8]{inputenc}
\usepackage[margin=1in]{geometry}
\usepackage{color, colortbl}
\usepackage[dvipsnames]{xcolor}
\usepackage{graphicx,epstopdf}
\usepackage{amsmath,amssymb,amsthm}
\usepackage{algorithm}
\usepackage{algorithmic}
\usepackage{bm}
\usepackage[caption=false]{subfig}
\usepackage{appendix}
\usepackage{multirow}
\usepackage{mathtools}
\usepackage{braket}
\usepackage{csquotes}
\usepackage[colorlinks=true,linkcolor=red,bookmarks=true,breaklinks=true]{hyperref}
\usepackage[numbers,comma,sort&compress]{natbib}
\usepackage[english]{babel}
\usepackage[qm]{qcircuit}
\usepackage[capitalize,nameinlink,noabbrev]{cleveref}
\usepackage{tikz}
\usetikzlibrary{shapes.geometric}
\usetikzlibrary{arrows.meta}
\usetikzlibrary{positioning}
\usetikzlibrary{calc}
\usepackage{array}
\usepackage[normalem]{ulem}

\definecolor{color1}{HTML}{FFDAB9}
\definecolor{color2}{HTML}{FFA07A}
\definecolor{color3}{HTML}{FF7F50}
\definecolor{color4}{HTML}{FF6347}

\definecolor{amber}{rgb}{1.0, 0.75, 0.0}

\newtheorem{thm}{\protect\theoremname}
  \theoremstyle{plain}
  \newtheorem{lem}[thm]{\protect\lemmaname}
  \theoremstyle{remark}
  
  \theoremstyle{plain}
  \newtheorem*{lem*}{\protect\lemmaname}
  \theoremstyle{plain}
  \newtheorem{prop}[thm]{\protect\propositionname}
  \theoremstyle{plain}
  \newtheorem{cor}[thm]{\protect\corollaryname}
  
  \newtheorem{defn}[thm]{Definition}

  \newcolumntype{x}[1]{>{\centering\let\newline\\\arraybackslash\hspace{0pt}}p{#1}}
  
\PassOptionsToPackage{USenglish}{babel}
\usepackage[USenglish]{babel}
  \providecommand{\corollaryname}{Corollary}
  \providecommand{\lemmaname}{Lemma}
  \providecommand{\propositionname}{Proposition}
  \providecommand{\remarkname}{Remark}
\providecommand{\theoremname}{Theorem}

\newcommand{\norm}[1]{\|#1\|}
\newcommand{\one}{I}

\begin{document}

\title{\centering Quantum differential equation solvers: \\limitations and fast-forwarding}

\author{Dong An$^{1,2}$, \quad Jin-Peng Liu$^{4,5,6}$, \quad Daochen Wang$^{2,3,7}$, \quad Qi Zhao$^{8}$ \\ 
\footnotesize $^{1}$ Beijing International Center for Mathematical Research, Peking University, Beijing, China\\
\footnotesize $^{2}$ Joint Center for Quantum Information and Computer Science, University of Maryland, College Park, MD 20742, USA\\
\footnotesize $^{3}$ Department of Mathematics, University of Maryland, College Park, MD 20742, USA\\
\footnotesize $^{4}$ Yau Mathematical Sciences Center, Tsinghua University, Beijing, China\\
\footnotesize $^{5}$ Center for Theoretical Physics, Massachusetts Institute of Technology, Cambridge, MA 02139, USA\\
\footnotesize $^{6}$ Simons Institute and Department of Mathematics, University of California, Berkeley, Berkeley, CA 94720, USA\\
\footnotesize $^{7}$ Department of Computer Science, University of British Columbia, BC V6T 1Z4, Canada\\
\footnotesize $^{8}$ QICI Quantum Information and Computation Initiative, School of Computing and Data Science,\\ 
\footnotesize The University of Hong Kong, Pokfulam Road, Hong Kong, China\\
}
\date{}

\maketitle

\begin{abstract}
    We study the limitations and fast-forwarding of quantum algorithms for linear ordinary differential equation (ODE) systems with a particular focus on non-quantum dynamics, where the coefficient matrix in the ODE is not anti-Hermitian or the ODE is inhomogeneous. On the one hand, for generic linear ODEs, by proving worst-case lower bounds, we show that quantum algorithms suffer from computational overheads due to two types of ``non-quantumness'': real part gap and non-normality of the coefficient matrix. We then show that homogeneous ODEs in the absence of both types of ``non-quantumness'' are equivalent to quantum dynamics, and reach the conclusion that quantum algorithms for quantum dynamics work best. To obtain these lower bounds, we propose a general framework for proving lower bounds on quantum algorithms that are \emph{amplifiers}, meaning that they amplify the difference between a pair of input quantum states. On the other hand, we show how to fast-forward quantum algorithms for solving special classes of ODEs which leads to improved efficiency. More specifically, we obtain exponential improvements in both $T$ and the spectral norm of the coefficient matrix for inhomogeneous ODEs with efficiently implementable eigensystems, including various spatially discretized linear evolutionary partial differential equations. We give fast-forwarding algorithms that are conceptually different from existing ones in the sense that they neither require time discretization nor solving high-dimensional linear systems. 
\end{abstract}

\tableofcontents

\section{Introduction}

Differential equations are widely used to model the evolution and dynamics of systems in many disciplines including mathematics, physics, chemistry, engineering, biology, and economics. 
In this work, we consider the linear ordinary differential equation (ODE) system
\begin{equation}\label{eqn:ODE_general}
\begin{split}
     \frac{d}{dt} u(t) &= Au(t) + b(t), \quad t \in [0,T], \\
     u(0) &= u_{\text{in}}, 
\end{split}
\end{equation}
where $u(t) \in \mathbb{C}^N$ is the solution of the ODE, $A \in \mathbb{C}^{N\times N}$ is the coefficient matrix, and $b(t) \in \mathbb{C}^N$ is the inhomogeneous term. 
We always consider a time-independent coefficient matrix $A$ and sometimes allow for a time-dependent inhomogeneous term. We call~\cref{eqn:ODE_general} a homogeneous system of ODEs if $b(t) = 0$ for all $t$, and an inhomogeneous system otherwise. 
A quantum algorithm is said to solve\footnote{Here the meaning of solving an ODE quantumly follows existing literature on quantum ODE solvers, see, \emph{e.g.},~\cite{Berry2014,BerryChildsOstranderEtAl2017,ChildsLiu2020,LiuKoldenKroviEtAl2021,Krovi2022,FangLinTong2022}. Notice that if we are interested in recovering the unnormalized solution, the norm $\|u(T)\|$ can usually be estimated by quantum amplitude estimation~\cite{BrassardHoyerMoscaEtAl2002}. }~\cref{eqn:ODE_general} if it prepares a quantum state approximately equal to the normalized final solution $u(T)/\|u(T)\|$ up to $2$-norm error $\epsilon$, where $\|\cdot\|$ denotes the $2$-norm of a vector. 

One notable example of~\cref{eqn:ODE_general} is quantum dynamics, where $b(t) = 0$ for all $t$ and $A$ is anti-Hermitian.
Simulating quantum dynamics is also known as time-independent Hamiltonian simulation, which traces its formulation to Feynman~\cite{Feynman1982} and is widely viewed as one of the most important and promising applications of quantum computers. 
The first quantum algorithm for time-independent Hamiltonian simulation was proposed by Lloyd~\cite{Lloyd1996}, and there has been rapid and significant theoretical progress in the past few decades~\cite{BerryAhokasCleveEtAl2007,Childs2009,BerryChilds2012,BerryCleveGharibian2014,BerryChildsCleveEtAl2015,BerryChildsKothari2015,BerryNovo2016,LowChuang2017,ChildsMaslovNamEtAl2018,LowWiebe2019,ChildsOstranderSu2019,Campbell2019,Low2019,ChildsSu2019,ChildsSuTranEtAl2020,ChenHuangKuengEtAl2020,SahinogluSomma2020}. 
The best asymptotic complexity scaling is achieved by quantum signal processing (QSP) and quantum singular value transformation (QSVT), and is given by $\mathcal{O}(T\|A\| + \log(1/\epsilon)/\text{loglog}(1/\epsilon))$. This scaling matches the worst-case lower bound for time-independent Hamiltonian simulation in time $T$, error $\epsilon$, and the norm of $A$~\cite{BerryChildsKothari2015}.

This work focuses on quantum algorithms for ODE systems that are not necessarily quantum dynamics, \emph{i.e.}, $A$ is not anti-Hermitian and/or the system is inhomogeneous. 
Our work has three motivations. First, real-world systems and applications of non-quantum dynamics typically have a huge degree of freedom, leading to the need to employ high-dimensional differential equations for accurate modeling.
This makes the topic of quantum algorithms for solving general ODE systems intriguing since quantum computers can potentially solve high-dimensional problems exponentially faster than classical computers. 
Second, compared to simulating quantum dynamics, there are much fewer quantum algorithms available~\cite{Berry2014,BerryChildsOstranderEtAl2017,ChildsLiu2020,Krovi2022}. 
Most of them first discretize the time variable according to different classical numerical methods and solve the resulting linear system of equations using quantum linear system solvers (such as the HHL algorithm~\cite{HarrowHassidimLloyd2009} or advanced ones~\cite{ChildsKothariSomma2017,ChakrabortyGilyenJeffery2018,SubasiSommaOrsucci2019,LinTong2020,AnLin2022,CostaAnYuvalEtAl2022}). 
Though the linear-in-$T$, linear-in-$\|A\|$ and polylog-in-$\epsilon$ scalings can still be achieved, existing quantum algorithms for general non-quantum dynamics are not as efficient as those for quantum dynamics, in the sense that their complexity scalings typically involve extra factors related to the ``non-quantumness'' of the evolution (\emph{e.g.}, the norm of $e^{At}$, the condition number of the eigenbasis and the rate of norm decay). 
Therefore it is desirable to investigate whether there exist quantum algorithms for generic non-quantum dynamics with better asymptotic scaling or whether the ``non-quantumness'' of such systems inherently limits the efficiency of any quantum algorithm that solves it. 
Third, the best existing algorithms with linear scalings in $T$ and $\|A\|$ can still be expensive in specific applications.
Some examples include the inhomogeneous parabolic and hyperbolic partial differential equations. 
These equations involve an unbounded spatial Laplacian operator, which has a huge spectral norm after spatial discretization. 
This makes linear scaling in $\|A\|$ undesirable for efficient simulation. 
The linear scaling in $T$ also means that simulation for long times is costly. Therefore, it is worthwhile to investigate whether there exist quantum algorithms for specific and practically applicable systems that are more computationally efficient compared to their generic counterparts. 

In this work, we study quantum algorithms for solving non-quantum dynamics from two different angles: 1. their limitations when applied to general non-quantum dynamics and 2. their fast-forwarding when applied to specific non-quantum dynamics. 
On the one hand, we identify two types of ``non-quantumness'' of ODEs that cause computational overhead. 
We show that generic quantum algorithms for solving homogeneous ODEs suffer from exponential computational overhead if the coefficient matrix has at least one pair of eigenvalues with different real parts and linear computational overhead if the coefficient matrix is non-normal\footnote{A matrix $A$ is normal iff $AA^{\dagger} = A^{\dagger}A$, or equivalently $A$ is unitarily diagonalizable.}. 
To avoid the computational overheads from these two sources of non-quantumness, the coefficient matrix $A$ should be normal and all of its eigenvalues should have the same real parts.
We then show that this class of homogeneous ODEs is equivalent to quantum dynamics due to what we call the \emph{shifting equivalence}. Our hardness results can also be generalized to the inhomogeneous case, and imply that existing generic quantum algorithms for inhomogeneous ODEs cannot be substantially improved. 
On the other hand, we show that fast-forwarding and improved asymptotic complexity scalings are possible for special cases of ODEs. 
More specifically, we quadratically improve the scaling in time $T$ if the coefficient matrix $A$ is negative semi-definite, and exponentially improve the scaling in $T$ and $\norm{A}$ if $A$ has classically computable eigenvalues and quantumly implementable eigenstates. 
The latter case includes common (spatially discretized) linear evolutionary partial differential equations, such as transport equation, heat equation, advection-diffusion equation, and so on. 
Our algorithms are ``one-shot'' in the sense that they neither require time discretization nor solve high-dimensional linear systems. 
Instead, they are based on the linear combination of quantum states, which makes them conceptually different from existing generic algorithms that rely on time discretization and solving high-dimensional linear systems. 

\subsection{Limitations of generic quantum ODE solvers}

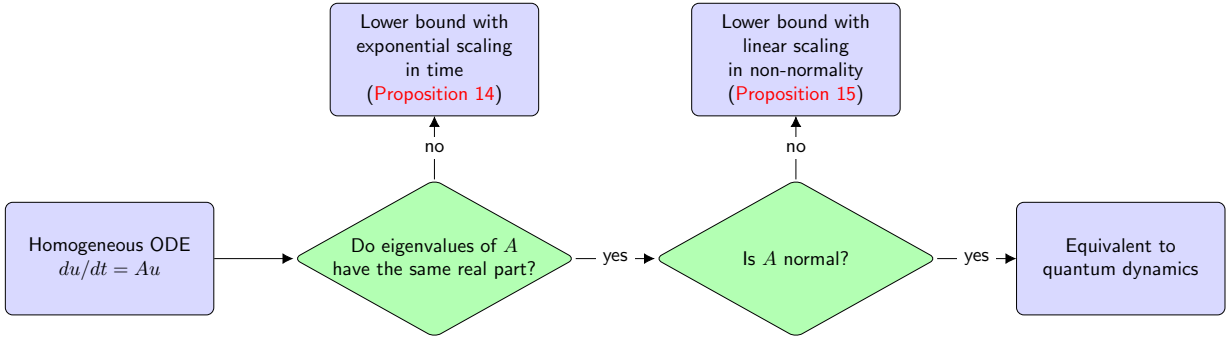
\begin{figure}
\centering
\scalebox{0.7}{
{\tikzset{%
  >={Latex[width=2mm,length=2mm]},
            base/.style = {rectangle, rounded corners, draw=black,
                          minimum width=3.7cm, minimum height=2cm,
                          text centered, font=\sffamily},
            triangle/.style = {diamond, rounded corners, draw=black,
                           minimum width = 5cm, minimum height = 2.8cm,aspect=1.5,
                           inner sep=-10ex,
                           text width=6cm,
                          text centered, font=\sffamily},
       result/.style = {triangle, fill=green!30},
       new/.style = {base, fill=blue!15},
         old/.style = {base, minimum width=2.5cm, fill=yellow!30},
}
\begin{tikzpicture}[node distance=2cm,
    every node/.style={fill=white, font=\sffamily, scale=1.07}, align=center]
  
  \node (A-known-p)     [new] {Lower bound with \\
   exponential scaling \\
   in time \\
   (\cref{prop:lb_eig_diff_homo})};
   
  \node (locate)     [result, below = 1.2 cm of A-known-p]   {Do eigenvalues of $A$ \\have the same real part?};
  
  \node (shrink)     [new, left = 1.5 cm of locate]   {Homogeneous ODE\\ $du/dt = Au$};
  
  \node (locate2)     [result, right = 1.5 cm of locate]   {Is $A$ normal?};
  
  \node (locate3)     [new, above = 1.2 cm of locate2]   {Lower bound with \\
   linear scaling \\
   in non-normality \\
   (\cref{prop:lb_non_normal_homo})};
  
  \node (locate4)     [new, right = 1.5 cm of locate2]   {Equivalent to \\quantum dynamics};

  \draw[->]              (locate) -- node {no} (A-known-p);
  
  \draw[->]              (shrink) -- (locate);
  
  \draw[->]              (locate) -- node {yes} (locate2);
  
  \draw[->]              (locate2) -- node {no} (locate3);
  
  \draw[->]              (locate2) -- node {yes} (locate4);
\end{tikzpicture}}}
\caption{Flowchart of our hardness results for generic quantum homogeneous ODE solvers and their implications. 
}
\label{fig:LB_flowchart}
\end{figure}

\subsubsection{Results}

We first focus on homogeneous ODE systems and show that the most efficient class of homogeneous ODE systems that generic quantum algorithms can solve is quantum dynamics. 
We identify two sources of non-quantumness that introduce computational overhead: real part gap (\emph{i.e.}, the maximal difference in the eigenvalues' real parts) and non-normality of the coefficient matrix, and prove two corresponding lower bounds. Specifically, when the real part gap $\delta > 0$ (\emph{i.e.}, there exists at least a pair of eigenvalues with different real parts), generic quantum ODE solvers cannot scale better than $o(e^{\delta T})$, an exponential overhead in both the real part gap and the evolution time. 
This lower bound is related to the inefficiency of quantum post-selection (assuming $\text{BQP}\neq\text{PostBQP}$, which is reasonable since $\text{PostBQP}$ is equal to an apparently massive complexity class called $\text{PP}$ \cite{Aaronson2004}), except our result is in the query model and generalizes to the case when the coefficient matrix $A$ has a possibly non-orthonormal eigenbasis because we do not assume $A$ is unitarily diagonalizable. 
Second, when $A$ is non-normal (even if all the eigenvalues have the same real parts), we show another 
non-existence of generic quantum ODE solvers with scaling better than $o(\mu(A))$, where $\mu(A) = \|AA^{\dagger}-A^{\dagger}A\|^{1/2}$ is a measure of non-normality~\cite{ElsnerPaardekooper1987} in the sense that $\mu(A)$ becomes larger if $A$ is more non-normal. 
Therefore, the most efficient class that generic quantum ODE solvers can solve is the one that sufficiently avoids both types of ``non-quantumness'', implying that the coefficient matrix should be normal and all the eigenvalues have the same real parts. 
Systems with such features are quantumly equivalent to quantum dynamics due to an observation that we call \emph{shifting equivalence}: the normalized solution of a homogeneous ODE system remains unchanged if the coefficient matrix is shifted by a real number. 
Therefore we can shift the coefficient matrix by the real part of its eigenvalues and the resulting one becomes an anti-Hermitian matrix. 
Our results and their implications are summarized in~\cref{fig:LB_flowchart}. 
We also remark that our lower bound results can also be generalized to inhomogeneous ODE solvers. 

Moreover, we prove lower bounds in terms of the time $T$ and error $\epsilon$ for solving inhomogeneous ODEs even when all the
eigenvalues of $A$ have non-positive real parts or the real part gap is $0$. 
For generic Hamiltonian simulation where $A$ is anti-Hermitian and $b = 0$, Berry \emph{et.~al.} \cite{BerryAhokasCleveEtAl2007} proved an $\Omega(T)$ query lower bound which indicates that no generic fast-forwarding in possible. This bound was later improved to $\Omega(T + \log(1/\epsilon)/\log\log(1/\epsilon))$ by Berry, Childs, and Kothari~\cite{BerryChildsKothari2015}. However, these results do not say anything about what the lower bound should be for ODEs where $A$ is not anti-Hermitian or $b\neq 0$. We extend these results to the setting where $A$ has negative logarithmic norm and $b\neq 0$. This setting is disjoint from that of Hamiltonian simulation since $b\neq 0$ and the logarithmic norm of an anti-Hermitian coefficient matrix is $0$. In this setting, we prove by reduction that generic quantum ODE solvers require $\min\left\{ \Omega(\log^{\alpha}(1/\epsilon)),  \Omega(T^{\alpha})\right\}$ queries if solving the linear system of equations $Ax = b$ requires $\Omega(\log^{\alpha}(1/\epsilon))$ queries. Notice that our lower bound depends on the lower bound for solving a particular class of linear system of equations since $A$ has negative logarithmic norm. It is proved in~\cite{HarrowKothari2022} that when $A$ is unrestricted, solving $Ax=b$ requires at least $\Omega(\log(1/\epsilon))$ queries. We leave to future work the investigation of how this lower bound changes if $A$ is restricted to have negative logarithmic norm. 

\subsubsection{Methods}

We prove our lower bounds on ODE solvers by first introducing a general framework for proving lower bounds on a class of quantum algorithms that we call \emph{amplifiers}.
Our framework generalizes techniques used by Somma and Subaşi to lower bound the complexity of verifying the output of a quantum linear system solver~\cite{SommaSubasi2021} and relates to those used to lower bound the complexity of quantum state discrimination~\cite{Helstrom1969,AbramsLloyd1998,Chefles2000,Montanaro2008,ChildsYoung2016,LiuKoldenKroviEtAl2021}. Suppose we have access to unitaries $O$ and $O^{-1}$ (and their controlled versions) such that $O$ is promised to (consistently) prepare either one of the states $\ket{\psi}$ or $\ket{\phi}$, with $|\braket{\psi|\phi}| \geq 1-\epsilon$ (\emph{i.e.}, the two states are close). If we have an algorithm $\mathcal{A}$ using queries to $O$ and $O^{-1}$ that amplifies the distance between $\ket{\psi}$ and $\ket{\phi}$ to a constant, then we can use $\mathcal{A}$ a constant number of times to readily determine which state $O$ prepares~\cite{Peres1995}. But we know that the quantum state discrimination problem is hard when $\ket{\psi}$ and $\ket{\phi}$ are close, so $\mathcal{A}$ must make many queries to $O$ and $O^{-1}$.
We call quantum algorithms like $\mathcal{A}$ that amplify the distance between two quantum states (to constant) \emph{amplifiers}, and we show that amplifiers must use $\Omega(1/\sqrt{\epsilon})$ queries to $O$ and $O^{-1}$.

Our lower bounds on quantum ODE solvers follow directly from our framework since quantum ODE solvers are amplifiers: the time evolution operator $e^{AT}$ of non-quantum dynamics is no longer unitary and thus small differences in the initial states can lead to large differences in the final states.

We highlight the fact that our lower bound on amplifiers differs from the standard lower bound on quantum state discrimination because the two settings differ. In the standard state discrimination setup~\cite{Helstrom1969,Montanaro2008,LiuKoldenKroviEtAl2021}, an information-theoretic perspective is taken that only assumes access to \emph{copies} of the quantum states that are to be discriminated. In contrast, our setup is motivated from an algorithmic perspective that allows access to quantum state preparation oracles $O$ and their inverses $O^{-1}$. Allowing access to $O^{-1}$ is more appropriate when proving lower bounds on quantum algorithms (like quantum ODE solvers) because we typically know the circuit of $O$ and can therefore easily implement $O^{-1}$ by inverting each gate in the circuit of $O$ and reversing their order. As a result of such stronger access to the quantum states, our lower bound is $\Omega(1/\sqrt{\epsilon})$ which is quadratically smaller than the state discrimination lower bound of $\Omega(1/\epsilon)$. 

Our lower bound framework may be of independent interest since it can be used to prove lower bounds on any quantum algorithm that can be viewed as an amplifier. For example, it can be easily used to recover lower bounds on quantum linear system solvers. We have used our framework to understand the power of quantum algorithms for solving ``non-quantum'' problems. More generally, we expect our framework to help elucidate the boundary between problems that are and are not efficiently solvable by quantum algorithms.

\subsection{Fast-forwarding of tailored quantum ODE solvers}

\begin{table}[t]
\renewcommand{\arraystretch}{1.8}
    \centering
    \scalebox{0.93}{
    \begin{tabular}{x{6.2cm}|x{4.9cm}|x{1.4cm}|x{1.6cm}|x{1.2cm}}\hline\hline
      \multicolumn{2}{c|}{\textbf{Assumptions}} &  \multicolumn{3}{c}{\textbf{Query Complexity}} \\\hline 
      $A$ & $b(t)$ & $T$ & $\|A\|$ & Norms \\\hline
      Eigenvalues with non-positive real parts & General & $\widetilde{\mathcal{O}}(T)$ & $\widetilde{\mathcal{O}}(\|A\|)$ & $\widetilde{\mathcal{O}}(g)$ \\\hline
      \multirow{2}{16em}{Known eigenstates and non-zero eigenvalues with non-positive real parts} & Time-independent & \color{red}{$\mathcal{O}(1)$} & \color{red}{$\mathcal{O}(1)$} & \color{amber}{$\mathcal{O}(g')$} \\\cline{2-5}
       & Time-dependent & $\mathcal{O}(T)$ & \color{red}{$\mathcal{O}(1)$} & \color{amber}{$\mathcal{O}(g'')$} \\\hline
       \multirow{2}{16em}{Known eigenstates and eigenvalues with non-positive real parts} & Time-dependent, worst-case & $\mathcal{O}(T)$ & \color{red}{$\mathcal{O}(1)$} & \color{amber}{$\mathcal{O}(g'')$}  \\\cline{2-5}
       & Time-dependent, best-case & \color{red}{$\mathcal{O}(1)$} & \color{red}{$\mathcal{O}(1)$} & -- \\\hline\hline
    \end{tabular}
    }
    \caption{ Summary of our fast-forwarding results. Here the top result (on the third line) is the scaling of the best generic quantum ODE solver~\cite{Krovi2022} and serves as the baseline. The parameters $g$, $g'$, and $g''$ are defined by $g \coloneqq \max_{t\in[0,T] }\|u(t)\|/\|u(T)\|$, $ g' \coloneqq (\|u(0)\|+\|b\|)/\|u(T)\|$, and $g'' \coloneqq (\|u(0)\|+ T^{-1}\int_0^T\|b(t)\|dt)/\|u(T)\|$. A dash means the corresponding parameter does not exist in that case. The complexities are measured by the number of queries to the input oracles of $A$, $b(t)$ and $u(0)$.  }
    \label{tab:comp}
\vspace{-0.5em}
\end{table}

The worst-case lower bounds discussed above apply to quantum solvers of generic ODEs. However, these lower bounds may be bypassed by quantum solvers tailored to solve particular classes of ODEs (see~\cref{fig:ODE_classification}). Indeed, in the case of Hamiltonian simulation, we know for example that fast-forwarding is possible when the Hamiltonian is $1$-sparse~\cite{ahokas2004improved}, can be diagonalized via an efficient quantum circuit~\cite{AtiaAharonov2017}, or corresponds to real-space dynamics~\cite{LowWiebe2019,AnFangLin2021,ChildsLengLiEtAl2022} or quadratic fermionic and bosonic systems~\cite{AtiaAharonov2017,GuSommaSahinoglu2021}.

In contrast to previous work, we study the fast-forwarding of particular classes of \emph{non-quantum} dynamics.
Our results and comparisons with the best generic algorithms are summarized in~\cref{tab:comp}. 

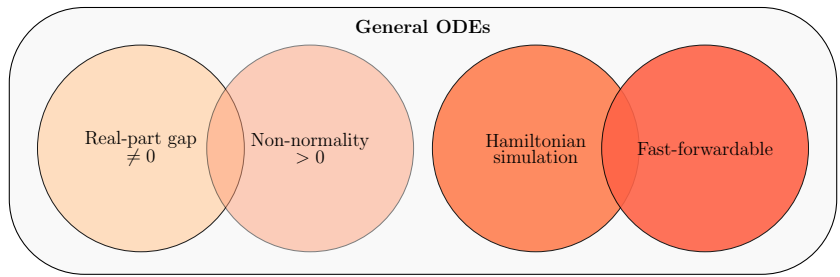
\begin{figure}[ht]
\centering
\scalebox{0.5}{
\begin{tikzpicture} [set/.style = {draw,
circle,
minimum size = 5.5cm,
fill = {rgb,255:red,244;green,111;blue,68},
text opacity = 1}]

\coordinate (top-left) at (-6cm,3cm);
\coordinate (bottom-right) at (16cm,-4.1cm);

\draw[thick, rounded corners=2cm, fill = lightgray!10] (top-left) rectangle (bottom-right);

\node at (5cm,2.5cm) {\Large \textbf{General ODEs}};

\node (A) at (-2.5cm,-0.75cm) [set,opacity = 0.9,fill=color1, align=center,font=\fontsize{15}{12}\selectfont] {Real-part gap \\ $\neq 0$};

\node (B) at (2cm,-0.75cm) [set,opacity = 0.5,fill=color2,align=center,font=\fontsize{15}{12}\selectfont] {Non-normality \\ $> 0$};

\node (C) at (8cm,-0.75cm) [set,opacity = 0.9, fill=yellow,fill=color3,align=center,font=\fontsize{15}{12}\selectfont] {Hamiltonian \\ simulation};

\node (D) at (12.5cm,-0.75cm) [set,opacity = 0.9,fill=color4,font=\fontsize{15}{12}\selectfont] {Fast-forwardable};

\end{tikzpicture}
}
\caption{A characterization of all linear ODEs (including both homogeneous and inhomogeneous ones). The two left-most circles cover the classes of ODEs for which quantum algorithms are inefficient, the circle third from the left covers Hamiltonian simulation which can be viewed as the baseline for quantum algorithms. The right-most circle covers fast-forwardable ODEs for which we design tailored quantum algorithms with improvements in $T$ and/or $\|A\|$. Note that the class of homogeneous ODEs is fully covered by the three left-most circles, as shown in~\cref{fig:LB_flowchart}.}
\label{fig:ODE_classification}
\end{figure}

\subsubsection{Homogeneous ODE}

In the homogeneous case, we obtain exponential speedups assuming $A$ is a normal matrix with an eigenbasis that can be efficiently prepared on a quantum computer and eigenvalues that can be efficiently computed classically. This assumption is similar to the ``QC-solvable'' condition of Atia and Aharonov \cite{AtiaAharonov2017}, which they show implies a violation of the so-called computational time-energy uncertainty principle. We sometimes refer to the assumption as \emph{knowing} the eigenvalues and eigenvectors of $A$. Under this assumption, we show how to implement a block-encoding of $e^{AT}$ with query complexity independent of $T$ and $\|A\|$. Our technique is similar to that used to fast-forward $1$-sparse Hamiltonian simulation~\cite{ahokas2004improved}. Suppose $A = U\Lambda U^{-1}$ is the eigen-decomposition of $A$. 
Then $e^{AT} = Ue^{\Lambda T}U^{-1}$ can be efficiently computed, since $U$ is assumed to be efficiently implementable, and the diagonal transformation $e^{\Lambda T}$ can be implemented using controlled rotations with $\mathcal{O}(1)$ query complexity. 
We remark that, although in both scenarios we can achieve speedup explicitly in the parameters $T$ and/or $\|A\|$, the overall complexities still suffer from the norm decay rate due to a post-selection step. As a result, the overall complexity of preparing the final solution might still have exponential overhead in the evolution time (when $A$ is strictly negative definite), and a genuine fast-forwarding can only be achieved by placing further assumptions on the input states (\emph{e.g.}, the input state has non-trivial overlap with the zero-energy eigenstate).

The key assumption for the aforementioned speedup is that we know the eigenvalues and eigenvectors of $A$. 
We remark that when this assumption is weaken and we only assume certain block-encoding accesses to a negative semi-definite matrix, we may also obtain quadratic speedups in $T$ and $\|A\|$. 
These results are similar to existing work that show how to quadratically speed up the preparation of quantum Gibbs states~\cite{GilyenSuLowEtAl2019,ChowdhurySomma2016}. They are also related to the quadratic fast-forwarding of quantum walks on graphs~\cite{ApersSarlette2019,ApersGilyenEtAl2021,ApersChakrabortyNovoEtAl2022}. 
In this case, the key technique is that $e^{AT}$ can be approximated using a polynomial with degree only $\mathcal{O}(\sqrt{\|A\|T})$. 

\subsubsection{Inhomogeneous ODE}

We generalize our fast-forwarding results to inhomogeneous ODEs. 
We first consider the time-independent inhomogeneous term $b(t) \equiv b$, and show that, once the homogeneous evolution operator $e^{AT}$ can be fast-forwarded, the corresponding inhomogeneous ODE can also be fast-forwarded with the same speedup. 
Our algorithm first separately computes the homogeneous and the inhomogeneous parts of the solution and then combines them together using the technique of linear combination of quantum states.
More specifically, notice that the solution of~\cref{eqn:ODE_general} with constant $b$ can be written as 
\begin{equation}\label{eqn:ODE_solu_intro}
    u(T) = e^{AT} u(0) + \int_0^T e^{A(T-s)} b \, ds, 
\end{equation}
where the first part is exactly the solution of the corresponding homogeneous equation. 
For the second part, when the eigenvalues and eigenstates of $A$ are known, the operator $\int_0^T e^{A(T-s)} ds$ also has known eigensystem and can be exponentially fast-forwarded as in the homogeneous case. 
Finally, the two parts of \cref{eqn:ODE_solu_intro} can be linearly combined via a circuit similar to that arising in the linear combination of unitaries (LCU) technique~\cite{ChildsKothariSomma2017}. 

In the time-independent case, our algorithm for solving inhomogeneous ODE systems is conceptually different from existing quantum ODE solvers. 
To the best of our knowledge, all existing algorithms for inhomogeneous ODE systems require two extra steps: 1. dividing the entire time interval using a very small time step size, and 2. transforming the problem to that of solving a linear system of equations in much higher dimension. 
However, our algorithm takes a one-shot strategy in the sense that our algorithm directly maps the input state to the final solution. 
It neither requires time discretization nor solving a high-dimensional linear system of equations. 
This is a more natural and straightforward way of designing quantum ODE solvers and facilitates fast-forwarding. 
We also expect our one-shot strategy to be applicable to the design of quantum solvers of generic ODEs with generic coefficient matrices $A$, since $e^{AT}$ can still be block-encoded (though the complexity might be super linear in $T$ or $\|A\|$) using the contour integral method~\cite{TakahiraOhashiSogabeEtAl2020,TongAnWiebe2021,FangLinTong2022}. 

We further generalize the exponential speedup in $T$ and $\|A\|$ to the case where the inhomogeneous term $b(t)$ is time-dependent and its norm is known at all times. In this case, the contribution of the inhomogeneous term to the final solution, $\int_0^T e^{A(T-s)} b(s) ds$, becomes a genuine integral over time rather than a matrix-vector multiplication. 
Nevertheless, we can still efficiently compute such an integral as follows. 
We first discretize the integral using first-order numerical quadrature (\emph{e.g.}, the Riemann sum formula), which results in a linear combination of vectors which again can be computed using the linear combination of quantum states. Although the number of the quadrature nodes needs to be very large to control the approximation error, the computational overhead is only logarithmic thanks to the efficiency of the LCU technique. Overall, the complexity of computing the final solution remains independent of the time $T$ and the norm $\|A\|$. Our generalization to time-dependent $b(t)$ crucially relies on the fact that the integral can be efficiently computed with only logarithmic scaling in the number of quadrature nodes. This fact is also the key to the efficiency of the truncated Dyson series method for interaction-picture Hamiltonian simulation~\cite{LowWiebe2019} and the qHOP method for time-dependent Hamiltonian simulation~\cite{AnFangLin2022}. In the time-dependent case, although we use time discretization to approximate the integral due to $b(t)$'s time dependence, our algorithm still does not require solving a high-dimensional linear system of equations.

\subsubsection{Application to evolutionary PDEs}

We identify various instances of high-dimensional systems of linear time-evolutionary partial differential equations (PDEs) that can be fast-forwarded exponentially in $T$ and $\|A\|$. 
Generally speaking, evolutionary PDEs can be classified into three types: parabolic (with degeneration), hyperbolic, and higher-order. 
All of these equations involve spatial differential operators such as the divergence and Laplacian operators. 
A key observation is that these differential operators (after spatial discretization) have known eigenvalues and eigenstates, since they can be diagonalized by the quantum Fourier transform circuit and the eigenvalues are explicit expressions in frequencies. 
Therefore we can apply our fast-forwarded algorithm for ODEs with known eigensystem. 

In~\cref{sec:app_parabolic_pde}, we study parabolic PDEs and show how to simultaneously compute the eigenvalues and prepare the eigenstates of the (spatially discretized) divergence and Laplacian operators. 
We discuss typical instances of parabolic PDEs that can be fast-forwarded, including the transport equation from continuum and statistical physics, the inhomogeneous heat equation from thermodynamics, and the advection-diffusion equation that models chemical processes as well as biological and ecological networks.
In Appendix, we provide more applications of our algorithms to hyperbolic PDEs and higher-order PDEs, such as the wave equation, the Klein-Gordon equation, the Airy equation, and the beam equation.

\subsection{Organization}

The rest of the paper is organized as follows. 
We first introduce preliminary results in~\cref{sec:preliminary}, including the notation and existing generic quantum ODE solvers. 
Then we discuss and prove our lower bound results for generic ODE solvers  in~\cref{sec:lower_bounds} and fast-forwarding algorithms and results for specific applications in~\cref{sec:fast_forwarding}. 
We conclude with a brief summary and several open questions in~\cref{sec:conclusion}. 

\section{Preliminaries}\label{sec:preliminary}

In this section, we introduce our notation and discuss existing generic algorithms for solving ODEs. 
In later analysis, we will also use some more results on linear algebra preliminary lemmas and quantum linear algebra operations, and these are introduced in~\cref{app:proof_la_lemma}. 

\subsection{Notation}

Let $u = (u_0,u_1,\cdots,u_{N-1})^{T}$ be an $N$-dimensional (possibly unnormalized) column vector. 
We use $u^*$ to denote its conjugate transpose, $\|u\|$ to denote its vector 2-norm, and $\|u\|_1$ to denote its vector 1-norm. 
The notation $\ket{u}$ represents a (pure) quantum state that is the normalized vector under 2-norm
\begin{equation}
    \ket{u} = u/\|u\|. 
\end{equation}
For quantum states, we follow the standard bra-ket notation where $\bra{u}$ denotes the conjugate transpose of $\ket{u}$ and $\braket{u|v}$ denotes the inner product between $\ket{u}$ and $\ket{v}$. 

For two quantum states $\ket{u}$ and $\ket{v}$, we write $\ket{u}_a\ket{v}_w$ as the multiple-register quantum state constructed by the tensor product of $\ket{u}$ and $\ket{v}$. 
Here the subscript ``a'' means ``ancilla'' and ``w'' means ``working register.''
We may use other letters (to be specified later) to represent the name of the corresponding register in the subscripts. 
The notation $\ket{\perp}$ generally represents the parts in a quantum state that are orthogonal to those of interest. 

Let $A$ be an $N$-by-$N$ matrix. 
We use $A^{\dagger}$ to denote its conjugate transpose, \emph{i.e.}, the adjoint operator of $A$. 
The norm $\|A\|$ without subscript represents the spectral norm, \emph{i.e.}, the matrix 2-norm 
\begin{equation}
    \|A\| = \sup_{u\neq 0} \|Au\|/\|u\|. 
\end{equation}
The norm $\|A\|_1$ with subscript $1$ is the Schatten 1-norm
\begin{equation}
    \|A\|_1 = \text{Tr}\left(\sqrt{A^{\dagger}A}\right). 
\end{equation}
The trace distance between two matrices $A$ and $B$ is defined to be $\frac{1}{2}\|A-B\|_1$, and the trace distance between two pure states $\ket{u}$ and $\ket{v}$ is the trace distance between the corresponding density matrices $\ket{u}\bra{u}$ and $\ket{v}\bra{v}$. 

Let two functions $f,g\colon \mathbb{R}_{>0} \to \mathbb{R}_{>0}$ represent some complexity scalings. 
We write $f = \mathcal{O}(g)$ if there exists a constant $C > 0$, independent of the arguments of $f$ and $g$, such that $f(x) \leq Cg(x)$ for all sufficiently large $x$. We write $f = \Omega(g)$ if $g = \mathcal{O}(f)$. We write $f = \Theta(g)$ if $f = \mathcal{O}(g)$ and $f = \Omega(g)$. We write
$f = \widetilde{\mathcal{O}}(g)$ if $f = \mathcal{O}(g\text{~poly}\log(g))$. 
We write $f = o(g)$ if $f/g \rightarrow 0$ as $x \rightarrow \infty$. 

\subsection{Existing generic algorithms for solving ODEs}

The best generic algorithm for solving ODEs with time-independent $A$ and $b$ is~\cite{Krovi2022} based on Taylor's expansion, and that for ODE with time-dependent $b(t)$ is the quantum spectral method~\cite{ChildsLiu2020}\footnote{The quantum spectral method can also deal with the time-dependent coefficient matrix $A(t)$, but in our work, we constrain ourselves with time-independent $A$.}. 
We summarize the main results of these two algorithms for comparison with our lower bounds and fast-forwarding results. 

\begin{lem}[\cite{Krovi2022}]\label{lem:DEsolver_TI}
    Consider solving~\cref{eqn:ODE_general} with time-independent $b$. 
    Suppose $A$ is an $N$-by-$N$ $s$-sparse matrix with sparse input oracles. 
    For $u(0)$ and $b$ we assume that their norms are known and their preparation oracles are given. 
    Then there exists a quantum algorithm which produces an $\epsilon$-approximation of $u(T)/\|u(T)\|$ in the $2$-norm sense, succeeding with probability $\Omega(1)$ with a flag indicating success, with query complexity 
    \begin{equation}
        \widetilde{\mathcal{O}}\left( g T \|A\| C(A) \text{poly}\left(s,\log(N), \log\left(1+{Te^2\|b\|}/{\|u(T)\|}\right), \log(1/\epsilon)\right) \right), 
    \end{equation}
    where 
    \begin{equation}
        g = \frac{\sup_{t\in[0,T]}\|u(t)\| }{\|u(T)\|}, \quad C(A) = \sup_{t\in[0,T]} \|e^{At}\|. 
    \end{equation}
\end{lem}

\begin{lem}[\cite{ChildsLiu2020}]\label{lem:DEsolver_TD}
    Consider solving~\cref{eqn:ODE_general} with time-dependent $b(t)$. 
    Suppose $A$ is an $N$-by-$N$ $s$-sparse matrix with sparse input oracles, and can be diagonalized as $A = V\Lambda V^{-1}$, where $\Lambda = \text{diag}(\lambda_0,\cdots,\lambda_{N-1})$ with $\text{Re}(\lambda_j) \leq 0$ for each $j$. 
    For $u(0)$ and $b(t)$ we assume that their norms are known and their preparation oracles are given. 
    Then there exists a quantum algorithm which produces an $\epsilon$-approximation of $u(T)/\|u(T)\|$ in the $2$-norm sense, succeeding with probability $\Omega(1)$ with a flag indicating success, with query complexity 
    \begin{equation}
        \widetilde{\mathcal{O}}\left( g T \|A\| s \kappa_V \text{~poly}\log\left(\frac{ \max_{t\in[0,T],k \geq 1}\|u^{(k)}(t)\| }{\epsilon \|u(T)\|}\right) \right), 
    \end{equation}
    where 
    \begin{equation}
        g = \frac{\sup_{t\in[0,T]}\|u(t)\| }{\|u(T)\|}, \quad \kappa_V =\|V\|\|V^{-1}\|.  
    \end{equation}
\end{lem}

In most existing algorithms based on solving a linear system of equations, the scaling in the number of queries to $A$ and to the preparation oracle of $u(0)$ are asymptotically the same. However, a recent work~\cite{FangLinTong2022} proposes a time-marching strategy for solving homogeneous ODE where the number of queries to the preparation oracle is asymptotically smaller. We summarize the result as follows.\footnote{Similar as the quantum spectral method, the time marching method can also deal with the time-dependent coefficient matrix $A(t)$, but in our work we constrain ourselves with time-independent $A$. }
\begin{lem}[\cite{FangLinTong2022}]\label{lem:DEsolver_TM}
    Consider solving the homogeneous case of~\cref{eqn:ODE_general} (\emph{i.e.}, $b = 0$). 
    Suppose $\|u(0)\| = 1$, and we are given an $(\alpha,n_A,0)$-block-encoding of the matrix $A$ and the preparation oracle of $u(0)$. 
    Then there exists a quantum algorithm which produces an $\epsilon$-approximation of $u(T)/\|u(T)\|$ in the $2$-norm sense, succeeding with probability $\Omega(1)$ with a flag indicating success, using 
    \begin{equation}
        \widetilde{\mathcal{O}}\left( \alpha^2 T^2 Q \log(1/\epsilon) \right)
    \end{equation}
    queries to the block-encoding of $A$, and 
    \begin{equation}
        \mathcal{O}\left( Q \right)
    \end{equation}
    queries to the preparation oracle of $u(0)$. 
    Here 
    \begin{equation}
        Q = \frac{\|e^{AT}\|}{\|u(T)\|}. 
    \end{equation}
\end{lem}

\section{Limitations}\label{sec:lower_bounds}

In this section, we consider ODEs with time-independent coefficient matrix $A$ and inhomogeneous term $b$
\begin{equation}\label{eqn:ODE}
\begin{split}
     \frac{d}{dt}u(t) &= Au(t) + b, \quad t \in [0,T], \\
     u(0) &= u_{\text{in}}. 
\end{split}
\end{equation}
We discuss the limitations of generic quantum algorithms for solving ODEs by proving worst-case lower bounds. 
Motivated by the quantum state discrimination problem, we first discuss and prove a lower bound for quantum algorithms as amplifiers given access to the preparation oracle and its inverse. 
Then we introduce the idea of the lower bounds for quantum ODE solvers using the lower bound of amplifiers and prove them in the homogeneous case. 
After that, we discuss the quantum shifting equivalence for homogeneous ODEs and explain the implications of our lower bounds. 
Finally, we present a generalization to the inhomogeneous case, implying that existing quantum algorithms cannot be significantly improved. 

\subsection{Quantum state discrimination}

\begin{figure}[ht]
    \centerline{
    \Qcircuit @R=1em @C=1em {
    \text{Ancilla} \ket{0} \quad\quad\quad\quad\quad\quad & \multigate{4}{U_1} & \qw & \multigate{4}{U_2} & \ctrl{4} & \multigate{4}{U_3} & \qw & \qw & & & \multigate{4}{U_q} & \gate{O^{\pm 1}} & \multigate{4}{U_{q+1}} & \meter & \\
    \text{Ancilla} \ket{0} \quad\quad\quad\quad\quad\quad & \ghost{U_1} & \qw & \ghost{U_2} & \qw & \ghost{U_3} & \gate{O^{\pm 1}} & \qw & & & \ghost{U_q} & \qw & \ghost{U_{q+1}} & \meter & \\ 
    \text{Ancilla} \ket{0} \quad\quad\quad\quad\quad\quad & \ghost{U_1} & \ctrl{1} & \ghost{U_2} & \qw & \ghost{U_3} & \qw & \qw & \cdots & & \ghost{U_q} & \ctrl{-2} & \ghost{U_{q+1}} & \meter & \\
    \text{System} \ket{0} \quad\quad\quad\quad\quad\quad & \ghost{U_1} & \gate{O^{\pm 1}} & \ghost{U_2} & \qw & \ghost{U_3} & \ctrl{-2} & \qw & & & \ghost{U_q} & \qw & \ghost{U_{q+1}} & \qw & \\
    \text{System} \ket{0} \quad\quad\quad\quad\quad\quad & \ghost{U_1} & \qw & \ghost{U_2} & \gate{O^{\pm 1}} & \ghost{U_3} & \qw & \qw & & & \ghost{U_q} & \qw & \ghost{U_{q+1}} & \qw & \\
    }
    }
    \caption{ An illustration of a quantum query algorithm. Here $O$ represents the given preparation oracle for the state to be determined, and $U_j$ are unitary operators that are independent of the choice of the state. }
    \label{fig:circuit_amplifier}
\end{figure}
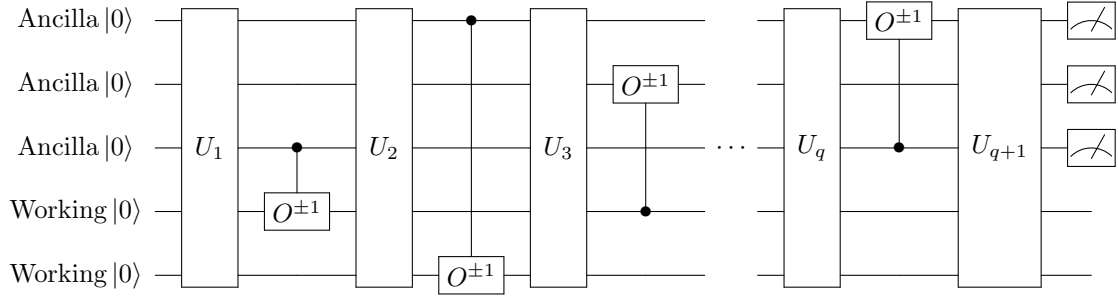

\begin{figure}
\centering
{\tikzset{%
  >={Latex[width=2mm,length=2mm]},
            base/.style = {rectangle, rounded corners, draw=black,
                          minimum width=4cm, minimum height=2.472cm,
                          text centered, font=\sffamily},
       result/.style = {triangle, fill=green!30},
       new/.style = {base, fill=blue!15},
         old/.style = {base, minimum width=2.5cm, fill=yellow!30},
}
\begin{tikzpicture}[>=latex]
\node (amplifier) [new] {Amplifier};
\coordinate (a) at (2,0.5);
\coordinate (b) at (3.5,0.5);

\coordinate (c) at (2,-0.5);
\coordinate (d) at (3.5,-0.5);

\coordinate (e) at (-2,0.5);
\coordinate (f) at (-3.5,0.5);

\coordinate (g) at (-2,-0.5);
\coordinate (h) at (-3.5,-0.5);

\draw [->,thick] (a) -- (b);
\draw [->,thick] (c) -- (d);
\draw [<-,thick] (e) -- (f);
\draw [<-,thick] (g) -- (h);

\node[] at (4.25,0.5) {$\ket{\psi}_{\text{out}}$};
\node[] at (4.25,-0.5) {$\ket{\phi}_{\text{out}}$};
\node[align=center] at (6.5,0)
{Large difference, \\ easy to distinguish};

\node[] at (-4.25,0.5) {$\ket{\psi}$};
\node[] at (-4.25,-0.5) {$\ket{\phi}$};
\node[align=center] at (-6.5,0)
{Small difference, \\ hard to distinguish};
\end{tikzpicture}}
\caption{An illustration of an amplifier. Amplifiers can be used to amplify the difference between two input states and make them easier to distinguish.}
    \label{fig:Lower_bound_amplifier}
\end{figure}
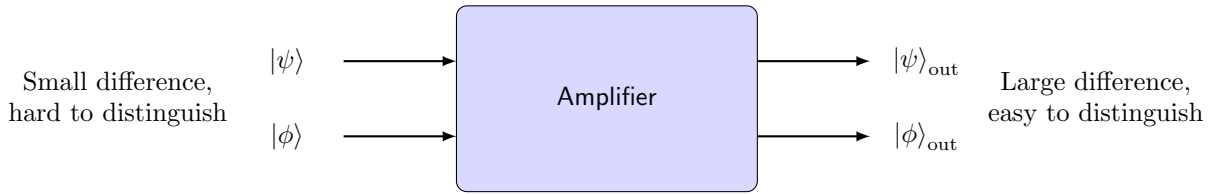

The goal of quantum state discrimination is to distinguish two quantum states $\ket{\psi}$ and $\ket{\phi}$ with large overlap. 
Standard results from the information-theoretic perspective show that, in order to distinguish two states with a constant average probability of error, the copies of required states scale at least linearly in terms of the infidelity between two states~\cite{Helstrom1969,Montanaro2008}. 

Here we study a different scenario to perform state discrimination with stronger assumptions. 
Let $O_{\psi}$ and $O_{\phi}$ be two oracles such that $O_{\psi} \ket{0} = \ket{\psi}$ and $O_{\phi} \ket{0} = \ket{\phi}$. 
Suppose that we are consistently given one of the two oracles, and we aim to determine which one is given. 
Then a general algorithm with query access can be illustrated as in~\cref{fig:circuit_amplifier} (see, e.g.,~\cite{BuhrmanDeWolf2002,HoyerSpalek2005} for a detailed review of the query model). 
Here a generic circuit is decomposed into layers and segments based on whether the operation depends on the oracles to be determined or not. 
Specifically, in~\cref{fig:circuit_amplifier}, $U_1, \cdots, U_{q+1}$ are unitaries which are independent of the choice of $\ket{\psi}$ and $\ket{\phi}$, and can be a composition of several unitary evolutions and gates. 
$O$ is the given oracle, either $O_{\psi}$ or $O_{\phi}$, and its inverse and controlled versions are also accessible. 
The algorithm alternatively applies $U_j$ and $O$, its inverse or controlled versions. 
After all the unitary evolutions, some projection measurements are performed over the ancillas to get a desired output quantum state in the working registers. 
Here $O$ can be applied to or controlled by any register, $U_j$ are arbitrary unitaries independent of the states to be determined, and the measurements are postponed without loss of generality thanks to the deferred measurement principle~\cite{NielsenChuang2000}. 
Therefore~\cref{fig:circuit_amplifier} represents a generic quantum algorithm using the oracle $O$ to output another quantum state. 

Let $\rho_u$ be the output state with $O_u$. 
For two quantum states $\ket{\psi}$ and $\ket{\phi}$ such that $|\braket{\psi|\phi}| \geq 1-\epsilon$, we call a quantum algorithm an \emph{amplifier} if the outputs $\rho_{\psi}$ and $\rho_{\phi}$ can be obtained with constant probability and have $\Omega(1)$ trace distance. 
We claim that the query complexity of an amplifier must have some lower bound, since the quantum state discrimination can be solved with the help of an amplifier as illustrated in~\cref{fig:Lower_bound_amplifier}. 
Specifically, with the given oracle $O$, we perform the amplifier to get the output state $\rho$ and perform the standard state discrimination technique over the output states. 
This is an alternative approach to distinguish $\ket{\psi}$ and $\ket{\phi}$, and therefore should be more expensive if the overlap between $\ket{\psi}$ and $\ket{\phi}$ is larger.  
Meanwhile, the cost of determining the output states $\rho_{\psi}$ and $\rho_{\phi}$ takes $O(1)$ time since they have at least $\Omega(1)$ trace distance. 
This implies that the complexity of the amplifier must be high. 

The quantitative lower bound on the query complexity of an amplifier is given as follows. 
The proof is inspired by~\cite{SommaSubasi2021} and also closely related to the lower bound of the unstructured search problem~\cite{BennettBernsteinBrassardEtAl1997}. In particular, our proof follows the idea for proving the Grover lower bound by comparing the actual algorithm with a ``fake'' algorithm that takes the same circuit structure but slightly changed oracle. However, our result cannot be directly derived from the Grover's lower bound due to the difference in the oracles. Grover's lower bound is in terms of the oracle that identifies the state to be amplified (in our case, the difference between $\ket{\phi}$ and $\ket{\psi}$), but here we would like a lower bound on the number of queries to the state preparation oracles. It is not clear how they can be directly related, so we will present a complete and self-contained proof. 
We also remark that, although our lower bound for amplifiers is related to and can be intuitively understood via the quantum state discrimination problem, the following theorem is standalone and does not use any results from the quantum state discrimination problem. 
\begin{thm}\label{thm:state_discrimination}
    Let $\ket{\psi}$ and $\ket{\phi}$ be two quantum states such that $\braket{\psi|\phi}$ is real and $\braket{\psi|\phi} \geq 1-\epsilon$. 
    Suppose that we are given an oracle $O$ and promised that it is consistently either $O_{\psi}$ or $O_{\phi}$, such that $O_{\psi}\ket{0} = \ket{\psi}$ and $O_{\phi}\ket{0} = \ket{\phi}$. 
    Let an amplifier be a quantum algorithm in~\cref{fig:circuit_amplifier} that outputs either $\rho_{\psi}$ or $\rho_{\phi}$ with $\Omega(1)$ success probability such that $\|\rho_{\psi} - \rho_{\phi}\|_{1} = \Omega(1)$. 
    Then the amplifier must use $\Omega(1/\sqrt{\epsilon})$ queries to the given oracle, its inverse or controlled versions in the worst case. 
\end{thm}
\begin{proof}
    Let the quantum circuit before measurement be $\mathcal{U}_{u} = U_{q+1} V_{q,u} U_q \cdots V_{2,u} U_2 V_{1,u} U_1$, where $U_j$'s are unitaries independent of the state $u$, and $V_{j,u}$ denotes a single application of $O_u$, its inverse or controlled version. 
    Then 
    \begin{equation}
        \begin{split}
            & \quad \left\| \mathcal{U}_{\psi}\ket{0}\bra{0}\mathcal{U}_{\psi}^{\dagger} - \mathcal{U}_{\phi}\ket{0}\bra{0}\mathcal{U}_{\phi}^{\dagger}\right\|_{1} \\
            & \leq 2\left\| \mathcal{U}_{\psi}\ket{0} - \mathcal{U}_{\phi}\ket{0} \right\| \\
            & \leq 2\left\| \mathcal{U}_{\psi} - \mathcal{U}_{\phi} \right\| \\
            & = 2 \left\| U_{q+1} V_{q,\psi} U_q \cdots V_{2,\psi} U_2 V_{1,\psi} U_1 - U_{q+1} V_{q,\phi} U_q \cdots V_{2,\phi} U_2 V_{1,\phi} U_1 \right\|  \\
            & = 2\left\| \prod_{j=1}^{q} (V_{j,\psi}U_j) - \prod_{j=1}^{q} (V_{j,\phi}U_j)\right\| \\
            & \leq 2 \sum_{j=1}^{q} \left\|V_{j,\psi}U_j - V_{j,\phi}U_j\right\| \\
            & = 2 \sum_{j=1}^{q} \left\|V_{j,\psi} - V_{j,\phi}\right\|. 
        \end{split}
    \end{equation}
    We bound $\left\|V_{j,\psi} - V_{j,\phi}\right\|$ by first analyzing $\|O_{\psi}-O_{\phi}\|$ in the worst case. 
    Here the worst case means that, besides preparing the desired state from $\ket{0}$, the preparation oracle $O$ is required to satisfy further assumptions on the orthogonal subspace. 
    Specifically, for two quantum states $\ket{\psi}$ and $\ket{\phi}$, there exist two preparation oracles $O_{\psi}$ and $O_{\phi}$ such that 
    \begin{equation}
        \|O_{\psi} - O_{\phi}\| = \|\ket{\psi} - \ket{\phi}\|. 
    \end{equation}
    For example, an explicit construction, given in~\cite[Eqs. (A26-A29)]{SommaSubasi2021}, is that $O_{\phi}$ is a rotation of $O_{\psi}$ in the two-dimensional subspace orthogonal to $\ket{\phi}$ and $\ket{\psi}$: 
    \begin{equation}
    \begin{split}
        O_{\phi} &= e^{i\theta M} O_{\psi}, \\
        \theta &= \arccos(\braket{\phi|\psi}), \\
        M &= i \ket{\psi}\bra{\psi^{\perp}} - i  \ket{\psi^{\perp}}\bra{\psi}, \\
        \ket{\psi^{\perp}} &= (I - \ket{\psi}\bra{\psi})\ket{\phi} / \|(I - \ket{\psi}\bra{\psi})\ket{\phi}\|. \\
    \end{split}
    \end{equation}
    Then, noting that $\braket{\psi|\phi}$ is real, 
    \begin{equation}
        \|O_{\psi} - O_{\phi}\| = \|\ket{\psi} - \ket{\phi}\| = \sqrt{2(1-\braket{\psi|\phi})} \leq \sqrt{2\epsilon}, 
    \end{equation}
    and 
    \begin{equation}
        \|O_{\psi}^{-1} - O_{\phi}^{-1}\| = \|O_{\phi}(O_{\psi}^{-1} - O_{\phi}^{-1})O_{\psi}\| =  \|O_{\psi} - O_{\phi}\| \leq \sqrt{2\epsilon}.
    \end{equation}
    The controlled versions of $O$ and $O^{-1}$ satisfy the same bound. 
    Therefore, in the worst case, we have 
    \begin{equation}
        \left\|V_{j,\psi} - V_{j,\phi}\right\| \leq \sqrt{2\epsilon}, 
    \end{equation}
    and thus 
    \begin{equation}
        \left\| U_{\psi}\ket{0}\bra{0}U_{\psi}^{\dagger} - U_{\phi}\ket{0}\bra{0}U_{\phi}^{\dagger}\right\|_{1} \leq 2q \sqrt{2\epsilon}. 
    \end{equation}
    Measuring $U_{\psi}\ket{0}$ or $U_{\phi}\ket{0}$ and succeeding give the states $\rho_{\psi}$ or $\rho_{\phi}$. 
    Since the success probability is $\Omega(1)$, the normalization factor due to the measurement is $\Omega(1)$, and we have 
    \begin{equation}
        \| \rho_{\psi} - \rho_{\phi} \|_{1} \leq \mathcal{O}(q\sqrt{\epsilon}). 
    \end{equation}
    Together with $\| \rho_{\psi} - \rho_{\phi} \|_{1} = \Omega(1)$ according to the definition of an amplifier, we must have 
    \begin{equation}
        q = \Omega (1/\sqrt{\epsilon}). 
    \end{equation}
\end{proof}

We remark that standard lower bounds for quantum state discrimination problem typically only assume access to multiple copies of the state, and show that the number of the copies required is $\Omega(1/\epsilon)$~\cite{Helstrom1969,Montanaro2008,LiuKoldenKroviEtAl2021}. 
Our result in~\cref{thm:state_discrimination} assumes stronger query access, \emph{i.e.}, the preparation oracle, its inverse and controlled versions, and the query lower bound becomes quadratically worse. 
Another remark is that, though in this paper we focus on the problem of solving ODEs,~\cref{thm:state_discrimination} holds for general quantum algorithms that can be viewed as amplifiers for at least one pair of quantum states and thus can potentially be applied to study other problems with certain ``non-quantumness''. 
For example,~\cite{LiuKoldenKroviEtAl2021} applies a similar idea to prove a lower bound for solving generic non-linear ODEs in terms of non-linearity. 
In~\cref{app:other_lower_bounds}, we show how to use~\cref{thm:state_discrimination} to recover existing lower bounds for solving linear systems of equations. 

\subsection{Proving lower bounds on generic quantum ODE solvers}

The central idea of our lower bound proofs is to regard an ODE solver as an amplifier. 
Non-quantum ODE with specific features may amplify the difference between two quantum states through the time evolution due to its non-unitary evolution operator $e^{AT}$, and thus satisfies the definition of an amplifier and has a worst-case query lower bound according to~\cref{thm:state_discrimination}. 

More specifically, consider the system in~\cref{eqn:ODE}. 
The solution can be explicitly given as 
\begin{equation}\label{eqn:ODE_solu}
    u(T) = e^{AT} u(0) + \int_0^T e^{A(T-s)} b ds. 
\end{equation}
Let $v(T)$ denote the solution to~\cref{eqn:ODE} with a differential initial condition $v(0)$, then 
\begin{equation}\label{eqn:ODE_solu_difference}
    u(T) - v(T) = e^{AT} (u(0)-v(0)). 
\end{equation}
\cref{eqn:ODE_solu_difference} indicates that the difference between $u(T)$ and $v(T)$ may be amplified through the operator $e^{AT}$. 
For example, let us assume that the spectral norm of $e^{AT}$ is large and $w = \text{argmax} \|e^{AT} x\|/\|x\|$. 
If we choose $u(0) - v(0)$ to be proportional to $w$, then $\|u(T)-v(T)\| = \|e^{AT}\|\|u(0)-v(0)\|$, which indicates that a quantum ODE solver may serve as an amplifier of $\ket{u(0)}$ and $\ket{v(0)}$ and the cost has a lower bound according to ~\cref{thm:state_discrimination}.

We remark that even if $\|e^{AT}\|$ is large, \cref{eqn:ODE_solu_difference} does not necessarily imply that $\ket{u(T)}$ and $\ket{v(T)}$ have large distance. 
One reason is that the difference between $u(T)$ and $v(T)$ still depends on the overlap between $u(0)-v(0)$ and the ``worst'' vector that $e^{AT}$ can amplify. 
Another important reason is that a large difference between $u(T)$ and $v(T)$ does not necessarily imply a large difference between $\ket{u(T)}$ and $\ket{v(T)}$ due to possible normalization. 
Consider a trivial example with $A = \alpha I$ and $b = 0$ for some real number $\alpha > 0$. 
The solution of the ODE becomes $u(T) = e^{\alpha T} u(0)$ and thus $\ket{u(T)} = \ket{u(0)}$. 
As a result, $\|\ket{u(T)} - \ket{v(T)}\| = \|\ket{u(0)} - \ket{v(0)}\|$, and the distance between quantum states is preserved, although $\|e^{\alpha T}\|$ or even $\|e^{\alpha T}(u(0)-v(0))\|$ may be very large\footnote{This is also related to the shifting equivalence for homogeneous ODEs. See~\cref{sec:LB_implications}.}. 
Therefore more careful computations are needed to identify the intrinsic sources of any computational overhead and prove the lower bounds.

We will show technical statements and proofs of our lower bounds in the next subsection. We prove our results by constructing a specific pair of quantum states that can be amplified by the ODE and using~\cref{thm:state_discrimination}. We remark that the examples we will construct are the worst-case examples for generic quantum ODE algorithms, which suggests that certain resources are required if we use a generic quantum ODE algorithm to prepare those states. This does not mean that they are hard for any quantum algorithm other than generic ODE solvers. Indeed, as our construction exhibits specific structures, efficient tailored quantum algorithms can be used to achieve fast-forwarding.

\subsection{Technical results}

Here we focus on homogeneous ODEs, \emph{i.e.}, $b = 0$ in~\cref{eqn:ODE}. 
We remark that all of the results can be generalized to the inhomogeneous case where $b \neq 0$, as stated in details in~\cref{app:proofs_LB_inhomo}. 
Typical examples of homogeneous ODEs include quantum dynamics, where $A$ is an anti-Hermitian matrix, and imaginary time evolution, where $A$ is a Hermitian matrix. 
For technical simplicity, we consider a diagonalizable matrix $A$, \emph{i.e.}, there exist an invertible matrix $V$ and a diagonal matrix $D = \text{diag}\left(\lambda_1,\cdots,\lambda_N\right)$, such that $A = VDV^{-1}$.  
Note that here we do not assume $V$ to be unitary. 
We will first show that, once there exists a pair of eigenvalues with different real parts, then the scaling of generic quantum algorithms for this differential equation must be at least exponential in evolution time and real part gap. 
Our second result is that generic quantum ODE solvers will become less efficient if the coefficient matrix $A$ is more ``non-normal''. 
These two hardness results together imply that ``quantum algorithms for quantum dynamics'' is the most efficient. 
Finally, we apply our results to quantum algorithms for imaginary time evolution and show that the cost scales at least linearly in the decay of the solution norm.  

\subsubsection{Eigenvalues with different real parts}

\begin{prop}\label{prop:lb_eig_diff_homo}
    Consider the homogeneous ODE with a diagonalizable matrix $A = VDV^{-1}$ where $D = \text{diag}\left(\lambda_1,\cdots,\lambda_N\right)$ is a diagonal matrix and $V$ is an invertible matrix. 
    Suppose that $b = 0$. 
    Then there is no generic quantum algorithm which can prepare $u(T)/\|u(T)\|$ with bounded error and failure probability, using $o\left(e^{T\max_{i\neq j} |\text{Re}(\lambda_i) - \text{Re}(\lambda_j)|}\right)$ queries to the preparation oracle of $\ket{u(0)}$, its inverse and controlled versions. 
\end{prop}

\begin{proof}
    Let $\lambda_j = \alpha_j + i \beta_j$ where $\alpha_j$ and $\beta_j$ are the real part and the imaginary part of $\lambda_j$, respectively. 
    Without loss of generality, assume that 
    \begin{equation}
        \max_{i\neq j} |\text{Re}(\lambda_i-\lambda_j)| = \text{Re}(\lambda_1-\lambda_2) = \alpha_1-\alpha_2 > 0. 
    \end{equation}
    Let $V = (v_1,\cdots,v_N)$ where each $v_j$ is the eigenvector of $A$ corresponding to the eigenvalue $\lambda_j$ and are normalized such that $\|v_j\| = 1$. 
    Notice that $V$ is not necessarily unitary so $\braket{v_i|v_j}$'s are not necessarily $0$, but we always have 
    \begin{equation}
        0 \leq |\braket{v_i|v_j}| < 1. 
    \end{equation}
    
    We consider solving~\cref{eqn:ODE} with two possible initial conditions 
    \begin{equation}
        \begin{split}
            u(0) &= e^{i\theta} v_2,\\
            w(0) &= \sqrt{\epsilon} v_1 + \xi v_2. 
        \end{split}
    \end{equation}
    Here $\theta$ is chosen such that $\braket{u(0)|w(0)}$ is real, and $\xi$ is chosen such that $\|w(0)\| = 1$, \emph{i.e.}, 
    \begin{equation}\label{eqn:nor_cond_ini_eig_diff}
        1 = w(0)^* w(0) = \epsilon + |\xi|^2 + 2\text{Re}(\sqrt{\epsilon} \xi \braket{v_1|v_2}). 
    \end{equation}
    Let the evolution time be 
    \begin{equation}\label{eqn:choice_T_eig_diff}
        T = \frac{1}{2(\alpha_1-\alpha_2)} \log(1/\epsilon). 
    \end{equation}
    The solutions $u(T)$ and $w(T)$ can be solved as follows. 
    We start with the equation $V^{-1} V = I$, which implies that 
    \begin{equation}
        V^{-1} v_j = e_j, 
    \end{equation}
    where $e_j$ is the vector with the only non-zero entry to be $1$ at the $j$-th position. 
    As a result, 
    \begin{equation}
        e^{AT} v_j = V e^{DT} V^{-1} v_j = V e^{DT} e_j = e^{\lambda_j T} v_j.  
    \end{equation}
    Therefore, 
    \begin{equation}
        u(T) = e^{i\theta}e^{\lambda_2 T} v_2 = e^{i\theta} e^{\alpha_2 T}e^{i\beta_2 T} v_2,
    \end{equation}
    and 
    \begin{equation}
        w(T) = \sqrt{\epsilon} e^{\lambda_1 T} v_1 + \xi e^{\lambda_2 T} v_2 = \sqrt{\epsilon} e^{\alpha_1 T}e^{i\beta_1 T} v_1 + \xi e^{\alpha_2 T}e^{i\beta_2 T} v_2. 
    \end{equation}
    
    We now compute the fidelity of the input states and the output states. 
    For the input states, we have 
    \begin{equation}
    \begin{split}
        |\braket{u(0)|w(0)}| &= \sqrt{\braket{(\sqrt{\epsilon}v_1+\bar{\xi} v_2)|v_2}\braket{v_2|(\sqrt{\epsilon}v_1+\xi v_2)}} \\
        & = \sqrt{(\sqrt{\epsilon}\braket{v_1|v_2} + \bar{\xi})(\sqrt{\epsilon}\braket{v_2|v_1} + \xi)} \\
        & = \sqrt{\epsilon |\braket{v_1|v_2}|^2 + |\xi|^2 + 2\text{Re}(\sqrt{\epsilon} \xi \braket{v_1|v_2})}. 
    \end{split}
    \end{equation}
    Using~\cref{eqn:nor_cond_ini_eig_diff}, we have 
    \begin{equation}\label{eqn:ini_fed_eig_diff}
        \begin{split}
            |\braket{u(0)|w(0)}| &= \sqrt{ \epsilon |\braket{v_1|v_2}|^2 + 1 - \epsilon } \\
            & \geq \sqrt{1-\epsilon} \\
            & > 1-\epsilon. 
        \end{split}
    \end{equation}
    For the output states, we have 
    \begin{equation}
        \begin{split}
            |\braket{u(T)|w(T)}|^2 &= \frac{  w(T)^{*} u(T)  u(T) ^{* }w(T) }{\|u(T)\|^2 \|w(T)\|^2} \\
            &= \frac{ (\sqrt{\epsilon} e^{\alpha_1 T}e^{i\beta_1 T} v_1 + \xi e^{\alpha_2 T}e^{i\beta_2 T} v_2)^{*}v_2 v_2^{*}(\sqrt{\epsilon} e^{\alpha_1 T}e^{i\beta_1 T} v_1 + \xi e^{\alpha_2 T}e^{i\beta_2 T} v_2) }{\|\sqrt{\epsilon} e^{\alpha_1 T}e^{i\beta_1 T} v_1 + \xi e^{\alpha_2 T}e^{i\beta_2 T} v_2\|^2} \\
            & = \frac{ (\sqrt{\epsilon} e^{(\alpha_1-\alpha_2) T}e^{i(\beta_1-\beta_2) T} v_1 + \xi v_2)^{*}v_2 v_2^{*} (\sqrt{\epsilon} e^{(\alpha_1-\alpha_2) T}e^{i(\beta_1-\beta_2) T} v_1 + \xi v_2) }{\|\sqrt{\epsilon} e^{(\alpha_1-\alpha_2) T}e^{i(\beta_1-\beta_2) T} v_1 + \xi v_2\|^2}.
        \end{split}
    \end{equation}
    The choice of $T$ in~\cref{eqn:choice_T_eig_diff} ensures that $\sqrt{\epsilon} e^{(\alpha_1-\alpha_2)T} = 1$, and thus 
    \begin{equation}
        \begin{split}
            |\braket{u(T)|w(T)}|^2 &= \frac{ (e^{i(\beta_1-\beta_2) T} v_1 + \xi v_2)^{*} v_2 v_2^{*}(e^{i(\beta_1-\beta_2) T} v_1 + \xi v_2) }{\|e^{i(\beta_1-\beta_2) T} v_1 + \xi v_2\|^2} \\
            & = \frac{|\braket{v_1|v_2}|^2 + |\xi|^2 + 2\text{Re}(e^{-i(\beta_1-\beta_2)T}\xi \braket{v_1|v_2})}{1 + |\xi|^2 + 2\text{Re}(e^{-i(\beta_1-\beta_2)T}\xi \braket{v_1|v_2})} .
        \end{split}
    \end{equation}
    Notice that 
    \begin{equation}
        |2\text{Re}(e^{-i(\beta_1-\beta_2)T}\xi \braket{v_1|v_2})| \leq 2 |e^{-i(\beta_1-\beta_2)T}\xi \braket{v_1|v_2}| \leq |\xi|^2 + |\braket{v_1|v_2}|^2, 
    \end{equation}
    and that the function $f(x) = \frac{a+x}{b+x}$ with $b > a > 0$ is monotonically increasing for $x \geq -a$, then we can further bound the final fidelity as 
    \begin{equation}
        |\braket{u(T)|w(T)}|^2 \leq \frac{2|\braket{v_1|v_2}|^2 + 2|\xi|^2 }{1 + |\braket{v_1|v_2}|^2 + 2|\xi|^2 }. 
    \end{equation}
    Using~\cref{eqn:nor_cond_ini_eig_diff}, we can bound $|\xi|^2$ as 
    \begin{equation}
        \begin{split}
            |\xi|^2 & = 1 - \epsilon - 2\text{Re}(\sqrt{\epsilon}\xi \braket{v_1|v_2}) \\
            & \leq 1 + 2|\sqrt{\epsilon} \xi \braket{v_1|v_2}| \\
            & \leq 1 + 2|\xi|, 
        \end{split}
    \end{equation}
    which implies that 
    \begin{equation}\label{eqn:bound_xi}
        |\xi| \leq 1 + \sqrt{2}. 
    \end{equation}
    Therefore, 
    \begin{equation}\label{eqn:final_fed_eig_diff}
        |\braket{u(T)|w(T)}| \leq \sqrt{\frac{2|\braket{v_1|v_2}|^2 + 2(3+2\sqrt{2}) }{1 + |\braket{v_1|v_2}|^2 + 2(3+2\sqrt{2}) }} \eqqcolon C. 
    \end{equation}
    Here $C < 1$ and only depends on $V$, as $|\braket{v_1|v_2}|$ will always be strictly smaller than $1$ due to the non-singularity of $A$. We consider $0 < \epsilon < 1-C$ and thus can regard $C$ as a constant and will absorb it into the notation $\mathcal{O}$ and $\Omega$. 

    To establish the non-existence of an efficient generic quantum algorithm as claimed, we use proof by contradiction, assuming the opposite that there exists a generic quantum algorithm with cost $o(e^{T(\alpha_1-\alpha_2)})$. Specifically, 
    given an oracle to consistently prepare either $\ket{u(0)}$ or $\ket{w(0)}$, we denote $\ket{\widetilde{u}(T)}$ and $\ket{\widetilde{w}(T)}$ as the corresponding outputs of a quantum ODE solver with $2$-norm distance at most $(1-C)/4$ of $\ket{u(T)}$ and $\ket{w(T)}$, respectively. 
    $\ket{\widetilde{u}(T)}$ and $\ket{\widetilde{w}(T)}$ can be obtained using queries to the state preparation oracles for a number of $o(e^{T(\alpha_1-\alpha_2)}) = o(1/\sqrt{\epsilon})$. 
    According to~\cref{lem:trace_distance},~\cref{lem:fed_2norm},~\cref{eqn:ini_fed_eig_diff} and~\cref{eqn:final_fed_eig_diff}, we have 
    \begin{equation}\label{eqn:final_dis_eig_diff}
    \begin{split}
        \| \ket{\widetilde{u}(T)}\bra{\widetilde{u}(T)} - \ket{\widetilde{w}(T)}\bra{\widetilde{w}(T)} \|_1 &= 
        2 \sqrt{1 - |\braket{\widetilde{u}(T)|\widetilde{w}(T)}|^2 } \\
        & \geq 2 \sqrt{1 - ( C + (1-C)/4 + (1-C)/4 )^2} \\
        & = \sqrt{(3+C)(1-C)} = \Omega(1). 
    \end{split}
    \end{equation}
    Therefore, \cref{eqn:ini_fed_eig_diff} and~\cref{eqn:final_dis_eig_diff} imply that an amplifier of $\ket{u(0)}$ and $\ket{w(0)}$ can be implemented using $o(1/\sqrt{\epsilon})$ queries to the state preparation oracle. 
    This contradicts with~\cref{thm:state_discrimination}, and thus completes the proof. 
\end{proof}

We remark that~\cref{prop:lb_eig_diff_homo} only requires the input state to be worst-case. 
More specifically, for any given diagonalizable coefficient matrix $A$, we can always find an input state such that the complexity is bounded from below as claimed. 
Meanwhile, our lower bound is only in terms of the preparation oracle of $u(0)$, and no result is obtained in the input oracle of $A$. 
This is different from the counterparts in the Hamiltonian simulation problem (\emph{e.g.},~\cite{BerryChildsKothari2015}), which require both $A$ and $u(0)$ to be worst-case but can yield lower bounds in terms of the input oracles of both $A$ and $u(0)$. 

\subsubsection{Non-normal matrices}

It has been demonstrated that quantum algorithms can efficiently solve Hamiltonian simulation problems. 
With the linear overhead of solution decay, quantum devices can also efficiently solve imaginary time evolution problems. 
In both cases the coefficient matrix $A$ is \emph{normal}, \emph{i.e.}, $A^{\dagger}A = AA^{\dagger}$ or, equivalently, $A$ can be diagonalized by a unitary matrix.

In this section, we wish to understand the capability of quantum algorithms for solving ODEs with a coefficient matrix beyond a normal matrix. 
We will show that generic quantum ODE solvers will become less efficient if the coefficient matrix is more ``non-normal''. 
To establish an explicit hardness result, we need to quantify the non-normality of a matrix. 
In this work, we use the following function 
\begin{equation}\label{eqn:def_mu_A}
    \mu(A) = \|A^{\dagger}A-AA^{\dagger}\|^{1/2}. 
\end{equation}
This is a quite natural measure because of $\mu(A) = 0$ for any normal matrix according to the definition. 
We refer to~\cite{ElsnerPaardekooper1987} for alternative ways to measure the non-normality of a matrix.

We are now ready to prove the following result, which shows that generic quantum ODE solvers must take $\Omega(\mu(A))$ cost in the worst case. 
\begin{prop}\label{prop:lb_non_normal_homo}
    Consider the homogeneous ODE problem in~\cref{eqn:ODE} with $b = 0$. 
    Let $\mu(A) = \|A^{\dagger}A-AA^{\dagger}\|^{1/2}$. 
    Then, there is no generic quantum algorithm which can prepare $u(T)/\|u(T)\|$ with bounded error and failure probability, using $o\left(\mu(A)\right)$ queries to the preparation oracle of $\ket{u(0)}$, its inverse and controlled versions. 
\end{prop}

\begin{proof}
    Consider the example with $N = 3$, $u = (u_1,u_2,u_3)^T$, $b = 0$, and 
     \begin{equation}\label{eqn:proof_lb_non_normal_A}
         A = \left(\begin{array}{ccc}
             i & i/\delta & 0 \\
             0 & 2i & 0 \\
             0 & 0 & 3i 
         \end{array}\right). 
     \end{equation}
     Here $\delta$ is a real positive parameter in $(0,1)$. 
     Notice that the matrix $A$ can be diagonalized such that $A = VDV^{-1}$ where 
     \begin{equation}
         V = \left(\begin{array}{ccc}
             1 & 1 & 0 \\
             0 & \delta & 0 \\
             0 & 0 & 1 
         \end{array}\right), 
         \quad D = \left(\begin{array}{ccc}
             i & 0 & 0 \\
             0 & 2i & 0 \\
             0 & 0 & 3i 
         \end{array}\right),
     \end{equation}
     and that 
     \begin{equation}
         A^{\dagger}A - AA^{\dagger} = \left(\begin{array}{ccc}
             -1/\delta^2 & -1/\delta & 0 \\
             -1/\delta & 1/\delta^2 & 0 \\
             0 & 0 & 0 
         \end{array}\right),
     \end{equation}
     \begin{equation}\label{eqn:mu_non_normal_homo}
        \mu(A) = \frac{\sqrt[4]{1+\delta^2}}{\delta} = \Theta\left(\frac{1}{\delta}\right). 
     \end{equation}
     We choose two initial conditions 
     \begin{equation}
         \begin{split}
             u(0) &= (0,0,1)^{T}, \\
             w(0) &= (0,\delta,\sqrt{1-\delta^2})^{T}. 
         \end{split}
     \end{equation}
     According to~\cref{eqn:ODE_solu} and noting that $e^{A}$ can be computed as $Ve^{D}V^{-1}$, we obtain up to $T = 1$
     \begin{equation}
         \begin{split}
             u(1) &= (0,0,e^{3i})^{T}, \\
             w(1) &= (e^{2i}-e^{i}, e^{2i}\delta, e^{3i}\sqrt{1-\delta^2})^{T}. 
         \end{split}
     \end{equation}

     Now, suppose the opposite that there exists an efficient generic quantum algorithm that can solve the general ODE with cost $o(\mu(A))$. 
     Suppose that we are given an oracle to prepare either $\ket{u(0)}$ or $\ket{w(0)}$. 
     Let $\ket{\widetilde{u}(1)}$ and $\ket{\widetilde{w}(1)}$ be the corresponding outputs of the quantum ODE solver with $2$-norm distance at most $1/10$ from $\ket{u(1)}$ and $\ket{w(1)}$, respectively. 
     In particular, $\ket{\widetilde{u}(1)}$ and $\ket{\widetilde{w}(1)}$ can be prepared using  $o(\mu(A)) = o(1/\delta)$ queries to state preparation oracles.
     On the one hand, 
     \begin{equation}\label{eqn:overlap_non_normal_homo}
         \braket{u(0)|w(0)} = \sqrt{1-\delta^2} \geq 1-\delta^2. 
     \end{equation}
     On the other hand, 
     \begin{equation}
         \begin{split}
             |\braket{u(1)|w(1)}| & = \frac{|e^{-3i}e^{3i}\sqrt{1-\delta^2}|}{\|w(1)\|} \\
             & = \frac{\sqrt{1-\delta^2}}{\sqrt{|e^{2i}-e^i|^2 + \delta^2 + 1-\delta^2}} \\
             & \leq \frac{1}{\sqrt{|e^{2i}-e^i|^2 + 1}}. 
         \end{split}
     \end{equation}
     According to~\cref{lem:trace_distance} and~\cref{lem:fed_2norm}, we have 
     \begin{equation}
         \begin{split}
             \| \ket{\widetilde{u}(1)}\bra{\widetilde{u}(1)} - \ket{\widetilde{w}(1)}\bra{\widetilde{w}(1)}\|_1 
             & = 2\sqrt{1 - |\braket{\widetilde{u}(1)|\widetilde{w}(1)}|^2} \\
             & \geq 2 \sqrt{1 - \left( \frac{1}{\sqrt{|e^{2i}-e^i|^2 + 1}} + 1/10 + 1/10 \right)^2} \\
             & > 0.77 = \Omega(1). 
         \end{split}
     \end{equation}
     Therefore, \cref{eqn:overlap_non_normal_homo} and~\cref{eqn:mu_non_normal_homo} imply that an amplifier of $\ket{u(0)}$ and $\ket{w(0)}$ can be implemented using $o(1/\delta)$ queries to the state preparation oracle. 
     This contradicts with~\cref{thm:state_discrimination}, and thus completes the proof. 
\end{proof}

We remark that, although we use $\mu(A)$, defined in~\cref{eqn:def_mu_A}, as the measure of non-normality in~\cref{prop:lb_non_normal_homo} and throughout this work, there are a variety of measures of non-normality~\cite{ElsnerPaardekooper1987} and our lower bound can be straightforwardly reformulated as lower bounds in terms of other measures by using the inequalities among different measures established in~\cite{ElsnerPaardekooper1987}. 

Furthermore, to provide more intuition on the role of non-normality, we state another lower bound for ODEs with diagonalizable coefficient matrix $A$. 
We show that generic quantum algorithms for ODEs with diagonalizable matrices must take $\Omega(\kappa_V)$ cost in the worst case, where $A = VDV^{-1}$ is the eigen-decomposition of $A$ and $\kappa_V = \|V\|\|V^{-1}\|$ is the condition number of $V$. 
Note that $\kappa_V$ appears in the upper bound in~\cref{lem:DEsolver_TD} and also serves as a measure of non-normality. 
This is because $\kappa_V$ attains its minimum $1$ if and only if $V$ is a unitary matrix, and thus if and only if $A$ is a normal matrix. 

\begin{prop}\label{prop:lb_non_normal_homo_2}
    Consider the homogeneous ODE problem in~\cref{eqn:ODE} with diagonalizable $A$ and $b = 0$. 
    Let $A = VDV^{-1}$ be the eigen-decomposition of $A$ and $\kappa_V = \|V\|\|V^{-1}\|$ is the condition number of $V$. 
    Then, there is no generic quantum algorithm which can prepare $u(T)/\|u(T)\|$ with bounded error and failure probability, using $o\left(\kappa_V\right)$ queries to the preparation oracle of $\ket{u(0)}$, its inverse and controlled versions. 
\end{prop}

\begin{proof}
    This proposition can be proved following the same proof of~\cref{prop:lb_non_normal_homo}, noticing that the worst-case matrix $A$ constructed in~\cref{eqn:proof_lb_non_normal_A} satisfies $\kappa_V = \Theta(\mu(A)) = \Theta(1/\delta)$. 
\end{proof}

\subsubsection{Implications and shifting equivalence}\label{sec:LB_implications}

Besides identifying two types of ``non-quantumness'' of ODEs,  \cref{prop:lb_eig_diff_homo} and~\cref{prop:lb_non_normal_homo} together also imply that the class of ODEs that generic quantum algorithms can solve most efficiently is equivalent to quantum dynamics. 

To illustrate this, let us consider a linear homogeneous ODE $du/dt = Au$ where $A$ is a general matrix. 
According to~\cref{prop:lb_non_normal_homo}, there exist worst-case examples such that the non-normality of the matrix $A$ introduces extra computational cost, and thus a sufficient condition to avoid such overhead is to restrict the matrix $A$ to be normal, \emph{i.e.}, unitarily diagonalizable. 
According to~\cref{prop:lb_eig_diff_homo}, to avoid any possible exponential overhead in time and real part gap, all the eigenvalues of $A$ must have equal real parts. 
These two observations suggest that the structure of the matrix $A$ is 
\begin{equation}\label{eqn:nice_A_homo}
    A = U(aI + J)U^{\dagger},
\end{equation}
where $U$ is a unitary matrix, $a$ is a real number, $I$ is the identity matrix and $J$ is a diagonal matrix with purely imaginary diagonal elements. 

We claim that the ODE with coefficient matrix~\cref{eqn:nice_A_homo} is equivalent to quantum dynamics. 
Let 
\begin{equation}
    H = iUJU^{\dagger}. 
\end{equation}
Here $H$ is a Hermitian matrix because (noticing that $J^{\dagger} = -J$)
\begin{equation}
    H^{\dagger} = -i U J^{\dagger} U^{\dagger} = i UJU^{\dagger} = H. 
\end{equation}
Then the matrix $A$ becomes 
\begin{equation}
    A = U(aI+J)U^{\dagger} = aI - iH. 
\end{equation}
The solution of the ODE can be written as 
\begin{equation}
    u(t) = e^{At} u(0) = e^{(aI-iH) t} u(0) = e^{at}e^{-iHt}u(0), 
\end{equation}
and thus the normalized solution $\ket{u(t)} = e^{-iHt}\ket{u(0)}$ exactly solves the quantum dynamics with Hamiltonian $H$. 

The above calculations can be generally summarized as the shifting equivalence: if two matrices $A$ and $B$ satisfy $A = \alpha I + B$ for some real number $\alpha$, then the ODEs with coefficient matrices $A$ and $B$, respectively, are quantumly equivalent in the sense that their normalized solutions are the same. 
We remark that the quantities we use to describe ``non-quantumness'' in our lower bounds~\cref{prop:lb_eig_diff_homo} and~\cref{prop:lb_non_normal_homo}, namely the real part gap $\delta$ and the measure of non-normality $\mu(A)$, are both invariant under shifting by a real number. 

\subsubsection{Norm decay and imaginary time evolution}

Now we restrict ourselves to the imaginary time evolution
\begin{equation}\label{eqn:ODE_ITE}
    \frac{d}{dt} u(t) = -H u(t), \quad t \in [0,T], 
\end{equation}
where $H$ is a positive semi-definite Hermitian matrix. 
The solution of~\cref{eqn:ODE_ITE} is given as 
\begin{equation}
    u(T) = e^{-HT}u(0), 
\end{equation}
and the imaginary time evolution problem typically aims at preparing an approximation of the quantum state $e^{-HT}\ket{u(0)}/\xi$ where $\xi = \|e^{-HT}\ket{u(0)}\|$. 

Existing approaches for imaginary time evolution, including those based on quantum phase estimation (QPE)~\cite{poulin2009sampling} and LCU~\cite{vanApeldoorn2020quantum}, scale linearly in terms of $1/\xi$, which indicates that the quantum algorithms might be expensive if the unnormalized solution $u(T)$ decays fast during the evolution. 
Similar dependence on the solution decay is also observed in quantum solvers for more general inhomogeneous linear ODE, for example the linear-in-$g$ dependence shown in~\cref{lem:DEsolver_TD} and~\cref{lem:DEsolver_TI}. 
It is noted in~\cite{BerryChildsOstranderEtAl2017} that dramatically improved dependence on the norm decay (\emph{e.g.}, exponentially improved) is highly unlikely since such improvement implies $\text{BQP}=\text{PostBQP}$, which, because $\text{PostBQP} = \text{PP}$ by \cite{Aaronson2004}, then implies $\text{BQP} = \text{PP}$. The equality $\text{BQP} = \text{PP}$ is considered highly unlikely not least because \text{PP} easily contains \text{NP} but $\text{BQP}$ is not even believed to contain $\text{NP}$ \cite{bbbv97}. In comparison, we establish a ``stronger'' \emph{unconditional}  hardness result but in the ``weaker'' \emph{query} model, showing that the linear dependence on the norm decay is asymptotically tight and further ruling out any possible polynomial improvement in the next corollary, which is a direct consequence of~\cref{prop:lb_eig_diff_homo}. 

\begin{cor}
    Consider the imaginary time evolution problem~\cref{eqn:ODE_ITE} where $H$ is a positive semi-definite Hermitian matrix, and let $\xi = \|e^{-HT}\ket{u(0)}\|$. 
    Then, there is no generic quantum algorithm that can prepare $u(T)/\|u(T)\|$ with bounded error and failure probability, using $o(1/\xi)$ queries to the preparation oracle of $\ket{u(0)}$, its inverse or controlled versions. 
\end{cor}
\begin{proof}
    This is a direct consequence of~\cref{prop:lb_eig_diff_homo}, by considering a positive semi-definite Hermitian matrix $H$ with at least one eigenvalue equal to $0$ and noticing that 
    \begin{equation}
        e^{T\max_{i\neq j} |\text{Re}(\lambda_i-\lambda_j)|} = e^{\|H\|T} = \left\|e^{-HT} \ket{u(0)}\right\|^{-1} = 1/\xi,
    \end{equation}
    where $\ket{u(0)}$ is chosen to be an eigenstate of $H$ with eigenvalue $\|H\|$.
\end{proof}

\subsection{Lower bound in time and precision}

We now discuss lower bounds in time $T$ and tolerated error $\epsilon$ for solving ODE. 
It is known that the lower bound for Hamiltonian simulation is $\Omega(T + \log(1/\epsilon)/\log\log(1/\epsilon))$~\cite{BerryChildsKothari2015}. 
Here we establish a different lower bound for inhomogeneous ODE beyond quantum dynamics. 
The difference between our lower bound and the existing one for Hamiltonian simulation is that our result considers the coefficient matrix to have a negative logarithmic norm, while the coefficient matrix in Hamiltonian simulation is always anti-Hermitian (so the logarithmic norm is $0$). 
Here the logarithmic norm $l(A)$ of $A$ is defined as~\cite{Soderlind2006} 
\begin{equation}
    l(A) = \lim_{h \rightarrow 0+} \frac{\|I+hA\|-1}{h}.
\end{equation}

The lower bound can be established by linking a quantum ODE solver with solving a linear system of equations and using the lower bound for quantum linear system solvers. 
The key observation is that the solution of a dissipative system of differential equations exponentially converges to its equilibrium, which is the solution of a system of linear equations. 
We prove the following proposition, which implies a lower bound $\Omega(T^{\alpha})$ for sufficiently accurate simulation and a lower bound $\Omega(\log^{\alpha}(1/\epsilon))$ for sufficiently long time simulation.

\begin{prop}\label{prop:lb_ND_inhomo}
    Consider the linear ODE problem in~\cref{eqn:ODE}. 
    Let $l(A)$ denote the logarithmic norm of $A$. 
    Suppose that the worst-case quantum query lower bound for solving a linear system of equations $Ax = b$ with $l(A) < 0$ up to error $\epsilon$ is $\Omega(\log^{\alpha}(1/\epsilon))$, then any quantum algorithm approximating $u(T)/\|u(T)\|$ up to error $\epsilon$ must have worst-case query complexity $\Omega(\min\left\{ \Omega\left(\log^{\alpha}(1/\epsilon)\right), \Omega\left(T^{\alpha}\right) \right\})$ to the same oracles. 
\end{prop}

\begin{proof}
    The definition of the logarithmic norm implies that it bounds the decay rate of the exponential operator in the sense that for any $T > 0$,~\cite{Soderlind2006}
    \begin{equation}\label{eqn:bdd_log_norm_exp}
        \|e^{AT}\| \leq e^{l(A) T}. 
    \end{equation}
    For now, we do not impose further structure on the coefficient matrix $A$. 
$b$ and $u(0)$ can also be arbitrary vectors with no constraint other than that $\|b\| = \|u(0)\| = 1$ for technical simplicity. 
Using~\cref{lem:succ_prob_error}, the solution of~\cref{eqn:ODE} is an approximation of the solution of the linear system $(-A)x = b$ as $T$ becomes larger. This is because the dynamics of~\cref{eqn:ODE} converges to its stable solution, which turns out to be $-A^{-1}b$. 
Alternatively, we may look at the variation of constants formula~\cref{eqn:ODE_solu}. 
The homogeneous part $e^{AT}u(0)$ will vanish using~\cref{eqn:bdd_log_norm_exp} and let $T \rightarrow \infty$. The inhomogeneous part can be written as $\int_0^T e^{A(T-s)} b ds = (e^{AT} - I)A^{-1}b$ so will converge to $-A^{-1}b$. Since the global phase is not important in quantum state, we can view the dynamics as the approximation of the solution to the quantum linear system problem $A\ket{x} \sim \ket{b}$. A more technical estimate can be carried out as follows
\begin{equation}
\begin{split}
     \left\| \ket{u(T)} - \ket{-A^{-1} b} \right\| & \leq \frac{2}{\|A^{-1}b\|} \left\| u(T) + A^{-1}b \right\| \\
     & = \frac{2}{\|A^{-1}b\|} \left\| e^{AT}u(0) + (e^{AT}-I)A^{-1}b + A^{-1}b \right\| \\
     & \leq \frac{2}{\|A^{-1}b\|} \left(\|e^{AT}\| + \|e^{AT}\|\|A^{-1}\|\right) \\
     & \leq \frac{2}{\|A^{-1}b\|} \left(1+\|A^{-1}\| \right) e^{l(A)T} \\
     & \leq 2\left( \|A\| + \kappa(A) \right) e^{-|l(A)|T}, 
\end{split}
\end{equation}
where the last inequality is because $1 = \|b\| = \|AA^{-1}b\| \leq \|A\|\|A^{-1}b\|$, and $\kappa(A) = \|A\|\|A^{-1}\|$ be the condition number of $A$. 

Let $\ket{\widetilde{u}(T)}$ be the numerical state obtained by a generic quantum ODE solver up to time $T$ with tolerated error $\epsilon$.  
Then we have 
\begin{equation}
    \left\| \ket{\widetilde{u}(T)} - \ket{A^{-1} b} \right\| \leq \epsilon + 2\left( \|A\| + \kappa(A) \right) e^{-|l(A)|T}. 
\end{equation}
In other words, a quantum ODE solver up to time $T$ and precision $\epsilon$ can be regarded as a quantum linear system solver with precision specified in the above estimate. 

Suppose that a generic quantum linear system solver for matrices with negative logarithmic norm has worst-case query complexity $\Omega((\log(1/\epsilon'))^{\alpha})$ where $\epsilon'$ is the tolerated error. 
Then a generic quantum ODE solver for solving ODE with the corresponding worst-case $A$ and $b$ has query complexity 
\begin{equation}
    \begin{split}
        \Omega\left( \log^{\alpha}\left(\frac{1}{\epsilon + 2\left( \|A\| + \kappa(A) \right) e^{-|l(A)|T}}\right) \right) 
        & \geq \Omega\left( \log^{\alpha}\left(\frac{1}{2\max\left\{\epsilon, 2\left( \|A\| + \kappa(A) \right) e^{-|l(A)|T} \right\}}\right) \right) \\
        & = \Omega\left( \log^{\alpha}\left(\frac{1}{2}\min\left\{ \frac{1}{\epsilon}, \frac{1}{2\left( \|A\| + \kappa(A) \right) e^{-|l(A)|T}}\right\} \right) \right) \\
        & = \Omega\left(\min\left\{ \log^{\alpha}\left(\frac{1}{2\epsilon} \right), \log^{\alpha}\left( \frac{e^{|l(A)|T}}{4\left( \|A\| + \kappa(A) \right) } \right) \right\}\right) \\
        & = \min\left\{ \Omega\left(\log^{\alpha}(1/\epsilon)\right), \Omega\left(T^{\alpha}\right) \right\}. 
    \end{split}
\end{equation}
\end{proof}

\section{Fast-forwarding}\label{sec:fast_forwarding}

In this section, we investigate several specific classes of ODEs and design quantum algorithms which scale better than existing generic ones (\cref{lem:DEsolver_TD} and~\cref{lem:DEsolver_TI}). 
We obtain quadratic improvement in $T$ if the coefficient matrix is bounded and negative definite or negative semi-definite with square-root access, and exponential improvement in $T$ and $\|A\|$ if the eigenvalues and eigenstates of $A$ are known. 
As we obtain greater speedup in the second scenario, here in the main text we only focus on the case where $A$ has known eigenvalues and eigenstates. 
Fast-forwarding results for negative definite or semi-definite matrices are presented in details in~\cref{app:FF_ND_NSD}. 
The exponential improvement can also be obtained in solving common linear evolutionary PDEs.

\subsection{Result for negative semi-definite coefficient matrix with known eigenvalues and eigenstates}\label{sec:FF_NSD}

\subsubsection{Oracles and notation}\label{sec:oracles_NSD}

Let $A$ be a negative semi-definite matrix. 
Specifically, $A$ can be represented as 
\begin{equation}
    A = U\Lambda U^{\dagger}.
\end{equation}
Here $U$ is a known unitary in the sense that a quantum circuit for efficiently implementing $U$ is given.  
With an abuse of notation, we still use $U$ to denote this circuit. 
The matrix $\Lambda = \text{diag}(\lambda_0,\cdots,\lambda_{N-1})$ is a diagonal matrix such that $\lambda_j \leq 0$ for all $j$. 

Suppose that the time $T$ and the eigenvalues $\lambda_j$ are binarily encoded through the oracles $O_T: \ket{0}_t \rightarrow \ket{T}_t$ and $O_\Lambda: \ket{0}_e\ket{j}_v \rightarrow \ket{\lambda_j}_e\ket{j}_v$. 
Let the function $e^{\lambda t}$ be given through the oracle $O_{\exp}: \ket{0}_f\ket{\lambda}_e\ket{t}_t \rightarrow \ket{e^{\lambda t}}_f\ket{\lambda}_e\ket{t}_t$, which can be constructed by classical arithmetic. 
Furthermore, let another classical-efficiently computable function $f(\lambda, t)$ to be 
\begin{equation}\label{eqn:NSD_func_integral}
    f(\lambda,t) = \frac{1}{t} \int_0^t e^{\lambda(t-s)} ds 
    = \begin{cases}
        1 , & \text{ if } \lambda = 0,\\ 
        \frac{1}{\lambda t}(e^{\lambda t} - 1) , & \text{ else,}
        \end{cases}
\end{equation}
and it is given by the oracle $O_f: \ket{0}_f\ket{\lambda}_e\ket{t}_t \rightarrow \ket{f(\lambda,t)}_f\ket{\lambda}_e\ket{t}_t$. 
The oracle $O_{\exp}$ will be used for constructing $e^{AT}$, and the oracle $O_f$ is associated with $\int_0^T e^{A(T-s)}ds$ and will be discussed later. 
Here the meanings of the subscripts are: t for time, e for eigenvalue, f for function, and v for vector. 

For the initial condition and the inhomogeneous term, we assume that $O_u$ and $O_b$ are the oracles such that $O_u \ket{0} = \frac{1}{\|u(0)\|} \sum_{j=0}^{N-1} u_j(0) \ket{j}$ and $O_b \ket{0} = \frac{1}{\|b\|} \sum_{j=0}^{N-1} b_j \ket{j}$, and assume that $\|u(0)\|$, $\|b\|$ are known.

\subsubsection{Homogeneous case}

According to the equation $e^{AT} = U e^{\Lambda T} U^{\dagger}$, it suffices to focus on the diagonal transform $e^{\Lambda T}$, which can be implemented by controlled rotations. 
Notice that sequentially applying $O_T$, $O_{\Lambda}$, $O_{\exp}$, $O_{\Lambda}^{\dagger}$ and $O_T^{\dagger}$ on corresponding registers can map $\ket{j}_v\ket{0}_f\ket{0}_t\ket{0}_e$ to $\ket{j}_v\ket{e^{\lambda_j T}}_f\ket{0}_t\ket{0}_e$, \emph{i.e.}, encode the information of $e^{\lambda_j T}$ into the function register. 
Then, starting with a quantum state $\sum_{j=0}^{N-1} v_j \ket{j}_v$ encoding a normalized vector $v$, we first append the information of $e^{\lambda_j T}$ through the previous operator to get $\sum_{j=0}^{N-1} v_j \ket{j}_v\ket{e^{\lambda_j T}}_f$, then append another control register with single qubit and apply a rotation conditioned by the function register to get $\sum_{j=0}^{N-1} v_j \ket{j}_v (e^{\lambda_j T}\ket{0}_r + \sqrt{1-e^{2\lambda_j T}}\ket{1}_r )\ket{e^{\lambda_j T}}_f$. 
Notice that, after uncomputing, the part with $\ket{0}_r$ exactly encodes the vector $e^{\Lambda T} v$ in the amplitude. 

The entire circuit for constructing a block-encoding of $e^{\Lambda T}$ is given in~\cref{fig:circuit_FF_NSD_homo}, and we demonstrate its effectiveness and cost as follows. 

\begin{figure}
    \centerline{
    \Qcircuit @R=1em @C=1em {
    \text{Rotation} \quad\quad\quad\quad\quad\quad & \qw & \qw & \qw & \gate{R} & \qw & \qw & \qw & \qw  \\
    \text{Function} \quad\quad\quad\quad\quad\quad & \qw & \qw & \multigate{2}{O_{\exp}} & \ctrl{-1} & \multigate{2}{O_{\exp}^{\dagger}} & \qw & \qw & \qw  \\ 
    \text{Time} \quad\quad\quad\quad\quad\quad & \qw & \gate{O_T} & \ghost{O_{\exp}} & \qw & \ghost{O_{\exp}^{\dagger}} & \gate{O_T^{\dagger}} & \qw & \qw  \\
    \text{Eigenvalue} \quad\quad\quad\quad\quad\quad & \qw &  \multigate{1}{O_{\Lambda}} & \ghost{O_{\exp}} & \qw & \ghost{O_{\exp}^{\dagger}} & \multigate{1}{O_{\Lambda}^{\dagger}} & \qw & \qw \\
    \text{Vector} \quad\quad\quad\quad\quad\quad & \gate{U^{\dagger}} & \ghost{O_{\Lambda}} & \qw & \qw & \qw & \ghost{O_{\Lambda}^{\dagger}} &  \gate{U} & \qw \\
    }
    }
    \caption{Quantum circuit for constructing a block-encoding of $e^{AT}$ where $A$ is negative semi-definite with known eigenvalues and eigenstates. }
    \label{fig:circuit_FF_NSD_homo}
\end{figure}

\begin{lem}\label{lem:FF_NSD_homo}
    Let $A = U\Lambda U^{\dagger}$ where $U$ is a known unitary and $\Lambda = \text{diag}(\lambda_0,\cdots,\lambda_{N-1})$ such that $\lambda_j \leq 0$ for all $j$. 
    Suppose we are given the oracles described in~\cref{sec:oracles_NSD}, and the time, eigenvalue and function registers consist of $n_t,n_e$ and $n_f$ qubits, respectively. 
    Then a $(1,n_t+n_e+n_f+1,0)$-block-encoding of $e^{AT}$ can be constructed, using $6$ queries to $O_T$, $O_{\Lambda}$, $O_{\exp}$ and their inverses, $2$ queries to $U$ and its inverse, and an extra controlled one-qubit rotation gate. 
\end{lem}
\begin{proof}
    Let $\ket{v}_v = \sum_{j=0}^{N-1} v_j \ket{j}_v$ be an arbitrary state, and we start with $\sum_{j=0}^{N-1} v_j \ket{j}_v \ket{0}_r \ket{0}_f \ket{0}_e \ket{0}_t$. 
    We first apply $U^{\dagger}$ on the vector register to get
    \begin{equation}
        \sum_{j=0}^{N-1} (U^{\dagger}v)_j \ket{j}_v \ket{0}_r \ket{0}_f \ket{0}_e \ket{0}_t. 
    \end{equation}
    Here for a matrix $M$ and a vector $x$, $(Mx)_j$ represents its $j$-th component. 
    
    Now we implement the diagonal transform $e^{\Lambda T}$. 
    Applying $O_T$ and $O_{\Lambda}$ on the corresponding registers gives 
    \begin{equation}
        \sum_{j=0}^{N-1} (U^{\dagger}v)_j \ket{j}_v \ket{0}_r \ket{0}_f \ket{\lambda_j}_e \ket{T}_t. 
    \end{equation}
    Apply $O_{\exp}$ on the corresponding registers to get 
    \begin{equation}
        \sum_{j=0}^{N-1} (U^{\dagger}v)_j \ket{j}_v \ket{0}_r \ket{e^{\lambda_j T}}_f \ket{\lambda_j}_e \ket{T}_t. 
    \end{equation}
    Notice that $e^{\lambda_j T} \leq 1$ since $\lambda_j \leq 0$, then we apply a rotation on the rotation register conditioned by the function register and get 
    \begin{equation}
    \begin{split}
        & \quad \sum_{j=0}^{N-1} (U^{\dagger}v)_j \ket{j}_v (e^{\lambda_j T}\ket{0}_r + \sqrt{1-e^{2\lambda_j T}} \ket{1}_r) \ket{e^{\lambda_j T}}_f \ket{\lambda_j}_e \ket{T}_t \\
        & = \sum_{j=0}^{N-1} (e^{\Lambda T} U^{\dagger}v)_j \ket{j}_v \ket{0}_r \ket{e^{\lambda_j T}}_f \ket{\lambda_j}_e \ket{T}_t +  \ket{\perp}. 
    \end{split}
    \end{equation}
    Then we uncompute the function, eigenvalue and time registers by applying $O_{\exp}^{\dagger}$, $O_{\Lambda}^{\dagger}$ and $O_T^{\dagger}$, and we get 
    \begin{equation}
    \begin{split}
        \sum_{j=0}^{N-1} (e^{\Lambda T} U^{\dagger}v)_j \ket{j}_v \ket{0}_r \ket{0}_f \ket{0}_e \ket{0}_t +  \ket{\perp}. 
    \end{split}
    \end{equation}
    
    Finally, applying $U$ gives 
    \begin{equation}
        \sum_{j=0}^{N-1} (Ue^{\Lambda T} U^{\dagger}v)_j \ket{j}_v \ket{0}_r \ket{0}_f \ket{0}_e \ket{0}_t +  \ket{\perp}. 
    \end{equation}
    According to the definition of the block-encoding and the equation $e^{AT} = Ue^{\Lambda_j T} U^{\dagger}$, such a circuit is exactly a $(1,n_t+n_e+n_f+1,0)$-block-encoding of $e^{AT}$. 
\end{proof}

A direct consequence of~\cref{lem:FF_NSD_homo} is the cost of solving homogeneous ODE. 
We can directly apply the block-encoding of $e^{AT}$ to the input state, and then measure all the ancilla qubits. 
If the outcome is all 0, then the remaining state is the desired $\ket{u(T)}$. 

\begin{thm}\label{thm:FF_NSD_homo}
    Consider solving the ODE system~\cref{eqn:ODE}, where $b = 0$ and $A$ is a negative semi-definite Hermitian matrix with known eigenvalues and eigenstates. 
    Suppose that we are given the oracles described in~\cref{sec:oracles_NSD}. 
    Then for any $T > 0$, there exists a quantum algorithm that outputs $\ket{u(T)}$ with $\Omega(1)$ success probability, using 
    \begin{enumerate}
        \item $\mathcal{O}(\|u(0)\|/\|u(T)\|)$ queries to $O_u$, $O_T$, $O_{\Lambda}$, $O_{\exp}$, $U$ and their inverses, 
        \item $(n_t+n_e+n_f+1)$ ancilla qubits, 
        \item $\mathcal{O}(\|u(0)\|/\|u(T)\|)$ extra one-qubit gates. 
    \end{enumerate}
\end{thm}
\begin{proof}
    Applying $O_u$ prepares $\ket{u(0)}_v\ket{0}_r\ket{0}_f\ket{0}_e\ket{0}_t$. 
    \cref{lem:FF_NSD_homo} tells that applying the block-encoding of $e^{AT}$ gives 
    \begin{equation}
        \frac{1}{\|u(0)\|} \sum_{j=0}^{N-1} (e^{AT}u(0))_j \ket{j}_v \ket{0}_r \ket{0}_f \ket{0}_e \ket{0}_t +  \ket{\perp}. 
    \end{equation}
    Therefore the first part after successful measurement is $\ket{u(T)}$, and the number of repeats to boost the success probability to $\Omega(1)$ is $\mathcal{O}(\|u(0)\|/\|u(T)\|)$ with amplitude amplification.  
\end{proof}

According to~\cref{thm:FF_NSD_homo}, the query complexity of our algorithm is $\mathcal{O}(\|u(0)\|/\|u(T)\|)$, which still might depend on $T$ and $\|A\|$. 
Therefore, the scalings on $T$ and $\|A\|$ heavily depends on the decay of the solution, and it is possible that the scalings are independent of $T$ and $\|A\|$. 
More specifically, a sufficient condition is that $A$ has $0$ eigenvalue and $u(0)$ has non-trivial overlap with the $0$ eigenstate. 
To see this, suppose that $\lambda_0 = 0$ and $|\braket{u(0)|\psi_0}| \geq \Omega(1)$ where $\ket{\psi_0}$ is the eigenstate corresponding to $\lambda_0$. 
Then, let $\ket{\psi_j}$ denote the eigenstates corresponding to $\lambda_j$, and we get 
\begin{equation}\label{eqn:FF_NSD_norm_decay}
\begin{split}
     \frac{u(T)}{\|u(0)\|} &= U e^{\Lambda T} U^{\dagger} \ket{u(0)} \\
     & = U e^{\Lambda T} U^{\dagger} (\sum_{j=0}^{N-1} \braket{\psi_j|u(0)}\ket{\psi_j}) \\
     & = \sum_{j=0}^{N-1} e^{\lambda_j T}\braket{\psi_j|u(0)}\ket{\psi_j}.
\end{split}
\end{equation}
Therefore $\|u(T)\|/\|u(0)\| \geq e^{\lambda_0 T} |\braket{\psi_0|u(0)}| = \Omega(1)$, and the query complexity is only $\mathcal{O}(1)$. 
We remark that the assumptions $\lambda_0=0$ and $|\braket{u(0)|\psi_0}| \geq \Omega(1)$ are satisfied in the application of the heat equation and the advection-diffusion equation with periodic boundary condition. 
There the largest eigenvalue of the discretized Laplacian operator and divergence operator is exactly $0$, and $|\braket{u(0)|\psi_0}| \geq \Omega(1)$ is naturally assured in practically applicable cases where the initial heat $\int_{[0,1]^d} u(0,x)dx$ is not $0$. 

\subsubsection{Inhomogeneous case}

Now we study the inhomogeneous ODE with time-independent $b$. 
We first show how to construct a linear combination of the homogeneous and the inhomogeneous parts. 
To make the result more general, here we estimate the complexity in terms of the number of queries to the block-encoding to $e^{AT}$ and $\int_0^T e^{A(T-s)}ds$. 
The quantum circuit is given in~\cref{fig:circuit_FF_LCS}. 

\begin{figure}
    \centerline{
    \Qcircuit @R=1em @C=1em {
    \ket{0}_c \quad\quad\quad & \gate{\mathrm{R_y}(\theta)} & \ctrl{2} & \ctrl{1} & \gate{\mathrm{X}} & \ctrl{2} & \ctrl{1} & \gate{\mathrm{X}} & \gate{\mathrm{H}} & \qw & \rstick{\!\!\!\!\!\bra{0}} \\
    \ket{0}_a \quad\quad\quad & \qw & \qw & \multigate{1}{U_0} & \qw & \qw & \multigate{1}{U_1} & \qw & \qw & \qw & \rstick{\!\!\!\!\!\bra{0}} \\ 
    \ket{0}_s \quad\quad\quad & \qw & \gate{O_u} & \ghost{U_0} & \qw & \gate{O_b} & \ghost{U_1} & \qw & \qw & \qw & \\
    }
    }
    \caption{ Quantum circuit for approximating the solution $\ket{u(T)}$. Here the angle $\theta$ using in the rotation is $-2\arcsin(\alpha_1\|b\|/\sqrt{\alpha_0^2\|u(0)\|^2 + \alpha_1^2\|b\|^2})$. $U_0$ and $U_1$ are the block-encodings of $e^{AT}$ and $\int_{0}^T e^{A(T-s)}ds$ with normalization factors $\alpha_0$ and $\alpha_1$, respectively.  }
    \label{fig:circuit_FF_LCS}
\end{figure}
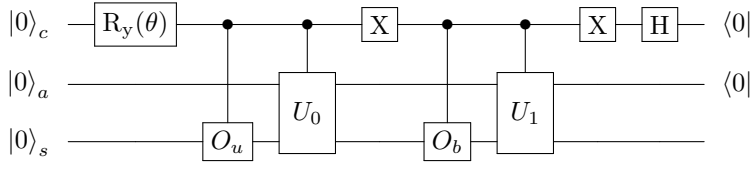

\begin{lem}\label{lem:FF_LCS}
    Consider solving the ODE system~\cref{eqn:ODE_general} with time-independent $b$ up to time $T$. 
    Suppose that $U_0$ and $U_1$ are $(\alpha_0,n',\epsilon_0)$- and $(\alpha_1,n',\epsilon_1)$-block-encodings of $e^{AT}$ and $\int_0^T e^{A(T-s)}ds$, respectively, and assume that $\alpha_0$ and $\alpha_1$ are known. 
    Then, for $0 \leq \epsilon < 1/2$, there exists a quantum algorithm that outputs an $\epsilon$-approximation of $\ket{u(T)}$ with $\Omega(1)$ success probability, using 
    \begin{enumerate}
        \item  \begin{equation}
        \mathcal{O}\left(\frac{\alpha_0\|u(0)\| + \alpha_1 \|b\|}{\|u(T)\|}\right)
    \end{equation}
    queries to the controlled versions of $O_u$, $O_b$, $U_0$ and $U_1$, with the tolerated error in $U_0$ and $U_1$ to be 
    \begin{equation}
        \epsilon_0 = \frac{\|u(T)\|\epsilon}{4\|u(0)\|}, \quad \epsilon_1 = \frac{\|u(T)\|\epsilon}{4\|b\|}, 
    \end{equation}
        \item $(n'+1)$ ancilla qubits, 
        \item \begin{equation}
        \mathcal{O}\left(\frac{\alpha_0\|u(0)\| + \alpha_1 \|b\|}{\|u(T)\|}\right)
    \end{equation} 
    extra one-qubit gates. 
    \end{enumerate}
\end{lem}
\begin{proof}
    We start with the state $\ket{0}_c\ket{0}_a\ket{0}_v$ where the subscript ``a'' represents ancilla qubits in the block-encodings, ``c'' represents a single control qubit, and ``v'' represents the vector space encoding the solution. 
    The first step is applying a rotation $\mathrm{R}_y(-2\arcsin(\alpha_1\|b\|/\sqrt{\alpha_0^2\|u(0)\|^2+\alpha_1^2\|b\|^2}))$ on the control register to obtain 
    \begin{equation}
        \frac{1}{\sqrt{\alpha_0^2\|u(0)\|^2+\alpha_1^2\|b\|^2}} (\alpha_0\|u(0)\|\ket{0}_c + \alpha_1\|b\| \ket{1}_c) \ket{0}_a\ket{0}_v. 
    \end{equation}
    Applying $\ket{0}_c\bra{0}_c \otimes O_u + \ket{1}_c\bra{1}_c \otimes O_b$ on the control and vector registers gives 
    \begin{equation}
        \frac{1}{\sqrt{\alpha_0^2\|u(0)\|^2+\alpha_1^2\|b\|^2}} (\alpha_0\|u_0\| \ket{0}_c\ket{0}_a\ket{u(0)}_v + \alpha_1\|b\| \ket{1}_c\ket{0}_a\ket{b}_v). 
    \end{equation}
    Next, applying $\ket{0}_c\bra{0}_c \otimes U_0 + \ket{1}_c\bra{1}_c \otimes U_1$, which is the block-encoding of $e^{AT}$ and $\int_0^T e^{A(T-s)}ds$ controlled by the control register, gives 
    \begin{equation}
    \begin{split}
        \frac{1}{\sqrt{\alpha_0^2\|u(0)\|^2+\alpha_1^2\|b\|^2}} &\Biggl(\|u(0)\| \ket{0}_c\ket{0}_a\left(e^{AT}\ket{u(0)}_v + \|e_0\|\ket{e_0}_v\right)  \\
        & \quad\quad + \|b\| \ket{1}_c\ket{0}_a\left( \int_0^T e^{A(T-s)}ds \ket{b}_v + \|e_1\|\ket{e_1}_v\right)\Biggr) + \ket{\perp}, 
    \end{split}
    \end{equation}
    where $e_0$ and $e_1$ are two error terms with $\|e_0\| \leq \epsilon_0$ and $ \|e_1\| \leq \epsilon_1$, and $\ket{\perp}$ represents terms where the ancilla qubit is not 0. 
    Applying a Hadamard gate on the control register yields 
    \begin{equation}
        \begin{split}
        \frac{1}{\sqrt{2}\sqrt{\alpha_0^2\|u(0)\|^2+\alpha_1^2\|b\|^2}} & \Biggl( \ket{0}_c\ket{0}_a\left(\|u(0)\| e^{AT}\ket{u(0)}_v + \|b\|\int_0^T e^{A(T-s)}ds \ket{b}_v \right)
         \\
        & \quad\quad + \ket{0}_c\ket{0}_a 
        \left( \|u(0)\|\|e_0\|\ket{e_0}_v + \|b\|\|e_1\|\ket{e
        _1}_v \right) \Biggr) + \ket{\perp}, 
    \end{split}
    \end{equation}
    and, according to~\cref{eqn:ODE_solu}, this quantum state is the same as 
    \begin{equation}
        \begin{split}
        \frac{1}{\sqrt{2}\sqrt{\alpha_0^2\|u(0)\|^2+\alpha_1^2\|b\|^2}} &\left( \ket{0}_c\ket{0}_a\left(\|u(T)\|\ket{u(T)}_v + \|u(0)\|\|e_0\|\ket{e_0}_v + \|b\|\|e_1\|\ket{e_1}_v \right) \right) \\
        & \quad\quad + \ket{\perp}. 
    \end{split}
    \end{equation}
    If we measure the control and ancilla registers and get $0$ for both measurements, then the vector register encodes an approximation of $\ket{u(T)}$. 
    The cost of a single run is a single use of the controlled versions of $O_u$, $O_b$, $U_0$ and $U_1$, and $\mathcal{O}(1)$ extra one-qubit gates. 
    
    Now we analyze the approximation error and the success probability. 
    Let $\ket{\psi}$ denote the quantum state in the vector register after a successful measurement. 
    According to~\cref{lem:succ_prob_error}, the approximation error can be bounded as 
    \begin{equation}
        \begin{split}
            \|\ket{\psi}-\ket{u(T)}\| &\leq  \frac{2\sqrt{2}\sqrt{\alpha_0^2\|u(0)\|^2+\alpha_1^2\|b\|^2}}{\|u(T)\|} \left\|\frac{ \|u(0)\|\|e_0\|\ket{e_0}_v + \|b\|\|e_1\|\ket{e_1}_v}{\sqrt{2}\sqrt{\alpha_0^2\|u(0)\|^2+\alpha_1^2\|b\|^2}}\right\| \\
            & \leq \frac{2(\|u(0)\|\epsilon_0 + \|b\|\epsilon_1)}{\|u(T)\|}, 
        \end{split}
    \end{equation}
    and the success probability, after amplitude amplification, can be bounded as 
    \begin{equation}
        \begin{split}
            \text{Prob} &\geq \frac{\|u(T)\|}{\sqrt{2}\sqrt{\alpha_0^2\|u(0)\|^2+\alpha_1^2\|b\|^2}} - \left\|\frac{ \|u(0)\|\|e_0\|\ket{e_0}_v + \|b\|\|e_1\|\ket{e_1}_v}{\sqrt{2}\sqrt{\alpha_0^2\|u(0)\|^2+\alpha_1^2\|b\|^2}}\right\| \\
            & \geq \frac{\|u(T)\|}{\sqrt{2}\sqrt{\alpha_0^2\|u(0)\|^2+\alpha_1^2\|b\|^2}} - \frac{ \|u(0)\|\epsilon_0 + \|b\|\epsilon_1}{\sqrt{2}\sqrt{\alpha_0^2\|u(0)\|^2+\alpha_1^2\|b\|^2}}. 
        \end{split}
    \end{equation}
    Choose $\epsilon_0 = \|u(T)\|\epsilon/(4\|u(0)\|)$ and $\epsilon_1 = \|u(T)\|\epsilon/(4\|b\|)$, then 
    \begin{equation}
        \|\ket{\psi}-\ket{u(T)}\| \leq \epsilon
    \end{equation}
    and 
    \begin{equation}
        \text{Prob} \geq \frac{\|u(T)\| (1-\epsilon/2)}{\sqrt{2}\sqrt{\alpha_0^2\|u(0)\|^2+\alpha_1^2\|b\|^2}} \geq \frac{\|u(T)\|}{2\sqrt{2}\sqrt{\alpha_0^2\|u(0)\|^2+\alpha_1^2\|b\|^2}}. 
    \end{equation}
    Therefore, the success probability can be boosted to $\Omega(1)$ using 
    \begin{equation}
        \mathcal{O}\left(\frac{\alpha_0\|u(0)\| + \alpha_1 \|b\|}{\|u(T)\|}\right)
    \end{equation}
    repeats with amplitude amplification. 
\end{proof}

To use~\cref{lem:FF_LCS}, we need to construct a block-encoding of the operator $\int_0^T e^{A(T-s)} ds$. 
Let $f(\lambda, t)$ be the function defined in~\cref{eqn:NSD_func_integral}. 
Then $\frac{1}{T}\int_0^T e^{A(T-s)} ds = UDU^{\dagger}$ where $D$ is a diagonal matrix with $j$-th diagonal component to be $f(\lambda_j,T)$. 
Since $|f(\lambda,t)| \leq 1$ for all $t > 0$ and $\lambda \leq 0$, we can implement a block-encoding of $\frac{1}{T}\int_0^T e^{A(T-s)} ds$ using a circuit that is almost the same as that in~\cref{fig:circuit_FF_NSD_homo}, except the oracle $O_{\exp}$ is replaced by the oracle $O_f$. 
The cost is as follows, whose proof follows exactly from that of~\cref{lem:FF_NSD_homo}. 

\begin{lem}\label{lem:FF_NSD_inhomo}
    Let $A = U\Lambda U^{\dagger}$ where $U$ is a known unitary and $\Lambda = \text{diag}(\lambda_0,\cdots,\lambda_{N-1})$ such that $\lambda_j \leq 0$ for all $j$. 
    Suppose we are given the oracles described in~\cref{sec:oracles_NSD}, and the time, eigenvalue and function registers consist of $n_t,n_e$ and $n_f$ qubits, respectively. 
    Then a $(T,n_t+n_e+n_f+1,0)$-block-encoding of $\int_0^T e^{A(T-s)} ds$ can be constructed, using $6$ queries to $O_T$, $O_{\Lambda}$, $O_{f}$ and their inverses, $2$ queries to $U$ and its inverse, and an extra controlled one-qubit rotation gate. 
\end{lem}

Combining~\cref{lem:FF_NSD_homo},~\cref{lem:FF_NSD_inhomo} and~\cref{lem:FF_LCS}, we can solve the inhomogeneous ODE using the circuit in~\cref{fig:circuit_FF_LCS}, and the cost is summarized as follows. 

\begin{thm}\label{thm:FF_NSD}
    Consider solving the ODE system~\cref{eqn:ODE_general} with time-independent $b$ up to time $T$, where $A$ is a negative semi-definite Hermitian matrix with known eigenvalues and eigenstates. 
    Suppose we are given the oracles described in~\cref{sec:oracles_NSD}, and the time, eigenvalue and function registers consist of $n_t,n_e$ and $n_f$ qubits, respectively. 
    Then for any $T>0$, there exists a quantum algorithm that outputs the normalized solution $\ket{u(T)}$ with $\Omega(1)$ success probability, using 
    \begin{enumerate}
        \item \begin{equation}
            \mathcal{O}\left(\frac{\|u(0)\| + T\|b\|}{\|u(T)\|}\right)
        \end{equation}
        queries to the given oracles, their inverses and controlled versions, 
        \item $(n_t+n_e+n_f+2)$ ancilla qubits, 
        \item \begin{equation}
            \mathcal{O}\left(\frac{\|u(0)\| + T\|b\|}{\|u(T)\|}\right)
        \end{equation}
        extra one-qubit gates. 
    \end{enumerate}
\end{thm}

Compared to~\cref{lem:DEsolver_TI}, we avoid the linear dependence on $\|A\|$ in query complexity. 
In the worst case, the query complexity of our algorithm is still linear in $T$, but a possible linear growth of $\|u(T)\|$ may cancel it and leads to an $\mathcal{O}(1)$ scalings in $T$ as well (see the discussions after~\cref{thm:FF_square_root_inhomo}). 
We also remark that an $\mathcal{O}(1)$ query complexity can be obtained in many more scenarios if we allow the eigenvalues to be complex and use their knowledge in constructing the quantum circuit. 
This will be discussed in the next subsection. 

\subsection{Generalization to matrices with known complex eigenvalues and eigenstates}\label{sec:FF_NSD_gen}

We generalize the results in~\cref{sec:FF_NSD} to the case where the eigenvalues are not assumed to be real and non-positive. 
Most techniques are the same as those in~\cref{sec:FF_NSD} with two main differences: we need to perform another controlled rotation to deal with the imaginary parts of the eigenvalues, and we use the real-part difference to perform the rotation associated with the real parts to deal with possibly positive real parts of the eigenvalues. 

\subsubsection{Oracles and notation}\label{sec:oracles_NSD_gen}

Let $A = U\Lambda U^{\dagger}$, $\Lambda = \text{diag}(\lambda_0,\cdots,\lambda_{N-1})$, and $\lambda_j = \alpha_j + i \beta_j$ where $\alpha_j$ and $\beta_j$ represent the real part and the imaginary part of $\lambda_j$, respectively. 
Let the time $T$ be given by $O_T: \ket{0}_t \rightarrow \ket{T}_t$. 
We assume that the eigenvalues are given by the oracles $O_{\Lambda,r}: \ket{0}_{e}\ket{j}_v \rightarrow \ket{\alpha_j}_{e}\ket{j}_v$ and $O_{\Lambda,i}: \ket{0}_{e}\ket{j}_v \rightarrow \ket{\beta_j}_{e}\ket{j}_v$, and the product operation between eigenvalues and time is given by $O_{\text{prod}}: \ket{0}_f\ket{\beta_j}_e\ket{t}_t \rightarrow \ket{\beta_j t}_f\ket{\beta_j}_e\ket{t}_t$. 

Assume two parameters $\alpha$ and $\beta$ are know, where $\alpha$ is the largest $\alpha_j$, and $\beta$ is a lower bound of $\min_{\alpha_j = 0} |\beta_j|$ (\emph{i.e.}, the minimum of the lengths of purely imaginary eigenvalues). 
Notice that $\alpha$ is the logarithmic norm of $A$ associated with the matrix 2-norm. 

Associated with the homogeneous operator $e^{AT}$, we assume access to the oracle $O_{\exp,\alpha}: \ket{0}_f\ket{\alpha_j}_e\ket{t}_t \rightarrow \ket{e^{(\alpha_j-\alpha)t}}_f\ket{\alpha_j}_e\ket{t}_t$. 
Associated with the inhomogeneous $\int_0^T e^{A(T-s)}ds$, we define 
\begin{equation}\label{eqn:FF_NSD_gen_defC}
    C(\alpha,\beta,T) \coloneqq \begin{cases}
            T, & \text{ if } \alpha = \beta = 0, \\
            \frac{2}{\beta},  & \text{ if } \alpha = 0 \text{ and } \beta \neq 0,\\
            \frac{1}{\alpha} (e^{\alpha T} - 1), & \text{ else,}
        \end{cases}
\end{equation}
and let 
\begin{equation}\label{eqn:FF_NSD_gen_deffg}
    C(\alpha,\beta,T)^{-1} \int_0^T e^{\lambda_j (T-s)} ds = f(\alpha_j,\beta_j,T) + i g(\alpha_j,\beta_j,T), 
\end{equation}
where $f$ and $g$ represent the real and imaginary parts, respectively. 
Notice that these two functions $f$ and $g$ are classically computable, and thus we assume access to the oracles 
\begin{equation}
    O_f: \ket{0}_f\ket{\alpha_j}_{e,r}\ket{\beta_j}_{e,i}\ket{t}_{t} \rightarrow \ket{f(\alpha_j,\beta_j,T)}_f\ket{\alpha_j}_{e,r}\ket{\beta_j}_{e,i}\ket{t}_{t}, 
\end{equation}
and 
\begin{equation}
    O_g: \ket{0}_f\ket{\alpha_j}_{e,r}\ket{\beta_j}_{e,i}\ket{t}_{t} \rightarrow \ket{g(\alpha_j,\beta_j,T)}_f\ket{\alpha_j}_{e,r}\ket{\beta_j}_{e,i}\ket{t}_{t}. 
\end{equation} 
Here the subscript ``e,r'' refers to eigenvalue's real part, and ``e,i'' refers to eigenvalue's imaginary part. 

Same as before, for the vectors $u(0)$ and $b$, we assume $O_u \ket{0} = \frac{1}{\|u(0)\|} \sum_{j=0}^{N-1} u_j(0) \ket{j}$ and $O_b \ket{0} = \frac{1}{\|b\|} \sum_{j=0}^{N-1} b_j \ket{j}$, and assume that $\|u(0)\|$, $\|b\|$ are known. 

\subsubsection{Homogeneous case}

Notice that $\|e^{\Lambda T}\|$ is now bounded by $e^{\alpha T}$, we need to perform the normalized operator $e^{\Lambda T}/e^{\alpha T}$. 
As a result, we replace the oracle $O_{\exp}$ in~\cref{sec:FF_NSD} by the oracle $O_{\exp,\alpha}: \ket{0}_f\ket{\alpha_j}_e\ket{t}_t \rightarrow \ket{e^{(\alpha_j-\alpha)t}}_f\ket{\alpha_j}_e\ket{t}_t$. 
Then the real parts can be implemented using the same approach in~\cref{sec:FF_NSD}. 
The imaginary parts only introduce another $e^{i\beta_jT}$ factor on the basis, which can be implemented by the controlled phase shift gate. 
The entire circuit is summarized in~\cref{fig:circuit_FF_NSD_gen_homo}, and we claim its effetiveness and cost as follows. 
The proof is very similar with that of~\cref{lem:FF_NSD_homo} and is given in~\cref{app:proof_FF_gen_complex}.

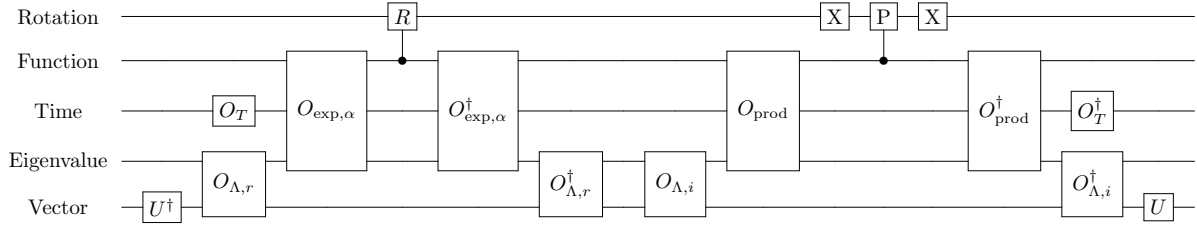
\begin{figure}
\centering
\scalebox{0.8}{
\hspace{1cm}
    \centerline{
    \Qcircuit @R=1em @C=1em {
    \text{Rotation} \quad\quad\quad\quad\quad\quad & \qw & \qw & \qw & \gate{R} & \qw & \qw & \qw & \qw & \qw & \gate{\mathrm{X}} & \gate{\mathrm{P}} & \gate{\mathrm{X}} & \qw & \qw & \qw & \qw\\
    \text{Function} \quad\quad\quad\quad\quad\quad & \qw & \qw & \multigate{2}{O_{\exp,\alpha}} & \ctrl{-1} & \multigate{2}{O_{\exp,\alpha}^{\dagger}} & \qw & \qw & \qw & \multigate{2}{O_{\text{prod}}} & \qw & \ctrl{-1} & \qw & \multigate{2}{O_{\text{prod}}^{\dagger}} & \qw & \qw & \qw \\ 
    \text{Time} \quad\quad\quad\quad\quad\quad & \qw & \gate{O_T} & \ghost{O_{\exp,\alpha}} & \qw & \ghost{O_{\exp,\alpha}^{\dagger}} & \qw & \qw & \qw & \ghost{O_{\text{prod}}} & \qw & \qw & \qw & \ghost{O_{\text{prod}}^{\dagger}} & \gate{O_T^{\dagger}} & \qw & \qw\\
    \text{Eigenvalue} \quad\quad\quad\quad\quad\quad & \qw &  \multigate{1}{O_{\Lambda,r}} & \ghost{O_{\exp,\alpha}} & \qw & \ghost{O_{\exp,\alpha}^{\dagger}} & \multigate{1}{O_{\Lambda,r}^{\dagger}} & \qw & \multigate{1}{O_{\Lambda,i}} & \ghost{O_{\text{prod}}} & \qw & \qw & \qw & \ghost{O_{\text{prod}}^{\dagger}} & \multigate{1}{O_{\Lambda,i}^{\dagger}} & \qw & \qw 
    \\
    \text{Vector} \quad\quad\quad\quad\quad\quad & \gate{U^{\dagger}} & \ghost{O_{\Lambda,r}} & \qw & \qw & \qw & \ghost{O_{\Lambda,r}^{\dagger}} & \qw & \ghost{O_{\Lambda,i}} & \qw & \qw & \qw & \qw & \qw & \ghost{O_{\Lambda,i}^{\dagger}} & \gate{U} & \qw
    }
    }
    }
    \caption{ Quantum circuit for constructing a block-encoding of $e^{AT}$ where $A$ is a unitarily diagonalizable matrix with known eigenvalues and eigenstates. }
    \label{fig:circuit_FF_NSD_gen_homo}
\end{figure}

\begin{lem}\label{lem:FF_NSD_gen_homo}
    Let $A = U\Lambda U^{\dagger}$ where $U$ is a known unitary and $\Lambda = \text{diag}(\lambda_0,\cdots,\lambda_{N-1})$ such that $\lambda_j = \alpha_j + i\beta_j$ for all $j$. 
    Suppose the oracles described in~\cref{sec:oracles_NSD_gen}, and the time, eigenvalue and function registers consist of $n_t,n_e$ and $n_f$ qubits, respectively. 
    Then an $(e^{\alpha T},n_t+n_e+n_f+1,0)$-block-encoding of $e^{AT}$ can be constructed, using $10$ queries to aforementioned oracles and their inverses, $2$ queries to $U$ and its inverse, and $4$ extra (controlled) one-qubit gates. 
\end{lem}

\begin{thm}\label{thm:FF_NSD_gen_homo}
    Consider solving the ODE system~\cref{eqn:ODE_general}, where $b = 0$ and $A$ is a unitarily diagonalizable matrix with known eigenvalues and eigenstates. 
    Suppose that we are given the oracles as described in~\cref{sec:oracles_NSD_gen}, and that $\alpha$ is the largest eigenvalues' real part.  
    Then for any $T > 0$, there exists a quantum algorithm that outputs $\ket{u(T)}$ with $\Omega(1)$ success probability, using 
    \begin{enumerate}
        \item $\mathcal{O}(e^{\alpha T}\|u(0)\|/\|u(T)\|)$ queries to $O_u$, $O_T$, $O_{\Lambda,r}$, $O_{\Lambda,i}$, $O_{\exp,\alpha}$, $O_{\text{prod}}$, $U$ and their inverses, 
        \item $(n_t+n_e+n_f+1)$ ancilla qubits, 
        \item $\mathcal{O}(e^{\alpha T}\|u(0)\|/\|u(T)\|)$ extra one-qubit gates. 
    \end{enumerate}
\end{thm}

There are two generalizations of~\cref{thm:FF_NSD_gen_homo} compared to~\cref{thm:FF_NSD_homo}. 
First, as defined,~\cref{thm:FF_NSD_gen_homo} applies to more general matrix $A$ with possibly complex eigenvalues. 
An important category among those is the fast-forwarding for time-independent Hamiltonian simulation, in which $\|u(0)\| = \|u(T)\|$ and $\alpha = 0$. 
Therefore,~\cref{thm:FF_NSD_gen_homo} implies an $\mathcal{O}(1)$ query complexity of simulating Hamiltonians with known eigenvalues and eigenstates, including diagonal Hamiltonians and more general $1$-sparse Hamiltonians. 
Second, thanks to the \emph{a priori} knowledge on $\alpha$, we could avoid possible exponential dependence caused by either positive real parts of the eigenvalues or norm decay. 
Specifically, let $\alpha = \alpha_0 \geq \alpha_1 \geq \cdots \geq \alpha_{N-1}$. 
According to~\cref{eqn:FF_NSD_norm_decay}, we have $\|u(T)\|/\|u(0)\| \geq |e^{\lambda_0 T}| |\braket{u(0)|\psi_0}| = e^{\alpha T} |\braket{u(0)|\psi_0}|$. 
Once $|\braket{u(0)|\psi_0}| \geq \Omega(1)$, we have $e^{\alpha T} \|u(0)\|/\|u(T)\| = \mathcal{O}(1)$ and thus the corresponding differential equation is solved with query complexity independent of $T$ and $\|A\|$. 
We remark that the underlying reason for this efficient fast-forwarding algorithm is again the shifting equivalence discussed in~\cref{sec:LB_implications}: instead of the original ODE, we indeed solve the ODE after shifting by $\alpha$. 
In the shifted ODE, all the eigenvalues of the coefficient matrix have non-positive real parts and at least one of them has zero real part, and thus the two sources of exponential computational overhead are avoided. 

\subsubsection{Inhomogeneous case}

General inhomogeneous ODE with possibly complex eigenvalues can also be solved using the same approach based on~\cref{lem:FF_LCS}. 
The key is still to construct a block-encoding of the operator $\int_0^T e^{A(T-s)} ds$, and this can be implemented by diagonal transformation as well. 

We first study real and imaginary parts of the diagonal components. 
Let $\lambda_j = \alpha_j + i\beta_j$ be a complex number, and the diagonal element $\int_0^T e^{\lambda_j (T-s)} ds $ can be bounded as follows. 
If $\alpha_j=\beta_j = 0$, then 
\begin{equation}
    \int_0^T e^{\lambda_j (T-s)} ds = T. 
\end{equation}
If $\alpha_j=0$ and $\beta_j \neq 0$, then 
\begin{equation}
    \left|\int_0^T e^{\lambda_j (T-s)} ds\right| = \left|\int_0^T e^{i \beta_j (T-s)} ds \right| = \left| \frac{1}{i\beta_j} (e^{i\beta_jT} - 1)\right| \leq \frac{2}{|\beta_j|}. 
\end{equation}
If $\alpha_j \neq 0$, then 
\begin{equation}
    \left|\int_0^T e^{\lambda_j (T-s)} ds\right| \leq \int_0^T |e^{\lambda_j (T-s)}| ds = \int_0^T e^{\alpha_j (T-s)} ds = \frac{1}{\alpha_j}(e^{\alpha_j T } - 1). 
\end{equation}
In summary, by the definition of the parameters $\alpha$, $\beta$ in~\cref{sec:oracles_NSD_gen} and the constant $C(\alpha,\beta,T)$ in~\cref{eqn:FF_NSD_gen_defC}, we always have 
\begin{equation}
    \left|\int_0^T e^{\lambda_j (T-s)} ds\right| \leq C(\alpha,\beta,T).
\end{equation}
Then the magnitude of $C(\alpha,\beta,T)^{-1} \int_0^T e^{\lambda_j (T-s)} ds$ is bounded by $1$, and we can implement the corresponding diagonal transformation by controlled rotation and phase shift.
More specifically, $\int_0^T e^{\lambda_j (T-s)} ds$ can be block-encoded using a circuit similar to that in~\cref{fig:circuit_FF_NSD_gen_homo}, with the following differences: 
\begin{enumerate}
    \item The single eigenvalue register is extended to two registers to encode real and imaginary parts of the eigenvalues separately and simultaneously. 
    \item The procedure of encoding and uncomputing the eigenvalues now becomes: (after applying $U^{\dagger}$) we first encode both the real and the imaginary parts by applying both $O_{\Lambda,r}$ and $O_{\Lambda,i}$ on the corresponding registers, then only uncompute at the very end, right before applying $U$. 
    \item The oracle $O_{\exp,\alpha}$ is replaced by $O_f$. 
    \item The oracle $O_{\text{prod}}$ is replaced by $O_g$. 
\end{enumerate}

\begin{lem}\label{lem:FF_NSD_gen_inhomo}
    Let $A = U\Lambda U^{\dagger}$ where $U$ is a known unitary and $\Lambda = \text{diag}(\lambda_0,\cdots,\lambda_{N-1})$ with $\lambda_j = \alpha_j + i \beta_j$ for all $j$. 
    Let $C(\alpha,\beta,T)$ be defined in~\cref{eqn:FF_NSD_gen_defC}. 
    Suppose we are given the oracles described in~\cref{sec:oracles_NSD_gen}, and the time, eigenvalue and function registers consist of $n_t,n_e$ and $n_f$ qubits, respectively. 
    Then a $(C(\alpha,\beta,T),n_t+2n_e+n_f+1,0)$-block-encoding of $\int_0^T e^{A(T-s)} ds$ can be constructed, using $10$ queries to aforementioned oracles and their inverses, $2$ queries to $U$ and its inverse, and $4$ extra one-qubit gates. 
\end{lem}

A combination of~\cref{lem:FF_NSD_gen_homo},~\cref{lem:FF_NSD_gen_inhomo} and~\cref{lem:FF_LCS} gives the following result. 
This generalizes~\cref{thm:FF_NSD} to general complex eigenvalues with possibly positive real parts. 

\begin{thm}\label{thm:FF_NSD_gen}
    Consider solving the ODE system~\cref{eqn:ODE_general} with time-independent $b$ up to time $T$, where $A$ is a unitarily diagonalizable matrix with known eigenvalues and eigenstates. 
    Let $\alpha$ be the largest $\alpha_j$, and $C(\alpha,\beta,T)$ be defined in~\cref{eqn:FF_NSD_gen_defC}. 
    Suppose we are given the oracles described in~\cref{sec:oracles_NSD_gen}, and the time, eigenvalue and function registers consist of $n_t,n_e$ and $n_f$ qubits, respectively. 
    Then for any $T>0$, there exists a quantum algorithm that outputs the normalized solution $\ket{u(T)}$ with $\Omega(1)$ success probability, using 
    \begin{enumerate}
        \item \begin{equation}
            \mathcal{O}\left( \frac{ e^{\alpha T} \|u(0)\| + C(\alpha,\beta,T) \|b\|}{\|u(T)\|}\right)
        \end{equation}
        queries to the given oracles, their inverses and controlled versions, 
        \item $(n_t+2n_e+n_f+2)$ ancilla qubits, 
        \item \begin{equation}
            \mathcal{O}\left( \frac{ e^{\alpha T} \|u(0)\| + C(\alpha,\beta,T) \|b\|}{\|u(T)\|} \right)
        \end{equation}
        extra one-qubit gates. 
    \end{enumerate}
\end{thm}

Similar as the discussion after~\cref{thm:FF_NSD}, one can achieve $\mathcal{O}(1)$ scaling in most scenarios. 
Suppose that all the eigenvalues of $A$ have non-positive real parts. 
Then 
\begin{enumerate}
    \item If $\alpha < 0$, \emph{i.e.}, all the eigenvalues have negative real parts and we regard $\alpha$ as a constant, then the query complexity becomes $\mathcal{O}((\|u(0)\|+\|b\|)/\|u(T)\|)$. 
    Furthermore, in this case $A$ is invertible, and thus $\|u(T)\| = \Omega(1)$, which indicates that the query complexity is independent of both $T$ and $\|A\|$. 
    \item If $\alpha = 0$ and $\beta > 0$, \emph{i.e.}, all the eigenvalues of $A$ have non-positive real parts and are non-zero, then the query complexity is also independent of both $T$ and $\|A\|$. 
    The reason is similar with the previous case, and notice that $\|u(T)\|$ is still $\Omega(1)$ for components corresponding to purely imaginary eigenvalues analog to the Hamiltonian simulation. 
    \item If $\alpha = 0$ and $\beta = 0$, \emph{i.e.}, there exists a $0$ eigenvalue. 
    Then the query complexity becomes $\mathcal{O}((\|u(0)\|+T\|b\|)/\|u(T)\|)$. 
    This is the same as~\cref{thm:FF_NSD}, which is always independent of $\|A\|$ and can be independent of $T$ if $b$ has non-trivial overlap with the eigenstate of $A$ corresponding to $0$ eigenvalue. 
\end{enumerate}

\subsection{Generalization to time-dependent inhomogeneous terms}\label{sec:FF_NSD_gen_td}

Consider the most general ODE~\cref{eqn:ODE_general} time-dependent inhomogeneous term $b(t)$. 
Here we still assume that the matrix $A$ is normal with known eigenvalues and eigenstates, following the assumptions in~\cref{sec:FF_NSD_gen}. 
We discuss how to solve this system with potential exponential fast-forwarding in $\|A\|$ and $T$. 

\subsubsection{Oracles}\label{sec:oracles_FF_NSD_gen_bt}

We state the oracles first, and their usages will be clear in later discussions. 
For the matrix $A$, let $\widetilde{\alpha} = \max_{0\leq j \leq N-1}\left\{0,\text{Re}(\lambda_j)\right\}$ be a known parameter. 
    Suppose we are given the oracles: 
    \begin{equation}
        O_T: \ket{0}_t \rightarrow \ket{T}_t, 
    \end{equation}
    \begin{equation}
         O_{\Lambda,r}: \ket{0}_{e}\ket{j}_v \rightarrow \ket{\alpha_j}_{e}\ket{j}_v, \quad O_{\Lambda,i}: \ket{0}_{e}\ket{j}_v \rightarrow \ket{\beta_j}_{e}\ket{j}_v, 
    \end{equation}
    \begin{equation}
        O_{\exp,\widetilde{\alpha}}: \ket{0}_f\ket{\alpha_j}_e\ket{t}_t \rightarrow \ket{e^{(\alpha_j-\widetilde{\alpha})t}}_f\ket{\alpha_j}_e\ket{t}_t,  
    \end{equation}
    \begin{equation}
        O_{\exp,\widetilde{\alpha},t}: \ket{0}_f\ket{\alpha_j}_e\ket{k}_t \rightarrow \ket{e^{\alpha_j T(1-k/M) - \widetilde{\alpha}T}}_f\ket{\alpha_j}_e\ket{k}_t, 
    \end{equation}
    \begin{equation}
        O_{\text{prod}}: \ket{0}_f\ket{\beta_j}_e\ket{t}_t \rightarrow \ket{\beta_j t}_f\ket{\beta_j}_e\ket{t}_t,  
    \end{equation}
    and 
    \begin{equation}
        O_{\text{prod},t}: \ket{0}_f\ket{\beta_j}_e\ket{k}_t \rightarrow \ket{\beta_j t}_f\ket{\beta_jT(1-k/M)}_e\ket{k}_t.
    \end{equation}
    
For the vector $u(0)$, we assume as usual the oracle $O_u: \ket{0}_v \rightarrow \ket{u(0)}_v$ and that $\|u(0)\|$ is known. 
The difference lies in the input models for $b(t)$ which now becomes a time-dependent function. 
For an integer $M$ to be specified later, we assume a simultaneous preparation oracle $O_{b,t}: \ket{k}_t\ket{0}_v \rightarrow \ket{k}_t\ket{b(kT/M)}_v$, and the information of $\|b(t)\|$ is given in the amplitude of a quantum state through the oracle $O_{\|b\|} : \ket{0}_t \rightarrow \frac{1}{\sqrt{\sum_k \|b(kT/M)\|^2}} \sum_{k=0}^{M-1} \|b(kT/M)\|\ket{k}_t$.
Notice that the assumption on $O_{\|b\|}$ is indeed quite strong. 
Nevertheless, if $b(t)$ has a closed form expression, then $\|b(t)\|$ can be efficiently computed and integrated classically, and there exists an efficient construction of $O_{\|b\|}$~\cite{GroverRudolph2002}. 
In a special case where $\|b(t)\|$ remains unchanged, the oracle $O_{\|b\|}$ just creates a uniform superposition and can be easily constructed using Hadamard gates. 

\subsubsection{Results}

The solution of~\cref{eqn:ODE_general} is given as 
\begin{equation}
    u(T) = e^{AT} u(0) + \int_0^T e^{A(T-s)} b(s) ds. 
\end{equation}
Our approach is still to prepare both vectors as quantum states and then perform a linear combination of the state, similar as sketched in~\cref{lem:FF_LCS}. 
The first part is the same as the time-independent case and can be implemented according to~\cref{lem:FF_NSD_gen_homo}. 
However, the second part is no longer an operator acting on a vector and thus requires further treatment. 

The idea is to first discretize the integral by numerical quadrature and then apply another linear combination of states. 
In particular, we use the first-order numerical quadrature (the same as the left Riemann sum) with $M$ equi-distant nodes to get 
\begin{equation}
    \int_0^T e^{A(T-s)} b(s) ds \approx \frac{T}{M} \sum_{k=0}^{M-1} e^{A(T-kT/M)} b(kT/M). 
\end{equation}
Then the right hand side is again a linear combination of vectors and can be implemented following the idea of~\cref{lem:FF_LCS} and~\cref{lem:LCU}.

We first bound the quadrature error and estimate the scaling of $M$, which is a direct consequence of the standard quadrature error bound~\cite{BurdenNA} and the chain rule.
\begin{lem}\label{lem:FF_bt_quadrature}
    Suppose $b(t)$ is continuously differentiable. 
    Then the numerical quadrature error can be bounded as 
    \begin{equation}
    \begin{split}
        & \quad \left\|\int_0^T e^{A(T-s)} b(s) ds - \frac{T}{M} \sum_{k=0}^{M-1} e^{A(T-kT/M)} b(kT/M)\right\| \\
        & \leq \frac{T^2 e^{\widetilde{\alpha} T}}{2M} \sup_{t\in[0,T]}\left(\|A\|\|b(t)\| +  \|db(t)/dt\| \right),
    \end{split}
    \end{equation}
    where $\widetilde{\alpha} \coloneqq \max_{0\leq j \leq N-1}\left\{0, \text{Re}(\lambda_j)\right\}$. 
    As a result, in order to bound the quadrature error by $\epsilon'>0$, it suffices to choose 
    \begin{equation}
        M = \Theta\left( \frac{T^2 e^{\widetilde{\alpha} T}}{2\epsilon'} \sup_{t\in[0,T]}\left(\|A\|\|b(t)\| +  \|db(t)/dt\| \right) \right). 
    \end{equation}
\end{lem}

We are now ready to discuss the quantum implementation and its complexity for solving~\cref{eqn:ODE_general}. 
The proof is generalized from the time-independent case and presented in~\cref{app:proof_FF_bt}. 
\begin{thm}\label{thm:FF_bt}
    Consider solving the ODE system~\cref{eqn:ODE_general} up to time $T$, where $A = U\Lambda U^{\dagger}$ is a unitarily diagonalizable matrix with known eigenvalues and eigenstates, and $b(t)$ is a continuously differentiable vector-valued function. 
    Let $\widetilde{\alpha} = \max_{0\leq j \leq N-1}\left\{0,\text{Re}(\lambda_j)\right\}$ be a known parameter, and 
    \begin{equation}\label{eqn:FF_bt_M}
        M = \Theta\left( \frac{T^2 e^{\widetilde{\alpha} T}}{\epsilon\|u(T)\|} \sup_{t\in[0,T]}\left(\|A\|\|b(t)\| +  \|db(t)/dt\| \right) \right)
    \end{equation}
    Suppose we are given the oracles described in~\cref{sec:oracles_FF_NSD_gen_bt}, and the time, eigenvalue and function registers consist of $n_t,n_e$ and $n_f$ qubits, respectively. 
    Then, for any $T>0$, there exists a quantum algorithm that outputs an $\epsilon$-approximation of the normalized solution $\ket{u(T)}$ with $\Omega(1)$ success probability, using 
    \begin{enumerate}
        \item \begin{equation}
            \mathcal{O}\left( \frac{e^{\widetilde{\alpha}T} \sqrt{2\|u(0)\|^2 + 2T^2 \|b\|^2_{\text{avg}} }}{(1-\epsilon/2)\|u(T)\|} \right)
        \end{equation}
        queries to the given oracles, their inverses and controlled versions, where $\|b\|^2_{\text{avg}} \coloneqq M^{-1} \sum_{k=0}^{M-1}\|b(kT/M)\|^2$,
        \item $(n_t+n_e+n_f+2)$ ancilla qubits, 
        \item \begin{equation}
            \mathcal{O}\left( \frac{e^{\widetilde{\alpha}T} \sqrt{2\|u(0)\|^2 + 2T^2\|b\|^2_{\text{avg}}}}{(1-\epsilon/2)\|u(T)\|} \right)
        \end{equation}
        extra one-qubit gates. 
    \end{enumerate}
\end{thm}

As remarked after~\cref{thm:FF_NSD_gen}, our generalized~\cref{thm:FF_bt} can yield an $\mathcal{O}(1)$ query complexity under typical scenarios when most $\|b(t)\|$'s are $\Omega(1)$ and the overlaps between $b(t)$'s and the eigenstates of $A$ corresponding to $0$ eigenvalue are non-degenerate. 
The key feature of our algorithm to achieve fast-forwarding is the efficient implementation of the numerical quadrature, \emph{i.e.}, linear combination of quantum states.
We note that, although the number of the quadrature nodes $M$ scales badly in $T$ and $\epsilon$, it does not contribute to the overall query complexity, and only introduces $\text{poly}\log(M)$ overhead through the number of qubits needed in the time register and the constructions of the oracles $O_{b,t}$ and $O_{\|b(t)\|}$. 
Similar idea of fast quantum implementation of numerical integration has been used and is also the key reason for speedups in the interaction picture Hamiltonian simulation~\cite{LowWiebe2019} and the qHOP method for time-dependent Hamiltonian simulation~\cite{AnFangLin2022}. 

\subsection{Application: parabolic PDEs}\label{sec:app_parabolic_pde}

We apply our fast-forwarded algorithms for coefficient matrices with known eigenvalues and eigenstates to solving various high-dimensional evolutionary PDEs of parabolic type. 
We remark that our fast-forwarded algorithms can also be applied to other types of evolutionary PDEs such as hyperbolic ones and high-order ones. 
We present more applications in~\cref{app:application_PDEs}.

A standard approach for solving evolutionary PDEs is the method of lines: we first discretize all the spatial variables and transform the PDE to an ODE, then apply ODE solvers to obtain numerical solutions. 
The key observation is that many classes of PDEs involve the spatial divergence and Laplacian operators, which, after suitable spatial discretization, can be simultaneously diagonalized under the Fourier basis and the eigenvalues have closed-form expressions. This enables fast-forwarding.

We consider PDEs with time variable $t \in [0,T]$ and spatial variable $x = (x_0,\cdots,x_{d-1}) \in [0,1]^d$. 
We impose periodic boundary conditions on each spatial dimension. 
The gradient operator is defined as 
\begin{equation}\label{eqn:app_pde_gradient_def}
    \nabla = \left(\frac{\partial}{\partial x_0},\cdots, \frac{\partial}{\partial x_{d-1}}\right)^{T}, 
\end{equation}
and the spatial Laplacian operator is defined as 
\begin{equation}\label{eqn:app_pde_laplacian_def}
    \Delta  = \sum_{j=0}^{d-1} \frac{\partial^2}{\partial x_j^2}. 
\end{equation}
More generally, we also consider generalized gradient and Laplacian operators defined, for a $d$-dimensional real vector $a$, as
\begin{equation}
    \nabla^a = \left(a_0\frac{\partial}{\partial x_0},\cdots, a_{d-1}\frac{\partial}{\partial x_{d-1}}\right)^{T}, 
\end{equation}
and 
\begin{equation}
    \Delta^a  = \sum_{j=0}^{d-1} a_j \frac{\partial^2}{\partial x_j^2}. 
\end{equation}
The general class of PDEs we consider is of the form 
\begin{equation}\label{eqn:app_PDE}
\begin{split}
     \frac{\partial}{\partial t}u(x,t) &= \Delta^a u(x,t) + \nabla^{a'} \cdot u(x,t) + cu(x,t) + b(x,t), \quad t \in [0,T], x \in [0,1]^d, \\
     u(x,0) &= u_0(x). 
\end{split}
\end{equation}
Here $a,a'$ are two bounded real vectors such that all the elements of $a$ are non-negative, $c$ is a non-positive real number and $b(x,t)$ is a smooth scalar function. 
We first show how to spatially discretize the spatial divergence and Laplacian operators, and then discuss various classes of PDEs where our fast-forwarded algorithms can apply. 

\subsubsection{Spatial discretization}

We discretize space by the central difference method, with $n$ grid points in each coordinate. 
For the Laplacian, it can be approximated by the $d$-dimensional discrete Laplacian $A_L^a$, defined as
\begin{equation}  \label{eqn:dis_lap_highD_tensor}
    A_L^a = a_0 D_h \otimes I  \otimes \cdots \otimes I + a_1 I \otimes D_h \otimes I \otimes \cdots \otimes I + \cdots +  a_{d-1} I \otimes \cdots \otimes I \otimes D_h.
\end{equation}
Here $D_h$ is the one-dimensional discrete Laplacian operator defined as 
\begin{equation}\label{eqn:Dh_D2}
    D_h \coloneqq \frac{1}{h^2} \left(\begin{array}{ccccc}
        -2 & 1 & & & 1 \\
         1& -2 & 1 & & \\
          & \ddots& \ddots& \ddots& \\
           & & 1& -2 & 1\\
        1 & & & 1 & -2 \\
    \end{array}\right),
\end{equation}
with $h = 1/n$. 
The eigenvalues of $D_h$ are $\mu_k = - 4n^2 \sin^2 (k\pi/n)$, and the corresponding eigenstate is $\frac{1}{\sqrt{n}}(1,\omega_n^k,\omega_n^{2k},\cdots,\omega_n^{(n-1)k})^{T}$ where $\omega_n = e^{2\pi i/n}$.
Similarly, for the divergence operator, it can be approximated by the $d$-dimensional discrete divergence $A_G^a$, defined as 
\begin{equation}
    A_G^a \coloneqq a_0 V_h \otimes I  \otimes \cdots \otimes I + a_1 I \otimes V_h \otimes I \otimes \cdots \otimes I + \cdots +  a_{d-1} I \otimes \cdots \otimes I \otimes V_h, 
\end{equation}
where 
\begin{equation}
    V_h \coloneqq \frac{1}{2h} \left(\begin{array}{ccccc}
        0 & 1 & & & -1 \\
         -1 & 0 & 1 & & \\
          & \ddots& \ddots& \ddots& \\
           & & -1 & 0 & 1\\
        1 & & & -1 & 0 \\
    \end{array}\right),
\end{equation}
with $h = 1/n$. 
The eigenvalues of $V_h$ are $\nu_k = i n \sin(2k\pi/n)$, and the corresponding eigenstate is the same as that of $D_h$, \emph{i.e.},   $\frac{1}{\sqrt{n}}(1,\omega_n^k,\omega_n^{2k},\cdots,\omega_n^{(n-1)k})^{T}$. 
Therefore, for any two vectors $a$ and $a'$, the linear combination $A_L^a + A_G^{a'} + cI$ has the eigenbasis $F_h^{\otimes d}$, where $F_h$ denotes the discrete Fourier transform matrix and can be efficiently implemented by the (inverse) quantum Fourier transform. 
The eigenvalues of $A_L^a + A_G^{a'} + cI$ are given as
\begin{equation}\label{eqn:app_eigenvalues}
    \mu = c - 4 n^2 \sum_{j=0}^{d-1} a_j \sin^2 (k_j \pi /n) + i n \sum_{j=0}^{d-1} a_j' \sin(2k_j\pi/n), \quad k_j \in [n]. 
\end{equation}
Given \cref{eqn:app_eigenvalues}, we can use classical arithmetic to construct an oracle encoding the eigenvalues of $A_L^a + A_G^{a'}$. 
However, we remark that the gate complexity of constructing such an oracle typically scales as $\mathcal{O}(d)$, since computing the eigenvalues of $A_L^a + A_G^{a'}$ requires computing the sum of $2d$ terms, which generally cannot be expressed in a simple closed form.

\subsubsection{Fast-forwarding}

After spatial discretization, the PDE~\cref{eqn:app_PDE} can be transferred to the ODE system 
\begin{equation}\label{eqn:app_ODE}
\begin{split}
    \frac{d}{dt}\vec{u}(t) &= (A_L^a+A_G^{a'}+cI) \vec{u}(t) + \vec{b}(t), \quad t \in [0,T], \\
    \vec{u}(0) &= \vec{u}_0.
\end{split}
\end{equation}
Here, $\vec{u}(t)$ and $\vec{b}(t)$ are $n^d$-dimensional vectors containing the function values of $u(\cdot,t)$ and $b(\cdot,t)$ evaluated at all the grid points $(k_1h,k_2h,\cdots,k_dh)$ for $k_j \in [n]$, respectively. 

Notice that generic quantum ODE solvers are not efficient for solving~\cref{eqn:app_ODE} due to the large spectral norm of the coefficient matrix. 
Specifically, $\|D_h\| \sim 1/h^2$ and $\|V_h\| \sim 1/h$, and thus $\|A_L^a\| = \Theta(dn^2)$ and $\|A_G^a\| = \Theta(dn)$. 
Therefore the coefficient matrix $A_L^a+A_G^{a'}+cI$ has a spectral norm $\Theta(dn^2)$, which can be very large in high-dimensional equations or when the number of the grid points for each coordinate is large for accurate spatial discretization. 
If we apply best existing generic ODE solver as in~\cref{lem:DEsolver_TI} and~\cref{lem:DEsolver_TD}, the overall query complexity scales 
\begin{equation}\label{eqn:cor_app_PDE_existing}
    \widetilde{\mathcal{O}}\left( \frac{\sup_{t\in[0,T]} \|\vec{u}(t)\| }{\|\vec{u}(T)\|}  T d n^2 \text{~poly}\log(T/\epsilon) \right). 
\end{equation}

We can apply our fast-forwarding result, namely~\cref{thm:FF_NSD_homo} for homogeneous case and~\cref{thm:FF_bt} for general inhomogeneous case, because~\cref{eqn:app_PDE} satisfies the assumptions in~\cref{thm:FF_NSD_homo} and~\cref{thm:FF_bt}, of which the most important one is that the coefficient matrix has known eigenvalues and eigenstates as discussed before. 

\begin{cor}\label{cor:app_PDE}
    Consider solving the semi-discretized PDE~\cref{eqn:app_PDE} up to time $T$.
    Then there exists a quantum algorithm that outputs the normalized solution $\ket{\vec{u}(T)}$ with $\Omega(1)$ success probability, using 
        \begin{equation}\label{eqn:cor_app_PDE}
            \mathcal{O}\left( \frac{\sqrt{2\|\vec{u}(0)\|^2 + 2T^2 \|\vec{b}\|^2_{\text{avg}} }}{\|\vec{u}(T)\|} \right)
        \end{equation}
        queries to the preparation oracles of $\vec{u}(0)$ and $\vec{b }(t)$, and 
        \begin{equation}
            \widetilde{\mathcal{O}}\left( \frac{\sqrt{2\|\vec{u}(0)\|^2 + 2T^2 \|\vec{b}\|^2_{\text{avg}} }}{\|\vec{u}(T)\|} d \log^2(n) \text{~poly}\log(T/\epsilon) \right)
        \end{equation}
        additional gates. 
        Here $\|\vec{b}\|^2_{\text{avg}} \coloneqq M^{-1} \sum_{k=0}^{M-1}\|\vec{b}(kT/M)\|^2$. 
\end{cor}
\begin{proof}
    The number of queries to the state preparation oracles is a direct consequence of~\cref{thm:FF_bt}. 
    The additional gates are mainly due to the implementation of the oracles mentioned in~\cref{sec:oracles_FF_NSD_gen_bt} and the $d$-dimensional QFT operator $F_h^{\otimes d}$. 
    Since the eigenvalues of the matrix $A_L^a + A_G^{a'} + cI$ have closed-form expression~\cref{eqn:app_eigenvalues} containing summation of $(2d+1)$ terms, constructing oracles $O_{\Lambda,r}$, $O_{\Lambda,i}$ would require $\widetilde{\mathcal{O}}(d\text{~poly}\log(1/\epsilon))$ elementary gates by performing classical arithmetic operations with reversible computational model~\cite[Chapter 6]{RieffelPolak2011}. 
    Similarly, constructing all other oracles requires $\widetilde{\mathcal{O}}(\text{poly}\log(T/\epsilon))$ gates, and implementing QFT operator $F_h^{\otimes d}$ requires $\mathcal{O}(d \log^2(n))$ gates. 
    These together contribute to an $\widetilde{\mathcal{O}}\left( d\log^2(n)\text{~poly}\log(T/\epsilon)\right)$ multiplicative factor. 
\end{proof}

Compared to generic ODE solvers,~\cref{cor:app_PDE} avoids the dependence on the spectral norm of the coefficient matrix, and thus can be more efficient for solving high-dimensional PDEs with a large degree of freedom in spatial discretization. 
Specifically, our result gets rid of the explicit dependence on $n$. 
Our result also gets rid of the explicit dependence on $d$ in the number of state preparation oracles, though the number of additional elementary gates required by other oracles and Fourier transform still scales linearly in $d$. 

As in the norm dependence we use slightly different parameters compared to existing algorithms, here we present more quantitative estimates to showcase the scenarios where our algorithm can be advantageous. 
First, if $\|\vec{b}\|_{\text{avg}} = \mathcal{O}(\|\vec{u}(0)\|)$, then the complexity of our algorithm~\cref{eqn:cor_app_PDE} becomes 
\begin{equation}
    \mathcal{O}\left(\frac{\|\vec{u}(0)\|}{\|\vec{u}(T)\|} T \right)
\end{equation}
while the previously best complexity estimate~\cref{eqn:cor_app_PDE_existing} is 
\begin{equation}
    \widetilde{\mathcal{O}}\left( \frac{\sup_{t\in[0,T]} \|\vec{u}(t)\| }{\|\vec{u}(T)\|}  T d n^2 \text{~poly}\log(T/\epsilon) \right) \geq \widetilde{\mathcal{O}}\left( \frac{ \|\vec{u}(0)\| }{\|\vec{u}(T)\|}  T d n^2 \text{~poly}\log(T/\epsilon) \right). 
\end{equation}
In this case, our algorithm is always better by getting rid of explicit dependence on $d,n$ and $\epsilon$. 
Notice that in the examples that we will present soon, the continuous version of the assumption $\|\vec{b}\|_{\text{avg}} = \mathcal{O}(\|\vec{u}(0)\|)$ means that the magnitude of the external source does not asymptotically exceed that of the original energy of the system. 

Similar to the discussions after~\cref{thm:FF_square_root_inhomo}, we may obtain further speedup in $T$ if $c = 0$ and both $\vec{u}(0)$ and $\vec{b}(t)$ has non-trivial overlap with the eigenstate of $A_L^a+A_G^{a'}$ corresponding to the $0$ eigenvalue. 
Specifically, let $(\lambda_j,\ket{\psi_j})$ be the eigenvalues and eigenstates of $A_L^a+A_G^{a'}$ and denote $\lambda_0 = 0$. 
Suppose that $\braket{\psi_0|\vec{u}(0)}$ and $\braket{\psi_0|\vec{b}(s)}$ are $\Omega(1)$. 
Then 
\begin{equation}
    \begin{split}
        \vec{u}(T) & =  e^{(A_L^a+A_G^{a'})T} \vec{u}(0) + \int_0^T e^{(A_L^a+A_G^{a'})(T-s)} \vec{b}(s) ds \\
        & = \|\vec{u}(0)\| \sum_j e^{\lambda_j T} \braket{\psi_j|\vec{u}(0)} \ket{\psi_j} + \sum_j\int_0^T \|\vec{b}(s)\| e^{\lambda_j (T-s)} \braket{\psi_j|\vec{b}(s)} \ket{\psi_j} ds \\
        & = \left(\|\vec{u}(0)\| \braket{\psi_0|\vec{u}(0)}  + \int_0^T \|\vec{b}(s)\|  \braket{\psi_0|\vec{b}(s)} ds \right) \ket{\psi_0}\\
        & \quad\quad + \|\vec{u}(0)\| \sum_{j>0} e^{\lambda_j T} \braket{\psi_j|\vec{u}(0)} \ket{\psi_j} + \sum_{j>0}\int_0^T \|\vec{b}(s)\| e^{\lambda_j (T-s)} \braket{\psi_j|\vec{b}(s)} \ket{\psi_j} ds. 
    \end{split}
\end{equation}
Once $\braket{\psi_0|\vec{u}(0)}$ and $\braket{\psi_0|\vec{b}(s)}$ are $\Omega(1)$, the amplitude of $\vec{u}(T)$ is $\Omega(\|\vec{u}(0)\| + T\|\vec{b}\|_{\text{avg}})$, and the complexity of our fast-forwarded algorithm is simply $\mathcal{O}(1)$. 
As a comparison, as $\sup_t\|\vec{u}(t)\| \geq \|\vec{u}(T)\|$, the scaling of the generic ODE solver even under this assumption is still $\mathcal{O}(T dn^2 \text{poly}\log(T/\epsilon))$. 
So our algorithm is also better in terms of $T$, besides $d,n$ and $\epsilon$. 
We remark that in the applications of heat equation and advection-diffusion equation being discussed soon, the assumptions are naturally satisfied in general scenarios. 
This is because, in the continuous analog, such conditions mean that $u(x,0)$ and $b(x,t)$ have a non-trivial overlap with constant functions, since the eigenfunction of the Laplacian and divergence operators corresponding to their $0$ eigenvalue is the constant function. 
Physically, they correspond to the non-zero total heat/energy of the initial system and the external source. 

\subsubsection{Examples}

\paragraph{Transport equation}
The transport equation describes phenomena where a conserved quantity, such as mass, heat, energy, momentum, or electric charge is transported in a space. Transport equations play a fundamental role in many applied areas, from wave propagation, electromagnetism, particle physics, nuclear reaction, to computer vision, medical imaging, and semiconductor conducting \cite{case1967linear,chandrasekhar2013radiative,markowich2012semiconductor,ryzhik1996transport}. The equation has the form
\begin{equation}\label{eqn:app_transport}
    \frac{\partial}{\partial t} u(x,t) = \nabla^{a'} \cdot u(x,t) + b(x,t). 
\end{equation}
\cref{eqn:app_transport} is a special form of~\cref{eqn:app_PDE} with $a = 0$ and $c = 0$, so can be fast-forwarded using~\cref{cor:app_PDE}. 
We remark that the coefficient matrix of~\cref{eqn:app_transport} after spatial discretization has purely imaginary eigenvalues, and thus the homogeneous version (\emph{i.e.}, $b \equiv 0$) is a Hamiltonian simulation problem and can be fast-forwarded without any exponential dependence in time related to the norm decay. 

A generalized version of the transport equation is the Liouville's equation, which is used to model the evolution of the distribution of particles in a phase space. Josiah Willard Gibbs was the first to recognize the importance of this equation as the fundamental equation of statistical mechanics, with connections to Hamiltonian dynamical systems, symplectic geometry, and ergodic theory, as well as various applications to kinetic theory of gases, atmospheric dynamics, astronomy, and thermodynamics \cite{muller2022basics,kubo1963stochastic,gibbs1884fundamental}. The equation has the form
\begin{equation}\label{eqn:app_Liouville}
    \frac{\partial}{\partial t} u(x,t) = \sum_{j=0}^{d-1} \frac{\partial}{\partial x_j} (a'_j(x,t) u(x,t)) + b(x,t). 
\end{equation}
However, our fast-forwarded algorithm requires the coefficient matrix to be constant, so we can only deal with special cases of~\cref{eqn:app_Liouville} with constant functions $a_j$'s. 
In these cases, the Liouville's equation degenerates to the transport equation. 

\paragraph{Heat equation} The heat equation, which is used to describe the conduction of heat energy. Here we consider the inhomogeneity as an external source of heat. As the prototypical parabolic PDE, the heat equation is as fundamental to the broader field of science and engineering, such as the connection with random walk and Brownian motion in probability theory, the Black-Scholes option pricing model in finance, the formulation of hydrodynamical shocks in fluid dynamics, and the edge detection model in image analysis \cite{evans2010partial,widder1976heat,cannon1984one,lawler2010random,thambynayagam2011diffusion}. The Schr\"{o}dinger equation of quantum mechanics can also be regarded as a heat equation in imaginary time. The equation has the form
\begin{equation}
    \frac{\partial}{\partial t} u(x,t) = \Delta u(x,t) + b(x,t). 
\end{equation}
This can be fast-forwarded using~\cref{cor:app_PDE} since it is a special case of~\cref{eqn:app_PDE} where all the entries of $a$ are equal, $a' = 0$ and $c = 0$. 

\paragraph{Advection-diffusion equation} The advection-diffusion equation is a combination of the heat and transport equations, and describes physical phenomena where physical quantities are transferred inside a system due to two processes: diffusion and advection. The model deals with the time evolution of chemical or biological species in water or air, exhibiting rich phenomena in chemical processes such as combustion, biological and ecological networks, and population growth \cite{evans2010partial,hundsdorfer2003numerical,perthame2012growth}. The equation has the form
\begin{equation}
     \frac{\partial}{\partial t}u(x,t) = a \Delta u(x,t) + \nabla^{a'} \cdot u(x,t) + b(x,t). 
\end{equation}
This can be fast-forwarded using~\cref{cor:app_PDE} since it is a special case of~\cref{eqn:app_PDE} where (with an abuse of notations) all the entries of $a$ are equal and $c = 0$. 

\section{Conclusion and outlook}\label{sec:conclusion}

This work investigates quantum algorithms for solving non-quantum dynamics from two opposite directions: their limitations when applied to general non-quantum dynamics and their fast-forwarding when applied to specific non-quantum dynamics. On the one hand, by proving worst-case lower bounds, we show that generic quantum algorithms for solving ODEs suffer from computational overheads due to their ``non-quantumness'', and conclude that quantum algorithms for quantum dynamics work best. On the other hand, for several specific but practically useful ODE systems, we design new ``one-shot'' algorithms to obtain exponential improvements compared to the best existing generic algorithms. Remarkably, our fast-forwarded algorithm can be applied to various high-dimensional linear evolutionary PDEs. 

Our lower bounds on generic quantum ODE solvers are in terms of the number of queries to the preparation oracle of the initial state. For inhomogeneous ODEs, it is not hard to prove the same lower bounds in terms of the number of queries to the preparation oracle of $b$, since the operator $\int_0^T e^{A(T-s)} ds$ is similar to $e^{AT}$ and can be regarded as an amplifier in the same subspace. 
It remains an interesting open question to lower bound the number of queries to the coefficient matrix $A$. We also note that our lower bounds are worst-case, meaning that we proved them by allowing the initial state $u(0)$ and/or the coefficient matrix $A$ to take their worst possible values. It would be interesting to prove lower bounds in the average case as they may be more practically relevant and a better bound may be attained. For example, one concrete direction is to generalize~\cref{prop:lb_non_normal_homo} and~\cref{prop:lb_non_normal_inhomo} to the case when the coefficient matrix $A$ is average-case or even arbitrary.

To prove our lower bounds, we propose a general framework that establishes lower bounds on any quantum algorithm that can be regarded as an amplifier. We believe that this framework can serve as a general tool for investigating the power of quantum algorithms for solving non-quantum problems, since they may embed non-unitary transformations and thus can be regarded as amplifiers. Our framework is closely related to that of quantum state discrimination but differs from it in two ways: 1.~we make the stronger assumptions that we can access the preparation oracle and its inverse whereas only access to copies of the unknown state is allowed in the standard setup, and 2.~we obtain a different lower bound of $\Omega(1/\sqrt{\epsilon})$ queries compared to  the standard lower bound of $\Omega(1/\epsilon)$ copies. Our result motivates the renewed study of the foundational quantum state discrimination problem but under different input access assumptions, and the investigation into why different assumptions lead to different complexities.

For several specific ODE systems, we exponentially improve over the efficiencies of the best existing generic quantum algorithms. 
A natural future direction is to identify more ODE systems that can be fast-forwarded. 
For example, as our algorithms assume $A$ itself has known eigensystem, it would be interesting to investigate whether fast-forwarding can still be achieved if $A = A_1 + A_2$ where both $A_1$ and $A_2$ have known eigensystem but do not commute with each other. 
Such a system represents evolutionary PDEs with potential terms. 

Our algorithms for solving ODEs with a time-independent inhomogeneous term that achieve fast-forwarding are one-shot in the sense that they neither require time discretization nor solving high-dimensional linear systems. Instead, they directly map the input state to the final solution, which we believe could be a more straightforward and intuitive way of designing generic quantum ODE solvers, irrespective of whether fast-forwarding is possible. For example, for ODEs with a time-independent inhomogeneous term $b$, we can separately implement block-encodings of $e^{AT}$ and $\int_0^T e^{A(T-s)} ds$ using the contour integral representation~\cite{TakahiraOhashiSogabeEtAl2020,TongAnWiebe2021,FangLinTong2022} $f(A) = \int_{\Gamma} f(z) (z-A)^{-1} dz$ for suitable choices of $f$. 
The contour integral can then be discretized using sufficiently many nodes and be implemented using LCU. 
We leave the rigorous analysis of such an algorithm to future work.

\section*{Acknowledgments} 

DA and QZ acknowledge funding from Innovation Program for Quantum Science and Technology via Project 2024ZD0301900. DA acknowledges the support by The Fundamental Research Funds for the Central Universities, Peking University, and the Department of Defense through the Hartree Postdoctoral Fellowship at QuICS. JPL acknowledges the support by the National Science Foundation (grant CCF-1813814, PHY-1818914), an NSF Quantum Information Science and Engineering Network (QISE-NET) triplet award (DMR-1747426), an NSF Quantum Leap Challenge Institute (QLCI) program (OMA-2016245), a Simons Foundation award (No. 825053), and the Simons Quantum Postdoctoral Fellowship. DW acknowledges support from the Army Research Office (grant W911NF-20-1-0015); the Department of Energy, Office of Science, Office of Advanced Scientific Computing Research, Accelerated Research in Quantum Computing program; and an NSF QISE-NET triplet award (DMR-1747426). QZ acknowledges funding from National Natural Science Foundation of China (NSFC) via Project No. 12347104 and No. 12305030, Guangdong Basic and Applied Basic Research Foundation via Project 2023A1515012185, Hong Kong Research Grant Council (RGC) via No. 27300823, N\_HKU718/23, and R6010-23, Guangdong Provincial Quantum Science Strategic Initiative No. GDZX2303007, HKU Seed Fund for Basic Research for New Staff via Project 2201100596. Part of this work was completed while JPL and QZ were affiliated with QuICS. We thank Andrew Childs, Di Fang, Lin Lin, Yi-Kai Liu, Yu Tong, and Jiasu Wang for helpful discussions. 

\section*{Competing interests} 

The authors have no competing interests to declare that are relevant to the content of this article. 

\section*{Data availability statement} 

Data sharing not applicable to this article as no datasets were generated or analysed during the current study. 


\bibliographystyle{unsrt}
\bibliography{LB-FF}

\clearpage
\appendix

\section{Linear algebra results and quantum linear algebra operations}\label{app:proof_la_lemma}

\subsection{Linear algebra results}

We discuss several linear algebra results that will be frequently used in our analysis. 
The first one is useful for bounding the error and the success probability related to quantum states from those related to unnormalized vectors. 
The second one is the relations of the overlaps between quantum states or the approximate versions. 
The third one is about the relations among trace distance, fidelity and $2$-norm for pure states. 

\begin{lem}\label{lem:succ_prob_error}
    For any two non-zero vectors $a,\widetilde{a}$ such that $\|a-\widetilde{a}\| \leq \epsilon$, we have 
    \begin{enumerate}
        \item $\|\widetilde{a}\| \geq \|a\| - \epsilon $, 
        \item $\|a/\|a\| - \widetilde{a}/\|\widetilde{a}\|\| \leq 2\epsilon/\|a\|$. 
    \end{enumerate}
\end{lem}

\begin{proof}
    The first claim can be proved by the triangle inequality that 
    \begin{equation}
        \|\widetilde{a}\| \geq \|a\| - \|\widetilde{a} - a\| \geq \|a\| - \epsilon. 
    \end{equation}
    The second claim can be proved as 
    \begin{equation}
        \begin{split}
            \left\| \frac{a}{\|a\|} - \frac{\widetilde{a}}{\|\widetilde{a}\|}  \right\| 
            & \leq \left\| \frac{a}{\|a\|} - \frac{\widetilde{a}}{\|a\|}  \right\| + \left\| \frac{\widetilde{a}}{\|a\|} - \frac{\widetilde{a}}{\|\widetilde{a}\|}  \right\| \\
            & = \frac{\|a-\widetilde{a}\|}{\|a\|} + \|\widetilde{a}\|\left|\frac{1}{\|a\|} - \frac{1}{\|\widetilde{a}\|}\right| \\
            & = \frac{\|a-\widetilde{a}\|}{\|a\|} + \frac{|\|\widetilde{a}\|-\|a\||}{\|a\|} \\
            & \leq \frac{2\epsilon}{\|a\|}. 
        \end{split}
    \end{equation}
\end{proof}

\begin{lem}\label{lem:fed_2norm}
    Let $\ket{\psi}$, $\ket{\phi}$, $\ket{\widetilde{\psi}}$ and $\ket{\widetilde{\phi}}$ be quantum states. 
    Then 
    \begin{equation}
        |\braket{\widetilde{\psi}|\widetilde{\phi}}| \leq 
          |\braket{{\psi}|{\phi}}| + \|\ket{\widetilde{\psi}} - \ket{\psi}\| + \|\ket{\widetilde{\phi}} - \ket{\phi}\|. 
    \end{equation}
\end{lem}

\begin{proof}
    Using triangle inequality and Cauchy-Schwarz inequality, we have 
    \begin{equation}
        \begin{split}
            |\braket{\widetilde{\psi}|{\widetilde{\phi}}}| &= |\braket{{\psi}|{\phi}} + \bra{\psi}(\ket{\widetilde{\phi}} - \ket{\phi}) + (\bra{\widetilde{\psi}} - \bra{\psi})\ket{\widetilde{\phi}} | \\
            & \leq |\braket{{\psi}|{\phi}} | + |\bra{\psi}(\ket{\widetilde{\phi}} - \ket{\phi})| + | (\bra{\widetilde{\psi}} - \bra{\psi})\ket{\widetilde{\phi}} | \\
            & \leq |\braket{{\psi}|{\phi}} | + \|\ket{\widetilde{\phi}} - \ket{\phi})\| + \| \ket{\widetilde{\psi}} - \ket{\psi}\|. 
        \end{split}
    \end{equation}
\end{proof}

\begin{lem}\label{lem:trace_distance}
    For two pure states $\ket{\psi}$ and $\ket{\phi}$, we have 
    \begin{equation}
        \frac{1}{2} \| \ket{\psi}\bra{\psi} - \ket{\phi}\bra{\phi} \|_{1} = \sqrt{1 - |\braket{\psi|\phi}|^2} \leq \|\ket{\psi} - \ket{\phi}\|. 
    \end{equation}
\end{lem}

\begin{proof}[Proof of~\cref{lem:trace_distance}]
    The first equality is the standard relation between the trace distance and the fidelity for pure states, and the proof can be found in~\cite{NielsenChuang2000}. 
    For the second, it can be proved as 
    \begin{equation}
    \begin{split}
        \|\ket{\psi} - \ket{\phi}\|^2 &= (\bra{\psi}-\bra{\phi}) (\ket{\psi} - \ket{\phi}) \\
        &= 2 - 2\text{Re}(\braket{\psi|\phi}) \\
        &= | \braket{\psi|\phi} - 1|^2 + 1 - |\braket{\psi|\phi}|^2 \\
        & \geq 1 - |\braket{\psi|\phi}|^2. 
    \end{split}
    \end{equation}
\end{proof}

\subsection{Block-encoding}

Block-encoding is a powerful technique to represent arbitrary matrices with unitary matrices. 
For a possibly non-unitary matrix $A$ and a parameter $\alpha \geq \|A\|$, the intuitive idea of block-encoding is to construct a unitary $U$ with higher dimension such that $A$ appears in its upper-left block, 
\begin{equation}
    U = \left(\begin{array}{cc}
        A/\alpha & * \\
        * & *
    \end{array}\right). 
\end{equation}
A formal definition with the presence of errors is as follows~\cite{GilyenSuLowEtAl2019}. 

\begin{defn}[Block-encoding]\label{def:block_encoding}
    Suppose $A$ is a $2^n$-dimensional matrix such that $\|A\| \leq \alpha$. 
    Then a $2^{n+n_a}$-dimensional unitary matrix $U$ is an $(\alpha,n_a,\epsilon)$-block-encoding of $A$, if 
    \begin{equation}
        \|A - \alpha (\bra{0}_a \otimes I)U (\ket{0}_a \otimes I)\| \leq \epsilon. 
    \end{equation}
\end{defn}

Block-encoding allows us to encode a general matrix on quantum computers and then perform matrix operations, including summation, multiplication, inverse, and polynomial transformation. 
Here we mostly follow~\cite{GilyenSuLowEtAl2019} again and briefly summarize several properties that will be used in our work. 

\begin{defn}[State preparation pair]\label{def:state_preparation_pair}
    Let $y$ be a $2^n$-dimensional vector and $\|y\|_1 \leq \beta$. 
    Then a pair of unitaries $(P_L,P_R)$ is a $(\beta,n,0)$-state-preparation-pair, if $P_L \ket{0} = \sum_{j=0}^{2^n-1} c_j \ket{j}$, $P_R \ket{0} = \sum_{j=0}^{2^n-1} d_j \ket{j}$ and $\beta c_j^* d_j = y_j$ for any $0 \leq j \leq 2^n-1$. 
\end{defn}

\begin{lem}[Linear combination of block-encodings{~\cite[Lemma 52]{GilyenSuLowEtAl2019}}]\label{lem:LCU}
    Let $A = \sum_{j=0}^{2^k-1} y_j A_j$ be a $2^n$-dimensional matrix. 
    Suppose that $(P_L,P_R)$ is a $(\beta,k,0)$-state-preparation-pair of $y$, and $W = \sum_{j=0}^{2^k-1} \ket{j}\bra{j} \otimes U_j$ is a $2^{n+n_a+k}$-dimensional unitary matrix such that $U_j$ is an $(\alpha,n_a,\epsilon)$-block-encoding of $A_j$. 
    Then $(P_L^{\dagger}\otimes I_{n_a} \otimes I_n) W (P_R\otimes I_{n_a} \otimes I_n)$ is an $(\alpha\beta,n_a+k,\alpha\beta\epsilon)$-block-encoding of $A$, with a single use of $W,P_L^{\dagger}$ and $P_R$. 
\end{lem}

\begin{lem}[Multiplication of block-encodings{~\cite[Lemma 53]{GilyenSuLowEtAl2019}}]\label{lem:multiplication_block_encoding}
    Let $A,B$ be $2^n$-dimensional matrices, $U_A$ be an $(\alpha,n_a,\delta)$-block-encoding of $A$, and $U_B$ be a $(\beta,n_b,\epsilon)$-block-encoding of $B$. 
    Then $(I_{n_b}\otimes U_A)(I_{n_a} \otimes U_B)$ is an $(\alpha\beta,n_a+n_b, \alpha\epsilon+\beta\delta)$-block-encoding of $AB$, with a single use of $U_A$ and $U_B$. 
\end{lem}

\begin{lem}[Inverse of a block-encoding{~\cite[Appendix B]{TongAnWiebe2021}}]\label{lem:inverse_block_encoding}
    Suppose $A$ is a $2^n$-dimensional invertible Hermitian matrix such that all the eigenvalues are within $[-1,-\delta]\cup [\delta,1]$, and $U_A$ is a $(1,n_A,0)$-block-encoding of $A$. 
    Then a $(4/(3\delta), n_A+1, \epsilon)$-block-encoding of $A^{-1}$ can be constructed, using $\mathcal{O}((1/\delta) \log(1/(\delta\epsilon)))$ queries to $U_A$ and its inverse. 
\end{lem}

\begin{lem}[Polynomial of a block-encoding{~\cite[Theorem 56]{GilyenSuLowEtAl2019}}]\label{lem:poly_block_encoding}
    Let $A$ be a $2^n$-dimensional Hermitian matrix, and $U_A$ is an $(\alpha,n_A,\epsilon)$-block-encoding of $A$. 
    If $P(x)$ is a degree-$d$ real polynomial such that $|P(x)| \leq 1/2$ for all $x\in[-1,1]$, then a $(1,n_A+2,4d\sqrt{\epsilon/\alpha})$-block-encoding of $P(A/\alpha)$ can be constructed using $d$ applications of $U_A$ and $U_A^{\dagger}$, a single application of controlled $U$ and $\mathcal{O}((n_A+1)d)$ additional one- and two-qubit gates. 
\end{lem}

\section{Lower bound for solving linear systems of equations}\label{app:other_lower_bounds}

We show how to use~\cref{thm:state_discrimination} to study lower bounds for quantum linear system problem. 
The goal of the quantum linear system problem is to prepare an approximation of the normalized solution $\ket{x} = A^{-1}\ket{b}/\|A^{-1}\ket{b}\|$ for an invertible matrix $A$ with $\|A\|=1$ and a state $\ket{b}$. 
We will show that, for any invertible matrix $A$, generic quantum linear system solves with bounded error must take $\Omega(\kappa)$ queries to the preparation oracle of $\ket{b}$ or its inverse in the worst case, where $\kappa = \|A\|\|A^{-1}\|$ is the condition number of $A$. 

Let $A = UDV^{\dagger}$ be the singular value decomposition.
Here $U = (\ket{u_1},\cdots,\ket{u_N})$ and $V = (\ket{v_1},\cdots,\ket{v_N})$ are two unitary matrices with column vectors $\ket{u_j}$ and $\ket{v_j}$, respectively. 
$D = \text{diag}(d_1,\cdots,d_N)$ is the diagonal matrix containing singular values, and without loss of generality we assume $1 = d_1 \geq d_2 \geq \cdots \geq d_N = 1/\kappa$. 
Note that the smallest singular value must be the inverse of the condition number due to the definition of the condition number. 

Consider two quantum states 
\begin{equation}
    \ket{b_1} = \ket{u_1}, \quad \ket{b_2} = \sqrt{1-1/\kappa^2} \ket{u_1} + (1/\kappa) \ket{u_N}, 
\end{equation}
and let 
\begin{equation}
    \ket{x_1} = \frac{A^{-1}\ket{b_1}}{\|A^{-1}\ket{b_1}\|}, \quad \ket{x_2} = \frac{A^{-1}\ket{b_2}}{\|A^{-1}\ket{b_2}\|}. 
\end{equation}
On the one hand, 
\begin{equation}\label{eqn:QLSP_LB_ini}
    \braket{b_1|b_2} = \sqrt{1-1/\kappa^2} \geq 1 - 1/\kappa^2. 
\end{equation}
On the other hand, using $A^{-1} = VD^{-1}U^{\dagger}$, for any $1 \leq j \leq N$, we can compute 
\begin{equation}
    \begin{split}
        A^{-1} \ket{u_j} = VD^{-1}U^{\dagger}\ket{u_j} = V D^{-1} \ket{e_j} = V (d_j^{-1}\ket{e_j}) = d_j^{-1} \ket{v_j}, 
    \end{split}
\end{equation}
where $\ket{e_j}$ is an $N$-dimensional vector with $1$ on its $j$-th entry and $0$ on all other entries. 
Then 
\begin{equation}
    A^{-1} \ket{b_1} = \ket{v_1}, \quad A^{-1} \ket{b_2} = \sqrt{1-1/\kappa^2} \ket{v_1} + \ket{v_N}, 
\end{equation}
and thus 
\begin{equation}
    \begin{split}
        \braket{x_1|x_2} = \frac{(A^{-1}\ket{b_1})^*(A^{-1}\ket{b_2}) }{\|A^{-1}\ket{b_1}\|\|A^{-1}\ket{b_2}\|} = \frac{\sqrt{1-1/\kappa^2}}{\sqrt{2-1/\kappa^2}} \leq \frac{1}{\sqrt{2}}. 
    \end{split}
\end{equation}

Let $\ket{\widetilde{x}_1}$ and $\ket{\widetilde{x}_2}$ be the outputs of a generic quantum linear system solver that approximate $\ket{x_1}$ and $\ket{x_2}$ with 2-norm error at most $1/10$. 
Then, according to~\cref{lem:fed_2norm} and~\cref{lem:trace_distance}, the fidelity between $\ket{\widetilde{x}_1}$ and $\ket{\widetilde{x}_2}$ is at most $1/\sqrt{2}+1/10+1/10 < 0.91$, and the trace distance between $\ket{\widetilde{x}_1}$ and $\ket{\widetilde{x}_2}$ is at least $0.41$. 
Therefore, a generic quantum linear system solver can be viewed as the amplifier for the state pair $\ket{b_1}$ and $\ket{b_2}$ with overlap $\geq 1-1/\kappa^2$, and~\cref{thm:state_discrimination} directly implies an $\Omega(\kappa)$ lower bound on the number of querying $O_b$, its inverse and controlled version. 

Our lower bound recovers the well-known $\Omega(\kappa)$ scaling for quantum linear system solvers. 
Notice that existing ones in \emph{e.g.},~\cite{HarrowHassidimLloyd2009,OrsucciDunjko2021} require both the matrix $A$ and the vector $b$ to be worst-case, and obtain lower bounds on queries to $A$ and $b$. 
As a comparison, our result holds for arbitrary matrix $A$ and only requires worst-case vector $b$, but our lower bound is in terms of the preparation oracle of $b$ and no lower bound is obtained on the input model of $A$. 
Our result also suggests $\Omega(\kappa)$ lower bound for solving linear system of equations with positive-definite matrix $A$, since no condition on $A$ is assumed in our result. 
There are two slight differences compared to~\cite{SommaSubasi2021}: we consider the algorithm for solving the linear system of equations deterministically while~\cite{SommaSubasi2021} considers the state verification problem with possible randomness, and our result applies to general matrices $A$ while the state pair constructed in~\cite[Appendix A]{SommaSubasi2021} applies to normal matrices. 

\section{Lower bounds for solving inhomogeneous ODEs}\label{app:proofs_LB_inhomo}


Both~\cref{prop:lb_eig_diff_homo} and~\cref{prop:lb_non_normal_homo} can be generalized to the inhomogeneous case using similar analysis and constructions. 

\begin{prop}\label{prop:lb_eig_diff_inhomo}
    Consider the inhomogeneous ODE problem with a diagonalizable matrix $A = VDV^{-1}$ where $D = \text{diag}(\lambda_1,\cdots, \lambda_N)$ is a diagonal matrix and $V$ is an invertible matrix. 
    Then, there is no generic quantum algorithm that can prepare $u(T)/\|u(T)\|$ with bounded error and failure probability, using $o(e^{\gamma T}/(T+1+\sqrt{2}))$ queries to the preparation oracle of $\ket{u(0)}$, its inverse or controlled versions, where $\gamma = \min\left\{\max_j \text{Re}(\lambda_j), \max_{i,j}|\text{Re}(\lambda_i-\lambda_j)| \right\}$. 
\end{prop}

\begin{proof}
    The proof is similar to that of~\cref{prop:lb_eig_diff_homo}. 
    Let $\lambda_j = \alpha_j + i \beta_j$ where $\alpha_j$ and $\beta_j$ are the real part and the imaginary part of $\lambda_j$, respectively. 
    Without loss of generality, assume that 
    \begin{equation}
        \max_{j} \text{Re}(\lambda_j) = \text{Re}(\lambda_1) = \alpha_1
    \end{equation}
    and 
    \begin{equation}
        \min_{j} \text{Re}(\lambda_j) = \text{Re}(\lambda_2) = \alpha_2. 
    \end{equation}
    Then 
    \begin{equation}
        \max_{i,j}\text{Re}(\lambda_i-\lambda_j) = \alpha_1 - \alpha_2. 
    \end{equation}
    Our claim is trivial if $\alpha_1 \leq 0$ or $\alpha_1 - \alpha_2 = 0$, so we assume that $\alpha_1 > 0$ and $\alpha_1 - \alpha_2 > 0$.  
    Let $V = (v_1,\cdots,v_N)$ where each $v_j$ is the eigenvector of $A$ corresponding to the eigenvalue $\lambda_j$ and are normalized such that $\|v_j\| = 1$. 
    Without loss of generality, we assume that $\braket{v_1|v_2}$ is real, otherwise we may rotate and redefine $v_2$ without changing $A$. 
    
    For any $0 < \epsilon < 1$, we consider solving~\cref{eqn:ODE} with 
    \begin{equation}
        b = v_2, 
    \end{equation}
    and two possible initial conditions 
    \begin{equation}
        \begin{split}
            u(0) &= v_2,\\
            w(0) &= \sqrt{\epsilon} v_1 + \xi v_2. 
        \end{split}
    \end{equation}
    Here $\xi$ is chosen to be a real number such that $\|w(0)\| = 1$, \emph{i.e.}, 
    \begin{equation}
        1 = |\braket{w(0)|w(0)}| = \epsilon + |\xi|^2 + 2\sqrt{\epsilon} \xi \braket{v_1|v_2}. 
    \end{equation}
    The solutions $u(T)$ and $w(T)$ can be solved as follows. 
    As proved in~\cref{prop:lb_eig_diff_homo}, for any $t$ and $j$ we have 
    \begin{equation}
        e^{At} v_j = e^{\lambda_j t} v_j. 
    \end{equation}
    Therefore, according to~\cref{eqn:ODE_solu}
    \begin{equation}
        u(T) = e^{\lambda_2 T} v_2 + \int_0^T e^{\lambda_2 (T-s)} v_2 ds \sim v_2, 
    \end{equation}
    and 
    \begin{equation}
        w(T) = \sqrt{\epsilon} e^{\lambda_1 T} v_1 + \xi e^{\lambda_2 T} v_2 + \int_0^T e^{\lambda_2 (T-s)} v_2 ds \coloneqq c_1 v_1 + c_2 v_2, 
    \end{equation}
    where 
    \begin{equation}\label{eqn:def_c_inhomo}
        c_1 = \sqrt{\epsilon} e^{\lambda_1 T}, \quad c_2 = \xi e^{\lambda_2 T} + \int_0^T e^{\lambda_2 (T-s)} ds. 
    \end{equation}
    
    We now compute the fidelity of the input states and the output states. 
    For the input states, the same as~\cref{eqn:ini_fed_eig_diff}, we have 
    \begin{equation}\label{eqn:ini_fed_eig_diff_inhomo}
        \begin{split}
            \braket{u(0)|w(0)}  > 1-\epsilon. 
        \end{split}
    \end{equation}
    For the output states, we have 
    \begin{equation}
        \begin{split}
            |\braket{u(T)|w(T)}|^2 &= \frac{  w(T)^{\dagger} u(T)  u(T) ^{\dagger }w(T) }{\|u(T)\|^2 \|w(T)\|^2} \\
            &= \frac{(c_1v_1+c_2v_2)^{\dagger} v_2 v_2^{\dagger}(c_1v_1+c_2v_2) }{\|c_1v_1 + c_2v_2\|^2} \\
            & = \frac{ |c_1|^2|\braket{v_1|v_2}|^2 + |c_2|^2 + 2\text{Re}(\bar{c}_1c_2 \braket{v_1|v_2}) }{|c_1|^2 + |c_2|^2 + 2\text{Re}(\bar{c}_1c_2 \braket{v_1|v_2})}. 
        \end{split}
    \end{equation}
    Notice that 
    \begin{equation}
        |2\text{Re}(\bar{c}_1c_2 \braket{v_1|v_2})| \leq 2 |\bar{c}_1c_2 \braket{v_1|v_2}| \leq |c_1|^2|\braket{v_1|v_2}|^2 + |c_2|^2, 
    \end{equation}
    and that function $f(x) = \frac{a+x}{b+x}$ with $b > a > 0$ is monotonically increasing for $x \geq -a$, we can further bound the final fidelity as 
    \begin{equation}
        |\braket{u(T)|w(T)}|^2 \leq \frac{ 2|c_1|^2|\braket{v_1|v_2}|^2 + 2|c_2|^2  }{|c_1|^2 (1+|\braket{v_1|v_2}|^2) + 2|c_2|^2 }. 
    \end{equation}
    According to~\cref{eqn:def_c_inhomo} and~\cref{eqn:bound_xi}, we can bound $|c_2|$ as 
    \begin{equation}
        \begin{split}
            |c_2| &\leq |\xi| e^{\alpha_2 T} + \int_0^T e^{\alpha_2 (T-s)} ds \\
            & \leq |\xi| e^{\max\{0,\alpha_2\} T} + \int_0^T e^{\max\{0,\alpha_2\} (T-s)} ds \\
            & \leq |\xi| e^{\max\{0,\alpha_2\} T} + \int_0^T e^{\max\{0,\alpha_2\} T } ds \\
            & \leq (1 + \sqrt{2} + T) e^{\max\{0,\alpha_2\} T}, 
        \end{split}
    \end{equation}
   and thus 
   \begin{equation}
   \begin{split}
       |\braket{u(T)|w(T)}|^2 &\leq \frac{ 2|c_1|^2|\braket{v_1|v_2}|^2 + 2(1 + \sqrt{2} + T)^2 e^{2 T \max\{0,\alpha_2\} }  }{|c_1|^2 (1+|\braket{v_1|v_2}|^2) + 2(1 + \sqrt{2} + T)^2 e^{2 T \max\{0,\alpha_2\}} } \\
       & = \frac{ 2 \epsilon e^{2T \alpha_1 } |\braket{v_1|v_2}|^2 + 2(1 + \sqrt{2} + T)^2 e^{2 T \max\{0,\alpha_2\} }  }{\epsilon e^{2T \alpha_1 } (1+|\braket{v_1|v_2}|^2) + 2(1 + \sqrt{2} + T)^2 e^{2 T \max\{0,\alpha_2\}} } \\
       & = \frac{ 2 \epsilon e^{2T \min\{\alpha_1,\alpha_1-\alpha_2\} } |\braket{v_1|v_2}|^2 + 2(1 + \sqrt{2} + T)^2   }{\epsilon e^{2T \min\{\alpha_1,\alpha_1-\alpha_2\}} (1+|\braket{v_1|v_2}|^2) + 2(1 + \sqrt{2} + T)^2  }. 
   \end{split}
   \end{equation}
   We choose $T$ such that 
   \begin{equation}\label{eqn:choice_T_eig_diff_inhomo}
       \epsilon e^{2T \min\{\alpha_1,\alpha_1-\alpha_2\} } = (1 + \sqrt{2} + T)^2. 
   \end{equation}
   Notice that such $T$ uniquely exists, because the function $f(T) = \sqrt{\epsilon} e^{T \min\{\alpha_1,\alpha_1-\alpha_2\} } - (1 + \sqrt{2} + T)$ defined on $[0,+\infty)$ first monotonically decreases and then monotonically increases with boundary values $f(0) < 0$ and $f(+\infty) = +\infty$. 
   Therefore, 
   \begin{equation}\label{eqn:final_fed_eig_diff_inhomo}
        |\braket{u(T)|w(T)}| \leq \sqrt{\frac{2|\braket{v_1|v_2}|^2 + 2 }{1 + |\braket{v_1|v_2}|^2 + 2 }} \eqqcolon C.
    \end{equation}
    Here $C < 1$ and only depends on $V$, and thus we regard $C$ as a constant and will absorb it into the notation $\mathcal{O}$ and $\Omega$. 
    
    Given a black box to prepare either $\ket{u(0)}$ or $\ket{w(0)}$, we denote $\ket{\widetilde{u}(T)}$ and $\ket{\widetilde{w}(T)}$ as the corresponding outputs of a quantum differential equation solver with $2$-norm distance at most $(1-C)/4$ of $\ket{u(T)}$ and $\ket{w(T)}$, respectively. 

    Suppose the opposite of our claim that there exists efficient generic quantum algorithm that can solve the ODE with cost $o(e^{T \min\{\alpha_1,\alpha_1-\alpha_2\} } / (1 + \sqrt{2} + T)) = o(1/\sqrt{\epsilon})$. 
    Then $\ket{\widetilde{u}(T)}$ and $\ket{\widetilde{w}(T)}$ can be prepared using $o(1/\sqrt{\epsilon})$ queries to the state preparation oracle. 
    Following the same argument as in~\cref{eqn:final_dis_eig_diff}, we have $\|\ket{\widetilde{u}(T)}\bra{\widetilde{u}(T)} - \ket{\widetilde{w}(T)}\bra{\widetilde{w}(T)}\|_1 = \Omega(1)$, so an amplifier of $\ket{u(0)}$ and $\ket{w(0)}$ can be constructed with cost $o(1/\sqrt{\epsilon})$. 
    This contradicts with~\cref{thm:state_discrimination}, and thus completes the proof. 
\end{proof}

\begin{prop}\label{prop:lb_non_normal_inhomo}
    Consider the inhomogeneous ODE problem in~\cref{eqn:ODE}. 
    Let $\mu(A) = \|A^{\dagger}A-AA^{\dagger}\|^{1/2}$. 
    Then, there is no generic quantum algorithm that can prepare $u(T)/\|u(T)\|$ with bounded error and failure probability, using $o(\mu(A))$ queries to the preparation oracle of $\ket{u(0)}$, its inverse or controlled versions. 
\end{prop}

\begin{proof}
     Consider the example with $N = 3$, $u = (u_1,u_2,u_3)^T$, 
     \begin{equation}
         A = \left(\begin{array}{ccc}
             -1 & -1/\delta & 0 \\
             0 & -2 & 0 \\
             0 & 0 & -1/2 
         \end{array}\right),
     \end{equation}
     and 
     \begin{equation}
         b = (0,0,1)^T.
     \end{equation}
     Here $\delta$ is a real positive parameter in $(0,1)$. 
     Notice that the matrix $A$ can be diagonalized such that $A = VDV^{-1}$ where 
     \begin{equation}
         V = \left(\begin{array}{ccc}
             1 & 1 & 0 \\
             0 & \delta & 0 \\
             0 & 0 & 1 
         \end{array}\right), 
         \quad D = \left(\begin{array}{ccc}
             -1 & 0 & 0 \\
             0 & -2 & 0 \\
             0 & 0 & -1/2 
         \end{array}\right),
     \end{equation}
     and that 
     \begin{equation}
         A^{\dagger}A - AA^{\dagger} = \left(\begin{array}{ccc}
             -1/\delta^2 & -1/\delta & 0 \\
             -1/\delta & 1/\delta^2 & 0 \\
             0 & 0 & 1/4 
         \end{array}\right),
     \end{equation}
     \begin{equation}\label{eqn:mu_non_normal}
        \mu(A) = \frac{\sqrt[4]{1+\delta^2}}{\delta} = \Theta\left(\frac{1}{\delta}\right). 
     \end{equation}
     We choose two initial conditions 
     \begin{equation}
         \begin{split}
             u(0) &= (0,0,1)^{T}, \\
             v(0) &= (0,\delta,\sqrt{1-\delta^2})^{T}. 
         \end{split}
     \end{equation}
     According to~\cref{eqn:ODE_solu} and noting that $e^{A}$ can be computed as $Ve^{D}V^{-1}$, we obtain
     \begin{equation}
         \begin{split}
             u(1) &= (0,0,2-e^{1/2})^{T}, \\
             v(1) &= (-e^{-1}+e^{-2}, e^{-2}\delta, 2-(2-\sqrt{1-\delta^2})e^{-1/2})^{T} \coloneqq (v_1(1),v_2(1),v_3(1))^{T}. 
         \end{split}
     \end{equation}
     
     Now, suppose that we are given a black box that prepares either $\ket{u(0)}$ or $\ket{v(0)}$. 
     Let $\ket{\widetilde{u}(1)}$ and $\ket{\widetilde{v}(1)}$ be the corresponding outputs of a quantum ODE solver with $2$-norm distance at most $1/1000$ of $\ket{u(1)}$ and $\ket{v(1)}$, respectively. 
     Also, suppose the opposite of our claim that there exists an efficient generic quantum algorithm that solves the general ODE with cost $o(\mu(A))$. Then, $\ket{\widetilde{u}(1)}$ and $\ket{\widetilde{v}(1)}$ can be obtained using $o(\mu(A)) = o(1/\delta)$ queries to the state preparation oracle. 
     On the one hand, 
     \begin{equation}\label{eqn:overlap_non_normal}
         \braket{u(0)|v(0)} = \sqrt{1-\delta^2} \geq 1-\delta^2. 
     \end{equation}
     On the other hand, 
     \begin{equation}
         \begin{split}
             |\braket{u(1)|v(1)}| & = \frac{|v_3(1)|}{\sqrt{v_1(1)^2+v_2(1)^2+v_3(1)^2}} \\
             & = \frac{1}{\sqrt{1 + v_1(1)^2/v_3(1)^2 + v_2(1)^2/v_3(1)^2 }} \\
             & \leq \frac{1}{\sqrt{1 + v_1(1)^2/v_3(1)^2}} \\
             & \leq \frac{1}{\sqrt{1+(e-1)^2/(4e^4)}},
         \end{split}
     \end{equation}
     where in the last inequality we use 
     \begin{equation}
         \frac{|v_1(1)|}{|v_3(1)|} = \frac{1/e-1/e^2}{2-(2-\sqrt{1-\delta^2})e^{-1/2}} \geq  \frac{1/e-1/e^2}{2}. 
     \end{equation}
     According to~\cref{lem:trace_distance}, we have 
     \begin{equation}
         \begin{split}
             \|\ket{\widetilde{u}(1)}\bra{\widetilde{u}(1)} - \ket{\widetilde{v}(1)}\bra{\widetilde{v}(1)}\|_1 
             & = 2\sqrt{1-|\braket{\widetilde{u}(1)|\widetilde{v}(1)}|^2} \\
             & \geq 2\sqrt{1-(1/\sqrt{1+(e-1)^2/(4e^4)} + 1/1000 + 1/1000 )^2}\\
             & \geq 0.19 = \Omega(1). 
         \end{split}
     \end{equation}
     This implies that an amplifier of $\ket{u(0)}$ and $\ket{v(0)}$ can be constructed with cost $o(1/\delta)$, which, together with~\cref{eqn:overlap_non_normal}, contradicts with~\cref{thm:state_discrimination} and thus completes the proof. 
\end{proof}

Compared to the homogeneous case, the lower bound due to non-normality (\cref{prop:lb_non_normal_inhomo}) is quite similar, yet the condition for exponential overhead due to the eigenvalues has changed. 
In the homogeneous case (\cref{prop:lb_eig_diff_homo}), exponential computational overhead occurs once there exists any difference in eigenvalues' real parts, but in the inhomogeneous case (\cref{prop:lb_eig_diff_inhomo}), that happens under an additional requirement that at least one eigenvalue has positive real part. 
The intuition for the difference is that the shifting equivalence no longer holds for inhomogeneous ODEs due to the the inhomogeneous term $b$.

Our lower bounds indicate that existing generic quantum ODE solvers cannot be significantly improved in the sense that the assumptions cannot be extensively relaxed and the scalings due to the non-normality cannot be exponentially improved. 
We consider a diagonalizable matrix $A$. 
An important assumption in existing algorithms is that all the eigenvalues of $A$ have a non-positive real part. 
This is assumed explicitly in~\cite{Berry2014,BerryChildsOstranderEtAl2017,ChildsLiu2020} (see, \emph{e.g.},~\cref{lem:DEsolver_TD}) and implicitly in~\cite{Krovi2022} through the dependence on the parameter $\max_t \|e^{At}\|$ (see~\cref{lem:DEsolver_TI}). 
Our~\cref{prop:lb_eig_diff_inhomo} indicates a worst-case exponential computational cost when an eigenvalue has a positive real part, except in the restrictive scenario where all the eigenvalues of $A$ have the same real part. 
In addition, there are some computational overheads due to the non-normality in the existing algorithms, for example, the parameter $\kappa_V$ in~\cref{lem:DEsolver_TD} and $\max_t\|e^{At}\|$ in~\cref{lem:DEsolver_TI}. 
Our~\cref{prop:lb_non_normal_inhomo} suggests that such dependencies are unlikely to be exponentially improved. 
We remark that we constructed all our lower bound witnesses using  diagonalizable matrices $A$, while~\cref{lem:DEsolver_TI} also applies to non-diagonalizable matrices. 

\section{More fast-forwarding results for negative definite and semi-definite matrices}\label{app:FF_ND_NSD}

Here we present more general fast-forwarding results for negative definite and semi-definite matrices with quadratic speedup in various parameters. 

\subsection{Negative-definite coefficient matrix}\label{sec:fast_forwarding_negative_definite}

\subsubsection{Oracles}\label{sec:oracles_ND}

Let $A \in \mathbb{C}^{N\times N}$ be a negative definite Hermitian matrix. 
For technical simplicity, we assume $\|A\| \leq 1$ throughout this scenario. 
Suppose we have a $(1,n_A,0)$-block-encoding of $A$ by a unitary $U_A$. 
Furthermore, suppose $O_u$ and $O_b$ are the oracles such that $O_u \ket{0} = \frac{1}{\|u(0)\|} \sum_{j=0}^{N-1} u_j(0) \ket{j}$ and $O_b \ket{0} = \frac{1}{\|b\|} \sum_{j=0}^{N-1} b_j \ket{j}$, and assume that $\|u(0)\|$, $\|b\|$ are known. 

\subsubsection{Homogeneous case}

The approach largely follows~\cite{GilyenSuLowEtAl2019}, and the key component for quadratic speedup is that the exponential function $e^{-T(1-x)}$ can be approximated by a degree-$\mathcal{O}(\sqrt{T})$ polynomial and can be implemented using QSVT (\cref{lem:poly_block_encoding}). 
Then, starting from a block-encoding of $A$, one can first construct a block-encoding of $(I+A)/2$ using the linear combination of unitaries technique (\cref{lem:LCU}), then construct a block-encoding of $(I+A)$ by uniformly amplifying singular values~\cite[Theorem 30]{GilyenSuLowEtAl2019}. 
Notice that here we require the matrix $A$ to be negative definite to control the approximation error. 
Then, the operator $e^{AT} = e^{-T(I-(I+A))}$ can be block-encoded using the circuit for $e^{-T(1-x)}$ with $\mathcal{O}(\sqrt{T})$ query complexity.

\begin{lem}\label{lem:FF_ND_BE_homo}
    Consider solving~\cref{eqn:ODE_general} with $b = 0$ and $A$ is a negative definite Hermitian matrix such that all the eigenvalues of $A$ are within the interval $[-1,-\delta]$ for a $\delta > 0$. 
    Suppose that we are given a $(1,n_A,0)$-block-encoding of $A$, denoted by $U_A$. 
    Then for any $T>0$ and $\epsilon<1/4$, a $(3,n_A+4,\epsilon)$-block-encoding of $e^{AT}$ can be constructed using 
    $$\mathcal{O}\left(  \frac{\sqrt{T}}{\delta}\log\left(\frac{1}{\epsilon}\right)\log\left(\frac{T\log(1/\epsilon)}{\epsilon}\right) \right)$$
    queries to $U_A$, its inverse and controlled versions, and 
    $$\mathcal{O}\left(\left(n_A+\frac{1}{\delta}\log\left(\frac{T\log(1/\epsilon)}{\epsilon}\right)\right)\sqrt{T}\log\left(\frac{1}{\epsilon}\right) \right)$$
    additional one- or two-qubit gates. 
\end{lem}
\begin{proof}

Starting from the block-encoding of $A$, $V''\coloneqq (H\otimes \one) c\text{-} U (H\otimes \one)$ is a $(2,n_A+1,0)$ block-encoding of $(\one + A)$. Note that $V''$ uses one call to $c\text{-} U$. Equivalently, $V''$ is a $(1, n_A+1, 0)$ block-encoding of $(\one+A)/2$. 

By assumption, $\norm{\one + A} \leq 1 - \delta$ for the parameter $\delta>0$. Therefore, $\norm{(\one+A)/2} \leq (1-\delta)/2$. Then, for $\epsilon_1>0$, by \cite[Theorem 30]{GilyenSuLowEtAl2019}, we can construct a $(1,n_A+2,\epsilon_1)$-block-encoding of $(\one + A)$ using $O(\frac{1}{\delta}\log(1/\epsilon_1))$ calls to $V''$, $V''^\dagger$ and additional one- and two-qubit gates. Call this block-encoding $V'$. 

Let $T, \epsilon_2>0$. 
By \cite[Corollary 64]{GilyenSuLowEtAl2019}, there exists an efficiently constructable real polynomial $P_{T,\epsilon_2}(x)$ of degree $d_{T,\epsilon_2} = O(\sqrt{\max(T,\log(1/\epsilon_2))\log(1/\epsilon_2)})$ such that
\begin{equation}
    |e^{-T(1-x)} - P_{T,\epsilon_2}(x)|_{[-1,1]} \leq \epsilon_2.
\end{equation}
Let $\tilde{P}_{T,\epsilon_2}\coloneqq \frac{1}{3}P_{T,\epsilon_2}(x)$. Then, for any $x\in [-1,1]$ and $\epsilon_2 \leq 1/2$, we have
\begin{equation}
    \Bigl|\tilde{P}_{T,\epsilon_2}(x) \Bigr| = \Bigl| \frac{P_{T,\epsilon_2}(x)}{3} \Bigr| \leq \max_{x\in [-1,1]} \Bigl| \frac{e^{-T(1-x)}}{3} \Bigr| + \frac{\epsilon_2}{3} \leq \frac{1}{2}.
\end{equation}
Then, by~\cref{lem:poly_block_encoding}, we can construct a $(1, n_A+4, 4d_{t,\epsilon_2}\sqrt{\epsilon_1})$-block-encoding of $\tilde{P}_{T,\epsilon_2}(\one+A)$ using $d_{T,\epsilon_2}$ calls to $V'$ and  $V'^{\dagger}$, a single call to $c\text{-} V'$ and $\mathcal{O}(n_A d_{T,\epsilon_2})$ additional gates. 
Call this block-encoding $V$. 

Note that $\norm{\tilde{P}_{T,\epsilon_2}(\one+A) - e^{AT}/3}\leq \epsilon_2/3$, $V$ can also be regarded as a $(3, n_A+4, 12d_{t,\epsilon_2}\sqrt{\epsilon_1}+\epsilon_2)$-block-encoding of $e^{AT}$. 
To bound the overall error by $\epsilon>0$, it suffices to choose $12d_{t,\epsilon_2}\sqrt{\epsilon_1}+\epsilon_2 \leq \epsilon$, and this can be achieved by choosing 
\begin{equation}
    \epsilon_2 = \Theta(\epsilon), \quad \epsilon_1 = \Theta\left(\frac{\epsilon^2}{\max(T,\log(1/\epsilon))\log(1/\epsilon)}\right).
\end{equation}
The overall number of calls to $U_A$ (and its inverse and controlled versions) is of order
\begin{equation}
    \mathcal{O}\left( d_{T,\epsilon_2}\times \frac{1}{\delta}\log(1/\epsilon_1) \right) = \mathcal{O}\left(  \frac{\sqrt{T}}{\delta}\log\left(\frac{1}{\epsilon}\right)\log\left(\frac{T\log(1/\epsilon)}{\epsilon}\right) \right), 
\end{equation}
and the number of additional gates required is 
\begin{equation}
    \mathcal{O}\left( n_A d_{t,\epsilon_2} + d_{T,\epsilon_2}\frac{1}{\delta}\log(1/\epsilon_1) \right) = \mathcal{O}\left(\left(n_A+\frac{1}{\delta}\log\left(\frac{T\log(1/\epsilon)}{\epsilon}\right)\right)\sqrt{T}\log\left(\frac{1}{\epsilon}\right) \right). 
\end{equation}
\end{proof}

To solve the homogeneous ODE, we first apply $O_u$ to prepare $\ket{0}_a\ket{u(0)}_s$ (where ``a'' refers to ``ancilla'' and ``s'' refers to ``state register'') and then apply the block-encoding of $e^{AT}$ to obtain $\ket{0}_a C \ket{\psi}_s + \ket{1}_a C' \ket{\psi'}_s$ for some rescaling factor $C$. 
Here $C \ket{\psi}_s$ is an $\mathcal{O}(\epsilon')$-approximation of $\frac{1}{3\|u(0)\|}e^{AT}u(0) = \frac{1}{3\|u(0)\|}u(T)$. 
According to~\cref{lem:succ_prob_error}, measuring the ancilla qubits and getting $0$ yield an $\mathcal{O}(\epsilon'\|u(0)\|/\|u(T)\|)$ approximation of $\ket{u(T)}$. 
We choose $\epsilon' = \epsilon \|u(T)\|/\|u(0)\|$ to bound the error by $\epsilon$. 
The success probability can be boosted to $\Omega(1)$ using amplitude amplification and repeating the procedure for $\mathcal{O}(\|u(0)\|/\|u(T)\|)$ times. 
Therefore the overall query complexity is 
\begin{equation}
    \widetilde{\mathcal{O}}\left(\frac{\|u(0)\|}{\|u(T)\|} \frac{\sqrt{T}}{\delta}\log^2\left(\frac{1}{\epsilon}\right) \right). 
\end{equation}
We remark that, although the explicit scaling in $T$ seems quadratic, the actual overall scaling is always exponential in $T$ since $\|u(T)\|$ is always exponentially small for negative definite matrix $A$. 
However, if we consider the inhomogeneous ODE instead, a truly quadratic speedup in $T$ may be obtained because the solution of inhomogeneous ODE does not exhibit exponential decay, and the block-encoding of the homogeneous evolution operator $e^{AT}$ is the first step to solve inhomogeneous ODE. We will discuss this in detail next.

\subsubsection{Inhomogeneous case}\label{sec:inhomo_ND}

Now we consider solving~\cref{eqn:ODE_general} with a negative-definite Hermitian matrix $A$ and a time-independent inhomogeneous term $b$. 
The idea is directly based on the solution form given in~\cref{eqn:ODE_solu_intro}: we first construct block-encodings of the operators $e^{AT}$ and $\int_0^T e^{A(T-s)}ds$, then apply them to the corresponding vectors to obtain $e^{AT}u_0$ and $\int_0^T e^{A(T-s)} b ds $, and finally use a technique similar to that of LCU to construct the linear combination of these two vectors. 

We start with the block-encodings of the operators. 
The block-encoding for $e^{AT}$ has been constructed in~\cref{lem:FF_ND_BE_homo}. 
We construct the block-encoding of $\int_0^T e^{A(T-s)} ds$ as in~\cref{fig:circuit_FF_ND_inhomo_term}, which is based on the equation $\int_0^T e^{A(T-s)} ds = (e^{AT}-I)A^{-1}$. 

\begin{figure}
    \centerline{
    \Qcircuit @R=1em @C=1em {
    \text{Control}\quad\quad\quad\quad\quad\quad\quad\quad\quad &  \gate{\mathrm{R}_y(-\pi/3)} & \ctrl{1} & \gate{\mathrm{R}_y(-\pi/3)} & \qw & \qw \\
    \text{Ancilla for } e^{AT} \quad\quad\quad\quad\quad\quad\quad\quad\quad  & \qw & \multigate{1}{U_{e^{AT}}} & \qw & \qw & \qw \\ 
    \text{Vector} \quad\quad\quad\quad\quad\quad\quad\quad\quad  & \qw & \ghost{U_{e^{AT}}} & \qw & \multigate{1}{U_{A^{-1}}} & \qw \\
    \text{Ancilla for } A^{-1} \quad\quad\quad\quad\quad\quad\quad\quad\quad  & \qw & \qw & \qw &  \ghost{U_{A^{-1}}} & \qw \\
    }
    }
    \caption{ Quantum circuit for constructing a block-encoding of $\int_0^T e^{A(T-s)}ds$ for a negative definite Hermitian matrix $A$. Here for a matrix $M$, $U_{M}$ represents its block-encoding. }
    \label{fig:circuit_FF_ND_inhomo_term}
\end{figure}

\begin{lem}\label{lem:FF_ND_inhomo_BE}
    Suppose that $A$ is a negative-definite Hermitian matrix such that all the eigenvalues of $A$ are within $[-1,-\delta]$, and $U_A$ is a $(1,n_A,0)$-block-encoding of $A$. 
    Then for any $T>0$ and $\epsilon<1/2$, an $(16/(3\delta), 2n_A+6,\epsilon)$-block-encoding of $\int_0^T e^{A(T-s)}ds$ can be constructed, using 
    \begin{equation}
        \mathcal{O}\left(\frac{\sqrt{T}}{\delta} \log\left(\frac{1}{\delta\epsilon}\right) \log\left(\frac{T\log(1/(\delta\epsilon))}{\delta\epsilon}\right) \right)
    \end{equation}
    queries to $U_A$, its inverse and controlled version, and 
    \begin{equation}
        \mathcal{O}\left(\left(n_A+\frac{1}{\delta}\log\left(\frac{T\log(1/(\delta\epsilon))}{\delta\epsilon}\right)\right)\sqrt{T}\log\left(\frac{1}{\delta\epsilon}\right) \right)
    \end{equation}
    extra one or two-qubit gates. 
\end{lem}
\begin{proof}
    Let $n_1 = n_A+4$. 
    We start with the $(3,n_1,\epsilon')$-block-encoding of $e^{AT}$ (or equivalently a $(1,n_1,\epsilon'/3)$-block-encoding of $e^{AT}/3$), denoted by $U_A'$, constructed in~\cref{lem:FF_ND_BE_homo} with $\epsilon'$ to be determined later. 
    Let $\mathrm{R}_y(\theta)$ denote the single-qubit rotation gate.
    Notice that 
    \begin{equation}
        \begin{split}
            \mathrm{R}_y(\pi/3) \ket{0} &= \frac{1}{2}(\sqrt{3}\ket{0}+\ket{1}), \\
            \mathrm{R}_y(-\pi/3) \ket{0} &= \frac{1}{2}(\sqrt{3}\ket{0}-\ket{1}), 
        \end{split}
    \end{equation}
    then $(\mathrm{R}_y(\pi/3), \mathrm{R}_y(-\pi/3))$ is a $(4,1,0)$-state-preparation-pair of the vector $(3,-1)$. 
    According to~\cref{lem:LCU}, a $(4,n_1+1,4\epsilon'/3)$-block-encoding of $(e^{AT}-I)$ can be constructed with a single use of controlled $U_A'$ and $2$ extra one-qubit gates. 
    This cost is equivalent to $\mathcal{O}((\sqrt{T}/\delta)\log(1/\epsilon') \log((T/\epsilon')\log(1/\epsilon')))$ of controlled $U_A$ and $\mathcal{O}((n_A+(1/\delta)\log((T/\epsilon')\log(1/\epsilon')))\sqrt{T}\log(1/\epsilon'))$ extra one or two-qubit gates according to~\cref{lem:FF_ND_BE_homo}. 
    Meanwhile,~\cref{lem:inverse_block_encoding} tells that a $(4/(3\delta),n_A+1,\epsilon'')$-block-encoding of $A^{-1}$ can be constructed using $\mathcal{O}((1/\delta)\log(1/(\delta\epsilon'')))$ queries to $U_A$ and its inverse. 
    
    Multiplying this two block-encodings together and using~\cref{lem:multiplication_block_encoding}, we can construct a $(16/(3\delta), n_A+n_1+2, 4\epsilon''+16\epsilon'/(9\delta))$-block-encoding of $(e^{AT}-I)A^{-1}$, using 
    \begin{equation}
        \mathcal{O}\left(\frac{1}{\delta}\left( \sqrt{T}\log\left(\frac{1}{\epsilon'}\right) \log\left(\frac{T\log(1/\epsilon')}{\epsilon'}\right) + \log\left(\frac{1}{\delta \epsilon''}\right) \right)\right)
    \end{equation}
    queries to $U_A$, its inverse and controlled version, and 
    \begin{equation}
       \mathcal{O}\left(\left(n_A+\frac{1}{\delta}\log\left(\frac{T\log(1/\epsilon')}{\epsilon'}\right)\right)\sqrt{T}\log\left(\frac{1}{\epsilon'}\right) \right)
    \end{equation}
    extra one or two-qubit gates. 
    The proof is completed by noticing that $(e^{AT}-I)A^{-1} = \int_0^T e^{A(T-s)}ds$ and choosing $8\epsilon'/(9\delta) = 2\epsilon'' = \epsilon $. 
\end{proof}

We are now ready to state our main result, which is a direct consequence of~\cref{lem:FF_ND_BE_homo},~\cref{lem:FF_ND_inhomo_BE} and~\cref{lem:FF_LCS}. 

\begin{thm}\label{thm:FF_ND}
    Consider solving the ODE system~\cref{eqn:ODE_general} with time-independent $b$ up to time $T$, where $A$ is a bounded negative definite Hermitian matrix such that all the eigenvalues of $A$ are within $[-1,-\delta]$. 
    Suppose the oracles described in~\cref{sec:oracles_ND}. 
    Then for any $T,\epsilon$ such that $0 < \epsilon < 1/2$, there exists a quantum algorithm that outputs an $\epsilon$-approximation of $\ket{u(T)}$ with $\Omega(1)$ success probability, using 
    \begin{enumerate}
        \item \begin{equation}
            \mathcal{O}\left(\frac{\|u(0)\|+\|b\|/\delta}{\|u(T)\|} \frac{\sqrt{T}}{\delta} \log\left(\frac{1}{\delta\epsilon}\right) \log\left(\frac{T\log(1/(\delta\epsilon))}{\delta\epsilon}\right)\right)
        \end{equation}
        queries to $U_A$, its inverse and controlled version, 
        \item \begin{equation}
            \mathcal{O}\left(\frac{\|u(0)\|+\|b\|/\delta}{\|u(T)\|}\right)
        \end{equation}
        queries to the controlled versions of $O_u$ and $O_b$,
        \item $(2n_A+7)$ ancilla qubits, 
        \item \begin{equation}
            \mathcal{O}\left(\frac{\|u(0)\|+\|b\|/\delta}{\|u(T)\|} \left(n_A+\frac{1}{\delta}\log\left(\frac{T\log(1/(\delta\epsilon))}{\delta\epsilon}\right)\right)\sqrt{T}\log\left(\frac{1}{\delta\epsilon}\right)\right)
        \end{equation}
        extra one or two-qubit gates. 
    \end{enumerate}
\end{thm}

Although there is still some $T$-dependence through the parameter $\|u(T)\|$ in the query complexity,~\cref{thm:FF_ND} indeed shows that quadratic fast-forwarding in $T$ can be achieved when solving an inhomogeneous ODE system with a negative definite coefficient matrix. 
This is because the norm of $\|u(T)\|$ can be asymptotically bounded independent of $T$ for a negative definite matrix $A$ (but might depend on $\delta$). 
To see this more clearly, let us consider a simplified case where we assume $\|u(0)\|$, $\|b\|$ and $\delta$ are $\Theta(1)$.
Then, according to~\cref{eqn:ODE_solu_intro}, 
\begin{equation}
\begin{split}
    \|u(T)\| &\geq \|(e^{AT}-I)A^{-1}b\| - \|e^{AT}u(0)\| \\
    & \geq \|A^{-1}b\| - \|e^{AT}A^{-1}b\| - \|e^{AT}u(0)\| \\
    & \geq \|b\| - e^{-\delta T}/\delta - e^{-\delta T} = \Omega(1). 
\end{split}
\end{equation}
Therefore, when the matrix $A$ is well conditioned (i.e., $\delta$ is on a constant level), the overall query complexity can be bounded by $\widetilde{\mathcal{O}}\left(\sqrt{T}\log^2(1/\epsilon)\right)$, yielding a quadratic speedup in $T$ up to a logarithmic factor. 

We also remark that our algorithm involves a polynomial dependence on the condition number $1/\delta$, while the best existing generic algorithm does not have this $1/\delta$ dependence. 
This is mainly because we would like to block-encode $e^{AT}$ directly from the block-encoding of $A$ with scaling $\mathcal{O}(\sqrt{T})$, so we perform a linear transformation of $A$ and then amplify its amplitude, which introduces this $\delta$ scaling. 
Under a stronger assumption on the input model, we can further remove this $\delta$ dependence as discussed in the next subsection.

\subsection{Negative semi-definite coefficient matrix with square-root access}

We now consider the case where $A$ is a negative semi-definite matrix but with square-root access. 
Specifically, we consider $A = -H^2$ for a Hermitian matrix $H$ and assume block-encoding access to $H$ (not to $A$). 
Compared to~\cref{sec:fast_forwarding_negative_definite}, although here $A$ can be semi-definite, we are actually making a stronger assumption on the input model.  
This is because once we have a block-encoding $U_H$ of $H$, then it is straightforward to construct a block-encoding of $A$ by two applications of $U_H$~\cite{GilyenSuLowEtAl2019}. 
On the other hand, constructing $U_H$ from the block-encoding of $A$ is much more difficult, as we need to implement the square-root function $\sqrt{x}$, which is singular at $0$. 
With this stronger assumption, we can design another fast-forwarded algorithm with better scaling and without dependence on the condition number compared to~\cref{sec:fast_forwarding_negative_definite}. 

\subsubsection{Oracles}\label{sec:oracles_NSD_square_root}

Let $A \in \mathbb{C}^{N\times N}$ be a negative semi-definite Hermitian matrix such that $A = -H^2$ for a Hermitian matrix $H$. 
Suppose we have an $(\alpha_H,n_H,0)$-block-encoding of $H$ by a unitary $U_H$. 
Furthermore, suppose $O_u$ and $O_b$ are the oracles such that $O_u \ket{0} = \frac{1}{\|u(0)\|} \sum_{j=0}^{N-1} u_j(0) \ket{j}$ and $O_b \ket{0} = \frac{1}{\|b\|} \sum_{j=0}^{N-1} b_j \ket{j}$, and assume that $\|u(0)\|$, $\|b\|$ are known. 

\subsubsection{Homogeneous case}

Fast-forwarding $e^{-TH^2}$ has been studied in several works including~\cite{ChowdhurySomma2016,GilyenSuLowEtAl2019,ApersChakrabortyNovoEtAl2022}. 
Here we follow the idea suggested in~\cite{GilyenSuLowEtAl2019} and work out technical details. 
The idea is straightforward: the function $e^{-\beta x^2}$ can be approximated with a degree-$\mathcal{O}(\sqrt{\beta})$ polynomial and then $e^{AT} = e^{-T H^2}$ can be implemented via~\cref{lem:poly_block_encoding}. 

\begin{lem}\label{lem:approx_NSD_square_root_homo}
    For any $\beta > 0$ and $0 < \epsilon < 1$, there exists an efficiently computable even polynomial $P(x)$ such that 
    \begin{equation}
        \sup_{x\in[-1,1]} |e^{-\beta x^2} - P(x)| \leq \epsilon
    \end{equation}
    and has degree $\mathcal{O}(\sqrt{\max(\beta,\log(1/\epsilon)) \log(1/\epsilon)})$. 
\end{lem}
\begin{proof}
     This is a direct consequence following~\cite[Theorem 4.1]{SachdevaVishnoi2014} that $e^{-y}$ on the interval $[0,\beta]$ can be approximated using a polynomial of degree $\mathcal{O}(\sqrt{\max(\beta,\log(1/\epsilon))\log(1/\epsilon)})$ and let $y = \beta x^2$. 
\end{proof}

\begin{lem}\label{lem:FF_square_root_access_homo}
    Consider solving~\cref{eqn:ODE_general} with $b = 0$ and $A$ is a negative semi-definite Hermitian matrix such that $A = -H^2$ for a Hermitian $H$. 
    Suppose that we are given a $(\alpha_H,n_H,0)$-block-encoding of $H$, denoted by $U_H$. 
    Then for any $T>0$ and $\epsilon<1/4$, a $(3,n_H+2,\epsilon)$-block-encoding of $e^{AT}$ can be constructed using 
    $\mathcal{O}(\sqrt{\max(T\alpha_H^2,\log(1/\epsilon)) \log(1/\epsilon)})$
    queries to $U_H$, its inverse and controlled versions, and 
    $\mathcal{O}(n_H \sqrt{\max(T\alpha_H^2,\log(1/\epsilon)) \log(1/\epsilon)})$
    additional one- or two-qubit gates. 
\end{lem}
\begin{proof}
     Let $\beta = T\alpha_H^2$, $P(x)$ be the polynomial in~\cref{lem:approx_NSD_square_root_homo} with error $\epsilon$, and 
     $d = \mathcal{O}(\sqrt{\max(\beta,\log(1/\epsilon))\log(1/\epsilon)})$
     Then, according to~\cref{lem:poly_block_encoding}, a $(1,n_H+2,0)$-block-encoding of $P(H/\alpha_H)/3$ can be constructed using $\mathcal{O}(d)$ queries and $\mathcal{O}(n_H d)$ additional gates. 
     Equivalently, this is a $(3,n_H+2,\epsilon)$-block-encoding of $e^{-TH^2} = e^{AT}$. 
\end{proof}

The parameter $\alpha_H$ depends on specific circuits that realize the oracle $U_H$, and it can be as small as $\sqrt{\|A\|}$. 
In this optimal case, the overall complexity becomes $\widetilde{\mathcal{O}}(\sqrt{T\|A\|} \log(1/\epsilon))$ and achieves quadratic speedup in both $T$ and $\|A\|$ for constructing block-encoding. 
When it comes to solving homogeneous ODE, similar to the discussion after~\cref{lem:FF_ND_BE_homo}, an extra post-selection step is required and the overall query complexity becomes 
\begin{equation}
    \widetilde{\mathcal{O}}\left(\frac{\|u(0)\|}{\|u(T)\|}\sqrt{T\|A\|} \log\left(\frac{1}{\epsilon}\right)\right). 
\end{equation}
However, unlike the negative definite case, a real quadratic improvement can hold once the matrix $A$ has an $0$ eigenvalue and $\ket{u(0)}$ has non-trivial overlap with the corresponding eigenstate. 
In this case, $\|u(T)\| = \|e^{AT}u(0)\| \geq \|u(0)\||\braket{u(0)|\psi_0}| = \Omega(\|u(0)\|)$ where $\ket{\psi_0}$ is the eigenstate corresponding to $0$ eigenvalue, and the overall complexity becomes $\widetilde{\mathcal{O}}(\sqrt{T\|A\|}\log(1/\epsilon))$. 

\subsubsection{Inhomogeneous case}\label{sec:inhomogenous_sqrt}

Now we consider~\cref{eqn:ODE_general} with a time-independent inhomogeneous term $b$. 
The idea is the same as~\cref{lem:FF_square_root_access_homo}, and to block encode $\int_0^T e^{A(T-s)}ds$ using QSVT with lower-degree polynomial. 
\begin{lem}\label{lem:FF_square_root_poly_2}
    For any $\beta > 0$ and $0 < \epsilon < 1$, there exists and efficiently computable even polynomial $Q(x)$ such that 
    \begin{equation}
        \sup_{x\in[-1,1]} \left| \int_0^1 e^{-\beta \tau x^2 } d\tau - Q(x)\right| \leq \epsilon
    \end{equation}
    and has degree $\mathcal{O}(\sqrt{\max(\beta,\log(1/\epsilon))\log(1/\epsilon)})$. 
\end{lem}
\begin{proof}
    According to~\cite[Theorem 4.1]{SachdevaVishnoi2014}, $e^{-y}$ on the interval $[0,\beta]$ can be approximated using a polynomial of degree $\mathcal{O}(\sqrt{\max(\beta,\log(1/\epsilon))\log(1/\epsilon)})$. 
    Formally, there exists a polynomial 
    \begin{equation}
        \sum_{j=0}^{J} q_j(\beta,\epsilon) y^j 
    \end{equation}
    such that 
    \begin{equation}
        \sup_{y\in[0,\beta]} \left| e^{-y} - \sum_{j=0}^{J} q_j(\beta,\epsilon) y^j \right| \leq \epsilon, 
    \end{equation}
    and $J = \mathcal{O}(\sqrt{\max(\beta,\log(1/\epsilon))\log(1/\epsilon)})$. 
    Let $y = \beta \tau x^2$, then we have 
    \begin{equation}
        \sup_{x\in[-1,1],\tau\in[0,1]} \left| e^{-\beta \tau x^2} - \sum_{j=0}^{J} q_j(\beta,\epsilon) \beta^j \tau^j x^{2j} \right| \leq \epsilon,  
    \end{equation}
    which implies (using $\int_0^1 \tau^j d\tau = 1/(j+1)$)
    \begin{equation}
        \sup_{x\in[-1,1]} \left| \int_0^1 e^{-\beta \tau x^2} d\tau  - \sum_{j=0}^{J} \frac{q_j(\beta,\epsilon) \beta^j }{j+1} x^{2j}  \right| \leq \epsilon. 
    \end{equation}
    So we may define 
    \begin{equation}
        Q(x) = \sum_{j=0}^{J} \frac{q_j(\beta,\epsilon) \beta^j }{j+1} x^{2j} 
    \end{equation}
    which is the desired estimate and has degree $2J = \mathcal{O}(\sqrt{\max(\beta,\log(1/\epsilon))\log(1/\epsilon)})$. 
\end{proof}

\begin{lem}\label{lem:FF_square_root_access_inhomo_BE}
    Let $A = -H^2$ for a Hermitian $H$. 
    Suppose that we are given a $(\alpha_H,n_H,0)$-block-encoding of $H$, denoted by $U_H$. 
    Then for any $T>0$ and $\epsilon<1/4$, a $(3T,n_H+2,\epsilon)$-block-encoding of $\int_0^T e^{A(T-s)} ds$ can be constructed using 
    $\mathcal{O}(\sqrt{\max(T \alpha_H^2,\log(T/\epsilon))\log(T/\epsilon)})$
    queries to $U_H$, its inverse and controlled versions, and 
    $\mathcal{O}(n_H \sqrt{\max(T \alpha_H^2,\log(T/\epsilon))\log(T/\epsilon)} )$
    additional one- or two-qubit gates. 
\end{lem}
\begin{proof}
    Notice that 
    \begin{equation}
        \frac{1}{T}\int_0^T e^{A(T-s)} ds = \frac{1}{T}\int_0^T e^{-H^2(T-s)} ds = \int_0^1 e^{-T\alpha_H^2 s (H/\alpha_H)^2 } ds. 
    \end{equation}
    Let $\beta = T \alpha_H^2$, $Q(x)$ be the polynomial in~\cref{lem:FF_square_root_poly_2} with error $\epsilon'$, and $d = \mathcal{O}(\sqrt{\max(\beta,\log(1/\epsilon'))\log(1/\epsilon')})$
    Then, according to~\cref{lem:poly_block_encoding}, a $(1,n_H+2,0)$-block-encoding of $Q(H/\alpha_H)/3$ can be constructed using $\mathcal{O}(d)$ applications of $U_H$, its inverse and controlled versions, and $\mathcal{O}(n_H d)$ additional one- and two-qubit gates. 
    Notice that this is also a $(1,n_H+2,\epsilon'/3)$-block-encoding of $\frac{1}{3T} \int_0^T e^{A(T-s)} ds $, and equivalently a $(3T,n_H+2,T\epsilon')$-block-encoding of $\int_0^T e^{A(T-s)} ds$. 
    We can choose $\epsilon' = \epsilon/T$. 
\end{proof}

The result for the inhomogeneous case is a direct consequence of~\cref{lem:FF_square_root_access_homo},~\cref{lem:FF_square_root_access_inhomo_BE} and~\cref{lem:FF_LCS}. 

\begin{thm}\label{thm:FF_square_root_inhomo}
    Consider solving the ODE system~\cref{eqn:ODE_general} with time-independent $b$ up to time $T$, where $A$ is a negative semi-definite Hermitian matrix such that $A=-H^2$ for a Hermitian $H$. 
    Suppose the oracles described in~\cref{sec:oracles_NSD_square_root}. 
    Then for any $T,\epsilon$ such that $0 < \epsilon < 1/4$, there exists a quantum algorithm that outputs an $\epsilon$-approximation of $\ket{u(T)}$ with $\Omega(1)$ success probability, using 
    \begin{enumerate}
        \item \begin{equation}
            \mathcal{O}\left( \frac{\|u(0)\| + T \|b\|}{ \|u(T)\| }  \sqrt{\max\left(T \alpha_H^2,\log\left(\frac{\|u(0)\| + T \|b\|}{\|u(T)\| \epsilon }\right)\right)\log\left(\frac{\|u(0)\| + T \|b\|}{\|u(T)\| \epsilon }\right)}\right) 
        \end{equation}
        queries to $U_H$, its inverse and controlled version, 
        \item \begin{equation}
            \mathcal{O}\left( \frac{\|u(0)\| + T \|b\|}{ \|u(T)\| } \right)
        \end{equation}
        queries to the controlled versions of $O_u$ and $O_b$,
        \item $( n_H+2 )$ ancilla qubits, 
        \item \begin{equation}
            \mathcal{O}\left( n_H \frac{\|u(0)\| + T \|b\|}{ \|u(T)\| }  \sqrt{\max\left(T \alpha_H^2,\log\left(\frac{\|u(0)\| + T \|b\|}{\|u(T)\| \epsilon }\right)\right)\log\left(\frac{\|u(0)\| + T \|b\|}{\|u(T)\| \epsilon }\right)}\right) 
        \end{equation}
        extra one or two-qubit gates. 
    \end{enumerate}
\end{thm}

If we assume the asymptotically optimal block-encoding of $H$ in the sense that $\alpha_H$ is on the same scaling of $\|H\|$, then, from $A = -H^2$, we have $\alpha_H = \mathcal{O}(\|H\|) = \mathcal{O}(\sqrt{\|A\|})$. 
Compared to~\cref{lem:DEsolver_TI}, we quadratically improve the dependence on $\|A\|$ in query complexity. 
Notice that if the block-encoding of $H$ has complicated construction, resulting in a much larger $\alpha_H$ compared to $\|H\|$, then such a quadratic speedup may become invalid. For time dependence, 
in the worst case, the query complexity of our algorithm is $\mathcal{O}(T^{3/2})$, which is even worse than the generic algorithm. 
However, it is possible and typical that $\|u(T)\|$ grows linearly in $T$ and cancels a linear-in-$T$ term in the numerator, yielding a quadratic improvement in both $T$ and $\|A\|$. 
To see this, let us consider the case where $\lambda_0 = 0$ and, for technical simplicity, assume $\|u(0)\| = \|b\| = 1$. 
Recall that this represents the inhomogeneous heat equation and the advection-diffusion equation with periodic boundary condition and a time-independent source term. 
Let $\ket{\psi_j}$ be the eigenstates of $A$. 
Then 
\begin{equation}
    \begin{split}
        u(T) &= e^{At} u(0) + \int_0^T e^{A(T-s)} b ds \\
        & = \sum_{j=0}^{N-1} e^{\lambda_j T} \braket{\psi_j|u(0)} \ket{\psi_j} + \sum_{j=0}^{N-1} T f(\lambda_j,T) \braket{\psi_j|b} \ket{\psi_j} \\
        & = (\braket{\psi_0|u(0)} + T \braket{\psi_0|b} ) \ket{\psi_0} + \sum_{j=1}^{N-1} c_j \ket{\psi_j}. 
    \end{split}
\end{equation}
Therefore, provided $|\braket{\psi_0|b}| = \Omega(1)$, we have $\|u(T)\| \geq |\braket{\psi_0|u(0)} + T \braket{\psi_0|b}| = \Omega(T)$ and thus the query complexity becomes $\mathcal{O}(\sqrt{T})$. 
Here the assumption $|\braket{b|\psi_0}| = \Omega(1)$ is natural in the inhomogeneous heat equation, which corresponds to the trivial assumption that the heat source has a non-zero total heat.

\section{Proof of~\texorpdfstring{\cref{lem:FF_NSD_gen_homo}}{lemma}}\label{app:proof_FF_gen_complex}

\begin{proof}
    Let $\ket{v}_v = \sum_{j=0}^{N-1} v_j \ket{j}_v$ be an arbitrary state, and we start with $\sum_{j=0}^{N-1} v_j \ket{j}_v \ket{0}_r \ket{0}_f \ket{0}_e \ket{0}_t$. 
    We first apply $U^{\dagger}$ on the vector register to get
    \begin{equation}
        \sum_{j=0}^{N-1} (U^{\dagger}v)_j \ket{j}_v \ket{0}_r \ket{0}_f \ket{0}_e \ket{0}_t. 
    \end{equation}
    Here for a matrix $M$ and a vector $x$, $(Mx)_j$ represents its $j$-th component. 
    
    Now we implement the diagonal transform $e^{\Lambda T}$. 
    Let $\Lambda = \mathcal{A} + i \mathcal{B} $ where $\mathcal{A} = \text{diag}(\alpha_0,\cdots,\alpha_{N-1})$ and $\mathcal{B} = \text{diag}(\beta_0,\cdots,\beta_{N-1})$. 
    Since $\mathcal{A}$ and $\mathcal{B}$ commute, we have $e^{\Lambda T} = e^{i\mathcal{B}T}e^{\mathcal{A}T}$. 
    We can then implement $e^{\mathcal{A}T}$ and $e^{i\mathcal{B}T}$ sequentially. 
    The validation of the parts regarding $e^{\mathcal{A}T}$ is the same as the proof of~\cref{lem:FF_NSD_homo}. 
    As a consequence, after applying $O_T$ and $O_{\Lambda,r}$, applying $O_{\exp,\alpha}$, performing the controlled rotation and uncomputing (except the time register), we arrive at 
    \begin{equation}
    \begin{split}
        e^{-\alpha T}\sum_{j=0}^{N-1} (e^{\mathcal{A} T} U^{\dagger}v)_j \ket{j}_v \ket{0}_r \ket{0}_f \ket{0}_e \ket{T}_t +  \ket{\perp}. 
    \end{split}
    \end{equation}
    For $e^{i\mathcal{B}T}$, we first apply $O_{\Lambda,i}$ to encode the imaginary parts and get 
    \begin{equation}
        e^{-\alpha T}\sum_{j=0}^{N-1} (e^{\mathcal{A} T} U^{\dagger}v)_j \ket{j}_v \ket{0}_r \ket{0}_f \ket{\beta_j}_e \ket{T}_t +   \ket{\perp}. 
    \end{equation}
    Applying $O_{\text{prod}}$ to get 
    \begin{equation}
        e^{-\alpha T}\sum_{j=0}^{N-1} (e^{\mathcal{A} T} U^{\dagger}v)_j \ket{j}_v \ket{0}_r \ket{\beta_jT}_f \ket{\beta_j}_e \ket{T}_t +   \ket{\perp}. 
    \end{equation}
    Now we append the phase factor $e^{i\beta_j T}$. 
    This can be done by first flipping the rotation register, applying the phase shift gate conditioned by the function register and flipping the rotation register back. 
    We obtain 
    \begin{equation}
    \begin{split}
        & \quad e^{-\alpha T}\sum_{j=0}^{N-1} e^{i\beta_jT}(e^{\mathcal{A} T} U^{\dagger}v)_j \ket{j}_v \ket{0}_r \ket{\beta_jT}_f \ket{\beta_j}_e \ket{T}_t +  \ket{\perp} \\
        & = e^{-\alpha T}\sum_{j=0}^{N-1} (e^{i\mathcal{B}T}e^{\mathcal{A} T} U^{\dagger}v)_j \ket{j}_v \ket{0}_r \ket{\beta_jT}_f \ket{\beta_j}_e \ket{T}_t +  \ket{\perp}. 
    \end{split}
    \end{equation}
    Uncomputing the last three registers gives 
    \begin{equation}
    \begin{split}
        e^{-\alpha T}\sum_{j=0}^{N-1} (e^{\Lambda T} U^{\dagger}v)_j \ket{j}_v \ket{0}_r \ket{0}_f \ket{0}_e \ket{0}_t +  \ket{\perp}. 
    \end{split}
    \end{equation}
    
    Finally, applying $U$ gives 
    \begin{equation}
        e^{-\alpha T}\sum_{j=0}^{N-1} (Ue^{\Lambda T} U^{\dagger}v)_j \ket{j}_v \ket{0}_r \ket{0}_f \ket{0}_e \ket{0}_t +   \ket{\perp}. 
    \end{equation}
    According to the definition of the block-encoding and the equation $e^{AT} = Ue^{\Lambda T} U^{\dagger}$, such a circuit is exactly a $(e^{\alpha T},n_t+n_e+n_f+1,0)$-block-encoding of $e^{AT}$. 
\end{proof}

\section{Proof of~\cref{thm:FF_bt}}\label{app:proof_FF_bt}
\begin{proof}[Proof of~\cref{thm:FF_bt}]
    We start with $\ket{0}_c\ket{0}_r\ket{0}_f\ket{0}_t\ket{0}_e\ket{0}_v$, where the subscripts represent control, rotation, function, time, eigenvalue and vector registers, respectively. 
    For notation simplicity, we let 
    \begin{equation}
        \|b\|_{\text{avg}} \coloneqq \sqrt{\frac{1}{M}\sum_{k=0}^{M-1}\|b(kT/M)\|^2}.
    \end{equation}
    We first perform a single qubit rotation on the control register to get 
    \begin{equation}\label{eqn:FF_bt_ini}
        \frac{1}{e^{\widetilde{\alpha}T} \sqrt{\|u(0)\|^2 + T^2 \|b\|_{\text{avg}}^2}}\left(e^{\widetilde{\alpha}T}\|u(0)\|\ket{0}_c + e^{\widetilde{\alpha}T} T \|b\|_{\text{avg}} \ket{1}_c\right)\ket{0}_r\ket{0}_f\ket{0}_t\ket{0}_e\ket{0}_v. 
    \end{equation}
    
    For the first part $\ket{0}_c\ket{0}_r\ket{0}_f\ket{0}_t\ket{0}_e\ket{0}_v$, we can directly apply~\cref{thm:FF_NSD_gen_homo} using the circuit in~\cref{fig:circuit_FF_NSD_gen_homo} with the only difference being that all the operations are controlled by the control register and only apply when it is $\ket{0}_c$. 
    Then we map the first part to 
    \begin{equation}\label{eqn:FF_bt_1st}
        e^{-\widetilde{\alpha}T} \ket{0}_c\ket{0}_r\ket{0}_f\ket{0}_t\ket{0}_e e^{AT}\ket{u(0)}_v + \ket{0}_c\ket{\perp}. 
    \end{equation}
    
    The second part $\ket{1}_c\ket{0}_r\ket{0}_f\ket{0}_t\ket{0}_e\ket{0}_v$ will encode the approximation of the integral, and all the related operations are controlled by the control register and only apply when it is $\ket{1}_c$. 
    We first apply $O_{\|b\|}$ to get 
    \begin{equation}
        \frac{1}{\sqrt{\sum_{k=0}^{M-1} \|b(kT/M)\|^2}} \sum_{k=0}^{M-1} \|b(kT/M)\|\ket{1}_c\ket{0}_r\ket{0}_f\ket{k}_t\ket{0}_e\ket{0}_v. 
    \end{equation}
    Notice that, with an abuse of notation, here the time register encodes the index of the time step rather than the binary encoding of the exact time. 
    Applying $O_{b,t}$ gives 
    \begin{equation}
        \frac{1}{\sqrt{\sum_{k=0}^{M-1} \|b(kT/M)\|^2}} \sum_{k=0}^{M-1} \|b(kT/M)\|\ket{1}_c\ket{0}_r\ket{0}_f\ket{k}_t\ket{0}_e\ket{b(kT/M)}_v. 
    \end{equation}
    This state gives the weighted superposition of the input states with different information on the time, and thus we can again apply~\cref{thm:FF_NSD_gen_homo} using a circuit similar to that in~\cref{fig:circuit_FF_NSD_gen_homo} with the following differences:
    \begin{enumerate}
        \item All the operators are controlled by the control register and only apply when it is $\ket{1}_c$. 
        \item We do not perform $O_T$ and $O_T^{\dagger}$ since here the time information has already been encoded. 
        \item $O_{\exp,\alpha}$ and $O_{\text{prod}}$ are replaced by $O_{\exp,\widetilde{\alpha},t}$ and $O_{\text{prod},t}$. 
    \end{enumerate}
    The corresponding circuit computes $e^{A(T-kT/M)} b(kT/M)$ in parallel in time and gives the state 
    \begin{equation}
        \frac{e^{-\widetilde{\alpha}T}}{\sqrt{\sum_{k=0}^{M-1} \|b(kT/M)\|^2}} \sum_{k=0}^{M-1} \|b(kT/M)\| \ket{1}_c\ket{0}_r\ket{0}_f\ket{k}_t\ket{0}_e e^{A(T-kT/M)}\ket{b(kT/M)}_v 
        + \ket{1}_c\ket{\perp}. 
    \end{equation}
    Applying $\otimes^{n_t} \mathrm{H}$ on the time register gives 
    \begin{equation}\label{eqn:FF_bt_2nd}
        \frac{e^{-\widetilde{\alpha}T}}{\sqrt{M}\sqrt{\sum_{k=0}^{M-1} \|b(kT/M)\|^2}} \ket{1}_c\ket{0}_r\ket{0}_f\ket{0}_t\ket{0}_e \left( \sum_{k=0}^{M-1} \|b(kT/M)\| e^{A(T-kT/M)} \ket{b(kT/M)}_v \right)
        + \ket{1}_c\ket{\perp}. 
    \end{equation}
    
    Now we combine~\cref{eqn:FF_bt_1st} and~\cref{eqn:FF_bt_2nd} together, and the state in~\cref{eqn:FF_bt_ini} is transformed to 
    \begin{equation}
    \begin{split}
        & \frac{1}{e^{\widetilde{\alpha}T} \sqrt{\|u(0)\|^2 + T^2 \|b\|_{\text{avg}}^2}} \ket{0}_r\ket{0}_f\ket{0}_t\ket{0}_e \otimes \\
        & \left( \ket{0}_c \left(\|u(0)\| e^{AT}\ket{u(0)}_v\right) + \ket{1}_c \left( \frac{T}{M} \sum_{k=0}^{M-1} \|b(kT/M)\| e^{A(T-kT/M)} \ket{b(kT/M)}_v \right) \right) + \ket{\perp}. 
    \end{split}
    \end{equation}
    The final step is to apply a Hadamard gate on the control register to add up the two states, which yields 
    \begin{equation}\label{eqn:FF_bt_output}
         \frac{\|\widetilde{u}(T)\|}{e^{\widetilde{\alpha}T} \sqrt{2(\|u(0)\|^2 + T^2 \|b\|_{\text{avg}}^2)}} \ket{0}_c\ket{0}_r\ket{0}_f\ket{0}_t\ket{0}_e  \ket{\widetilde{u}(T)}_v  + \ket{\perp}, 
    \end{equation}
    where 
    \begin{equation}
        \widetilde{u}(T) =  e^{AT} u(0) +  \frac{T}{M} \sum_{k=0}^{M-1} e^{A(T-kT/M)} b(kT/M). 
    \end{equation}
    According to~\cref{lem:FF_bt_quadrature} and~\cref{eqn:FF_bt_M}, we have $\|\widetilde{u}(T) - u(T)\| \leq \frac{1}{2}\epsilon\|u(T)\|$. 
    Therefore, according to~\cref{lem:succ_prob_error}, $\ket{\widetilde{u}(T)}$ is an $\epsilon$-approximation of $\ket{u(T)}$. $\ket{u(T)}$ can be obtained by measuring the state in~\cref{eqn:FF_bt_output} and obtaining all $0$ in the ancilla registers which indicates success. Since the state in ~\cref{eqn:FF_bt_output} takes $\mathcal{O}(1)$ queries and extra one-qubit gates to prepare, to prepare $\ket{u(T)}$ with $\Omega(1)$ success probability, we can use a number of queries and extra one-qubit gates of order
     \begin{equation}
        \frac{e^{\widetilde{\alpha}T} \sqrt{2(\|u(0)\|^2 + T^2 \|b\|_{\text{avg}}^2)}}{\|\widetilde{u}(T)\|} \leq \frac{e^{\widetilde{\alpha}T} \sqrt{2(\|u(0)\|^2 + T^2 \|b\|_{\text{avg}}^2)}}{\|u(T)\| - \|\widetilde{u}(T)-u(T)\|} \leq \frac{e^{\widetilde{\alpha}T} \sqrt{2(\|u(0)\|^2 + T^2 \|b\|_{\text{avg}}^2)}}{(1-\epsilon/2)\|u(T)\|}. 
    \end{equation}
\end{proof}


\section{Applications of fast-forwarded algorithms to more types of PDEs}\label{app:application_PDEs}

In the main text, we show how our fast-forwarded algorithm can better solve parabolic evolutionary PDEs.
Here we present the applications of our algorithm to more types of PDEs, including hyperbolic PDEs and high-order ones. 

\subsection{Application: hyperbolic PDEs}\label{sec:app_hyperbolic_pde}

We now consider high-dimensional evolutionary PDEs of hyperbolic type. 
The main difference between hyperbolic PDEs and parabolic PDEs is that hyperbolic PDEs typically involve a second-order time derivative. 
In order to apply our fast-forwarded algorithm, which only applies to first-order equations, we will first extend the original single PDE to an equivalent system of two PDEs of both the solution and its time derivative. 
The resulting system of PDEs become first-order in time and can be fast-forwarded as in~\cref{sec:app_parabolic_pde}. 

We continue using the notation introduced in~\cref{sec:app_parabolic_pde}. 
The hyperbolic PDEs we consider have the form
\begin{equation}\label{eqn:app_PDE_hyper}
\begin{split}
     \frac{\partial^2}{\partial t^2 }u(x,t)  &= \Delta^a u(x,t) + c u(x,t) + b(x,t), \quad t \in [0,T], x \in [0,1]^d, \\
     u(x,0) &= u_0(x), \\
     \frac{\partial}{\partial t}u(x,0) &= w_0(x). 
\end{split}
\end{equation}
There are two differences compared to~\cref{eqn:app_PDE}. 
First,~\cref{eqn:app_PDE_hyper} involves a second-order time derivative but~\cref{eqn:app_PDE} only has first-order time derivatives. 
Second, there is no spatial divergence term in~\cref{eqn:app_PDE_hyper}. 
The reason is that the eigenvalues of the coefficient matrix must be real numbers in order to transform the second-order equation into a first-order system with a normal coefficient matrix.

\subsubsection{Spatial discretization and lifting}

We first spatially discretize~\cref{eqn:app_PDE_hyper} using the same approach as in~\cref{sec:app_parabolic_pde}. 
The resulting system of ODEs is 
\begin{equation}\label{eqn:app_ODE_hyper}
\begin{split}
    \frac{d^2}{dt^2}\vec{u}(t) &= (A_L^a+cI) \vec{u}(t) + \vec{b}(t), \quad t \in [0,T], \\
    \vec{u}(0) &= \vec{u}_0, \\
    \frac{d}{dt}\vec{u}(0) &= \vec{w}_0. 
\end{split}
\end{equation}
To apply~\cref{thm:FF_bt}, we need to transform this equation to a first-order ODE system. 
To this end, we follow the idea in~\cite{CostaJordanOstrander2019}.  Write $A_L^a+cI = (F_h^{\otimes d})^{-1} D F_h^{\otimes d}$ where $D$ is a diagonal matrix with diagonal entries
\begin{equation}
    \mu =  c - 4 n^2 \sum_{j=0}^{d-1} a_j \sin^2 (k_j \pi /n), \quad k_j \in [n]. 
\end{equation}
Define $B = (F_h^{\otimes d})^{-1} \sqrt{-D} F_h^{\otimes d}$ where $\sqrt{-D}$ denotes another diagonal matrix with diagonal entries to be the positive square root of $-D$'s, then $B$ is a Hermitian matrix such that $B^2 = -(A_L^a+cI)$. 
Let $\vec{v}(t)$ be another $n^d$-dimensional vector and consider the ODE for $(\vec{u},\vec{v})$ in the form 
\begin{equation}\label{eqn:app_ODE_hyper_1st}
\begin{split}
    \frac{d}{dt} \left(\begin{array}{c}
         \vec{u}(t) \\
         \vec{v}(t) 
    \end{array}
    \right) = \left(\begin{array}{cc}
        0 & iB \\
        iB & 0
    \end{array}\right) \left(\begin{array}{c}
         \vec{u}(t) \\
         \vec{v}(t) 
    \end{array}
    \right) + \left(\begin{array}{c}
         0 \\
         \vec{b}(t) 
    \end{array}
    \right), 
\end{split}
\end{equation}
which becomes a first-order ODE system. 
It is straightforward to check that $\vec{u}(t)$ satisfies~\cref{eqn:app_ODE_hyper}, and by definition, $\vec{v}(t)$ solves the linear system $iB\vec{v}(t) = \frac{d}{dt} \vec{u}(t)$. 

\subsubsection{Fast-forwarding}

We discuss quantum algorithms for~\cref{eqn:app_ODE_hyper}, \emph{i.e.}, preparing a quantum state approximate to $\vec{u}(T)$.  
Our approach is to first consider~\cref{eqn:app_ODE_hyper_1st} and prepare a quantum state encoding $(\vec{u},\vec{v})$, then post-select $\vec{u}$. 
Our fast-forwarding result can be summarized as follows. 
\begin{cor}\label{cor:app_PDE_hyper}
    Consider the semi-discretized PDE~\cref{eqn:app_ODE_hyper} up to time $T$. 
    Suppose that $c$ is always non-positive and, when $c = 0$, an extra condition $\sum_{j\in [n]^d} w_0(j/n) = 0$ holds. 
    Then there exists a quantum algorithm that outputs the normalized solution $\ket{\vec{u}(T)}$ with $\Omega(1)$ success probability, using 
    \begin{equation}
            \mathcal{O}\left( \frac{\sqrt{2\|\vec{u}(0)\|^2 + 2\|\vec{v}(0)\|^2 + 2T^2 \|\vec{b}\|^2_{\text{avg}} }}{\|\vec{u}(T)\|} \right)
    \end{equation}
    queries to the preparation oracles of $\vec{u}_0$ and $\vec{b}(t)$, 
    \begin{equation}
            \mathcal{O}\left( \frac{\|\vec{w}_0\|\sqrt{2\|\vec{u}(0)\|^2 + 2\|\vec{v}(0)\|^2 + 2T^2 \|\vec{b}\|^2_{\text{avg}} }}{\|\vec{v}(0)\|\|\vec{u}(T)\|} \right)
    \end{equation}
    queries to the preparation oracle of $\vec{w}_0$, and 
    \begin{equation}
        \widetilde{\mathcal{O}}\left(\frac{\|\vec{w}_0\|\sqrt{2\|\vec{u}(0)\|^2 + 2\|\vec{v}(0)\|^2 + 2T^2 \|\vec{b}\|^2_{\text{avg}} }}{\|\vec{v}(0)\|\|\vec{u}(T)\|} d \log^2(n) \text{~poly}\log(T/\epsilon) \right)
    \end{equation}
    additional gates. 
\end{cor}
\begin{proof}
    We first analyze the eigenvalue and eigenstate of the coefficient matrix. 
    By discrete Fourier transform and switching the blocks, 
    \begin{equation}
        \begin{split}
            & \quad \left(\begin{array}{cc}
                \frac{1}{\sqrt{2}} I & \frac{1}{\sqrt{2}} I  \\
                \frac{1}{\sqrt{2}} I & -\frac{1}{\sqrt{2}} I
            \end{array}\right) 
            \left(\begin{array}{cc}
                F_h^{\otimes d} & 0  \\
                0 & F_h^{\otimes d} 
            \end{array}\right) 
            \left(\begin{array}{cc}
                0 & iB \\
                iB & 0
            \end{array}\right) 
            \left(\begin{array}{cc}
                (F_h^{\otimes d})^{-1} & 0 \\
                0 & (F_h^{\otimes d})^{-1}
            \end{array}\right)
            \left(\begin{array}{cc}
                \frac{1}{\sqrt{2}} I & \frac{1}{\sqrt{2}} I  \\
                \frac{1}{\sqrt{2}} I & -\frac{1}{\sqrt{2}} I 
            \end{array}\right) \\
            & = \left(\begin{array}{cc}
                \frac{1}{\sqrt{2}} I & \frac{1}{\sqrt{2}} I  \\
                \frac{1}{\sqrt{2}} I & -\frac{1}{\sqrt{2}} I
            \end{array}\right) 
            \left(\begin{array}{cc}
                0 & i\sqrt{-D} \\
                i\sqrt{-D} & 0
            \end{array}\right)
            \left(\begin{array}{cc}
                \frac{1}{\sqrt{2}} I & \frac{1}{\sqrt{2}} I  \\
                \frac{1}{\sqrt{2}} I & -\frac{1}{\sqrt{2}} I
            \end{array}\right) \\
            & = \left(\begin{array}{cc}
                i\sqrt{-D} & 0 \\
                0 & -i\sqrt{-D}
            \end{array}\right). 
        \end{split}
    \end{equation}
    Therefore, the eigenvalues of the coefficient matrix are 
    \begin{equation}
         \pm i \sqrt{-c + 4 n^2 \sum_{j=0}^{d-1} a_j \sin^2 (k_j \pi /n)}, \quad k_j \in [n], 
    \end{equation}
    and the corresponding eigenstates form the unitary matrix 
    \begin{equation}
        \left(\begin{array}{cc}
                (F_h^{\otimes d})^{-1} & 0 \\
                0 & (F_h^{\otimes d})^{-1}
            \end{array}\right)
            \left(\begin{array}{cc}
                \frac{1}{\sqrt{2}} I & \frac{1}{\sqrt{2}} I  \\
                \frac{1}{\sqrt{2}} I & -\frac{1}{\sqrt{2}} I 
            \end{array}\right). 
    \end{equation}
    Therefore, the coefficient matrix of~\cref{eqn:app_ODE_hyper_1st} has closed-form eigenvalues (though they require summation of $\mathcal{O}(d)$ terms), and the eigenstate transformation matrix can be quantum implemented by $2d$ (inverse) quantum Fourier transform and a Hadamard gate. 
    According to the same reasoning of~\cref{cor:app_PDE}, we may use~\cref{thm:FF_bt} to prepare a quantum state proportional to $(\vec{u}(T),\vec{v}(T))$, with 
    \begin{equation}
            \mathcal{O}\left( \frac{\sqrt{2\|\vec{u}(0)\|^2 + 2\|\vec{v}(0)\|^2 + 2T^2 \|\vec{b}\|^2_{\text{avg}} }}{\sqrt{\|\vec{u}(T)\|^2 + \|\vec{v}(T)\|^2}} \right)
    \end{equation}
    queries to the preparation oracles of $(\vec{u}(0),\vec{v}(0))$ and $\vec{b}(t)$, and 
    \begin{equation}
            \widetilde{\mathcal{O}}\left( \frac{\sqrt{2\|\vec{u}(0)\|^2 + 2\|\vec{v}(0)\|^2 + 2T^2 \|\vec{b}\|^2_{\text{avg}} }}{\sqrt{\|\vec{u}(T)\|^2 + \|\vec{v}(T)\|^2}} d \log^2(n) \text{~poly}\log(T/\epsilon) \right)
    \end{equation}
    additional gates. 

    The state preparation oracle of $(\vec{u}(0),\vec{v}(0))$ requires the controlled version of the preparation oracles of $\ket{\vec{u}(0)}$ and $\ket{\vec{v}(0)}$.  
    The preparation oracle of $\ket{\vec{u}(0)}$ is given as assumed. 
    For the preparation oracle of $\ket{\vec{u}(0)}$, using the relation $i B \vec{v}(t) = \frac{d}{dt}\vec{u}(t)$, we need to solve the linear system 
    \begin{equation}\label{eqn:app_PDE_hyper_QLSP}
        i B \vec{v}(0) = \vec{w}_0. 
    \end{equation} 
    \cref{eqn:app_PDE_hyper_QLSP} has two features. 
    First,~\cref{eqn:app_PDE_hyper_QLSP} is not always solvable. 
    Notice that the matrix $B$ is invertible when $c \neq 0$, and in this case $\vec{v}(0)$ uniquely exists for any $\vec{w}_0$. 
    However, $B$ becomes singular when $c = 0$, and in this case we need to require $\vec{w}_0$ does not overlap with the eigenspace of $B$ corresponding to $0$ eigenvalue. 
    Since the eigenstates of $B$ form Fourier basis with periodic boundary condition, this eigenspace is $\text{span}\left\{ (1,1,\cdots,1)\right\}$. 
    Therefore, under the assumption $\sum_{j\in [n]^d} w_0(j/n) = 0$,~\cref{eqn:app_PDE_hyper_QLSP} has the solution. 
    Second, though solving general linear system of equations requires complexity linear in the condition number,~\cref{eqn:app_PDE_hyper_QLSP} can be solved using fast inversion proposed in~\cite{TongAnWiebe2021} since $B$ is a normal matrix with known eigenvalues and eigenstates. 
    According to~\cite[Proposition 7]{TongAnWiebe2021}, we need $\mathcal{O}(\|\vec{w}_0\|/\|\vec{v}(0)\|)$ queries to preparing $\ket{\vec{w}_0}$ and $(F_h)^{\otimes d}$. 
    So the overall numbers of queries to $\ket{\vec{w}_0}$ and the required additional gates due to quantum Fourier transform have an extra multiplicative factor $\|\vec{w}_0\|/\|\vec{v}(0)\|$. 

    Finally, to obtain $\ket{\vec{u}(T)}$ from a state proportional to $(\vec{u}(T),\vec{v}(T))$, we need post-selection over the correct subspace, and the average number of repeats required after amplitude amplification is of order $\sqrt{\|\vec{u}(T)\|^2 + \|\vec{v}(T)\|^2}/\|\vec{u}(T)\|$. This contributes another multiplicative factor in all the complexities. 
\end{proof}

\cref{cor:app_PDE_hyper} suggests that the overall complexity in $d$, $n$ and $T$ depends on the scalings of $\|\vec{v}(0)\|$ and $\|\vec{u}(T)\|$. 
Here we focus on a more explicit scaling of $\|\vec{v}(0)\|$, while we postpone the discussions on $\|\vec{u}(T)\|$ to specific examples. 
According to~\cref{eqn:app_ODE_hyper_1st}, the vector $\vec{v}(0)$ is determined by solving the possibly generalized linear system $i B \vec{v}(0) = \vec{w}_0$. 
As discussed in the proof, the extra condition $\sum_{j\in [n]^d} w_0(j/n) = 0$ when $c = 0$ is to ensure that such linear system is always solvable. 
We remark that this condition is a discrete analog of the continuous counterpart $\int_{[0,1]^d} w_0(x) dx = 0$ and is also required by some existing algorithm for wave equation~\cite{CostaJordanOstrander2019}. 
The solution $\vec{v}(0)$, however, may be as small as $\Theta(d^{-1/2}n^{-1}\|\vec{w}_0\|)$, leading to an extra multiplicative factor $\mathcal{O}(\sqrt{d}n)$ in the worst case. 
This is because the effective condition number of $B$ is on the scale of $\sqrt{d}n$. 
However, $\|\vec{v}(0)\|$ can also be asymptotically comparable to $\|\vec{w}_0\|$ in the best situations. 
As discussed in~\cite[Proposition 8]{TongAnWiebe2021}, this happens when $\vec{w}_0$ has $\Theta(1)$ overlap with the eigenstates corresponding to $\Theta(1)$ eigenvalues of $B$, which can be guaranteed if $w_0(x)$ is a smooth function. 
We summarize the above discussions with explicit best- and worst-case scalings of $\|\vec{v}(0)\|$ as follows. 

\begin{cor}\label{cor:app_PDE_hyper_best_worst}
    Under the same conditions as in~\cref{cor:app_PDE_hyper}, 
    \begin{enumerate}
        \item the worst-case complexity can be estimated as: 
        \begin{equation}
            \mathcal{O}\left( \frac{\sqrt{2\|\vec{u}(0)\|^2 + 2\|\vec{w}_0\|^2 + 2T^2 \|\vec{b}\|^2_{\text{avg}} }}{\|\vec{u}(T)\|} \right)
    \end{equation}
    queries to the preparation oracles of $\vec{u}_0$ and $\vec{b}(t)$, 
    \begin{equation}
            \mathcal{O}\left( \frac{\sqrt{d}n\sqrt{2\|\vec{u}(0)\|^2 + 2\|\vec{w}_0\|^2 + 2T^2 \|\vec{b}\|^2_{\text{avg}} }}{\|\vec{u}(T)\|} \right)
    \end{equation}
    queries to the preparation oracle of $\vec{w}_0$, and 
    \begin{equation}
        \widetilde{\mathcal{O}}\left(\frac{\sqrt{2\|\vec{u}(0)\|^2 + 2\|\vec{w}_0\|^2 + 2T^2 \|\vec{b}\|^2_{\text{avg}} }}{\|\vec{u}(T)\|} d^{3/2} n\log^2(n) \text{~poly}\log(T/\epsilon) \right)
    \end{equation}
    additional gates, 
    \item if we further assume that $w_0(x)$ is smooth with uniformly bounded partial derivatives, then the overall complexity can be estimated as: 
    \begin{equation}
            \mathcal{O}\left( \frac{\sqrt{2\|\vec{u}(0)\|^2 + 2\|\vec{w}_0\|^2 + 2T^2 \|\vec{b}\|^2_{\text{avg}} }}{\|\vec{u}(T)\|} \right)
    \end{equation}
    queries to the preparation oracles of $\vec{u}_0$ and $\vec{b}(t)$, 
    \begin{equation}
            \mathcal{O}\left( \frac{\sqrt{2\|\vec{u}(0)\|^2 + 2\|\vec{w}_0\|^2 + 2T^2 \|\vec{b}\|^2_{\text{avg}} }}{\|\vec{u}(T)\|} \right)
    \end{equation}
    queries to the preparation oracle of $\vec{w}_0$, and 
    \begin{equation}
        \widetilde{\mathcal{O}}\left(\frac{\sqrt{2\|\vec{u}(0)\|^2 + 2\|\vec{w}_0\|^2 + 2T^2 \|\vec{b}\|^2_{\text{avg}} }}{\|\vec{u}(T)\|} d \log^2(n) \text{~poly}\log(T/\epsilon) \right)
    \end{equation}
    additional gates. 
    \end{enumerate}
\end{cor}

Now we compare our algorithm with existing ones. 
For technical simplicity we only compare the numbers of queries to the state preparation oracles, and we focus on the scaling in $d$ and $n$ and postpone the discussions on potential speedup in terms of $T$ to specific examples. 
Due to the second-order nature of~\cref{eqn:app_PDE_hyper}, there are two ways of applying existing generic algorithms. 
The first approach is to solve~\cref{eqn:app_ODE_hyper_1st} using~\cref{lem:DEsolver_TD}. 
As shown in the proof of~\cref{cor:app_PDE_hyper}, the spectral norm of the coefficient matrix $[0,iB;iB,0]$ is $\mathcal{O}(\sqrt{d}n)$, so the query complexity of generic algorithms scales 
\begin{equation}\label{eqn:scaling1_rev}
    \widetilde{\mathcal{O}}\left( \frac{\sup_t \sqrt{\|\vec{u}(t)\|^2+\|\vec{v}(t)\|^2}}{\|\vec{u}(T)\|} T \sqrt{d}n \text{~poly}\log(1/\epsilon) \right), 
\end{equation}
which is always $\widetilde{\mathcal{O}}(\sqrt{d}n)$. 
However, notice that here the query model to the initial condition is that to $(\vec{u}(0),\vec{v}(0))$, so the oracle of preparing $\vec{v}(0)$ still requires extra steps, \emph{i.e.}, solving the linear system $i B \vec{v}(0) = \vec{w}_0$, which is the same as our fast-forwarded algorithm. 
Taken this into consideration, generic algorithm for~\cref{eqn:app_ODE_hyper_1st} uses $\widetilde{\mathcal{O}}(\sqrt{d}n)$ queries to $\vec{u}_0$ and $\vec{b}(t)$, and $\widetilde{\mathcal{O}}(\sqrt{d}n)$ to $\widetilde{\mathcal{O}}(dn^2)$ queries to $\vec{w}_0$ depending on the regularity of $w_0(x)$ in the same fashion of~\cref{cor:app_PDE_hyper_best_worst}. 
Specifically, as discussed before~\cref{cor:app_PDE_hyper_best_worst}, in the worst case, solving $iB\vec{v}(0) = \vec{w}_0$ requires $\mathcal{O}(\sqrt{d}n)$ queries to $\vec{w}_0$, so the overall complexity becomes $\mathcal{O}(dn^2)$. 
With further assumption that $w_0$ is smooth, solving $iB\vec{v}(0) = \vec{w}_0$ can be achieved with $\mathcal{O}(1)$ cost~\cite{TongAnWiebe2021}, so the overall complexity is still $\mathcal{O}(\sqrt{d}n)$.
Meanwhile, according to~\cref{cor:app_PDE_hyper_best_worst}, our fast-forwarded algorithm only needs $\mathcal{O}(\sqrt{d}n)$ in the worst case and $\mathcal{O}(1)$ with smooth condition. 
Therefore our algorithm is always better than the existing one by a multiplicative factor $\sqrt{d}n$. 

The second approach is to use a more natural way to lift~\cref{eqn:app_ODE_hyper} to first-order equations by introducing $\vec{\widetilde{v}}(t) = \frac{d}{dt} \vec{u}(t)$.  
The corresponding first-order ODE system becomes 
\begin{equation}\label{eqn:app_ODE_hyper_1st_another}
\begin{split}
    \frac{d}{dt} \left(\begin{array}{c}
         \vec{u}(t) \\
         \vec{\widetilde{v}}(t) 
    \end{array}
    \right) = \left(\begin{array}{cc}
        0 & I \\
        A_L^a + cI & 0
    \end{array}\right) \left(\begin{array}{c}
         \vec{u}(t) \\
         \vec{\widetilde{v}}(t) 
    \end{array}
    \right) + \left(\begin{array}{c}
         0 \\
         \vec{b}(t)
    \end{array}
    \right). 
\end{split}
\end{equation}
This is more natural because the input oracle of $\vec{\widetilde{v}}(0)$ is exactly the preparation oracle of $w_0$, so there is no extra linear system issue any more. 
Notice that our fast-forwarded algorithm cannot be applied to ~\cref{eqn:app_ODE_hyper_1st_another} because the coefficient matrix here is no longer normal. 
However, one may apply the generic algorithm in~\cref{lem:DEsolver_TD}. 
Notice that the coefficient matrix has spectral norm $\mathcal{O}(\sqrt{d} n)$ and condition number $\mathcal{O}(n)$~\cite{CostaJordanOstrander2019}, so the query complexity of generic algorithms scales 
\begin{equation}
    \widetilde{\mathcal{O}}\left( \frac{\sup_t \sqrt{\|\vec{u}(t)\|^2+\|\vec{v}(t)\|^2}}{\|\vec{u}(T)\|} T \sqrt{d}n^2 \text{~poly}\log(1/\epsilon) \right), 
\end{equation}
which is always $\widetilde{\mathcal{O}}(\sqrt{d}n^2)$. 
Therefore our fast-forwarded algorithm always outperforms existing one in $n$, and can be better in $d$ as well if $w_0(x)$ is smooth. 
Nevertheless, we remark that our fast-forwarded algorithm can only solve a subclass of~\cref{eqn:app_ODE_hyper} as specified in~\cref{cor:app_PDE_hyper}, while generic algorithm applied to~\cref{eqn:app_ODE_hyper_1st_another} works for all possible initial conditions. 

Finally, we remark that one may also obtain a quantum state encoding $\frac{d}{dt}\vec{u}(T)$ from our algorithm. 
This can be done by first solving~\cref{eqn:app_ODE_hyper_1st} and post-select $\vec{v}(T)$ rather than $\vec{u}(T)$, then compute $\frac{d}{dt}\vec{u}(T) = i B \vec{v}(T)$. 
To apply the operator $iB$, one may still use the QFT and controlled rotations similar as the quantum fast inversion in~\cite{TongAnWiebe2021}, since the eigenvalues and eigenstates of $iB$ are known as well.

\subsubsection{Examples}

\paragraph{Wave equation}

The wave equation, as the prototypical hyperbolic PDE, is for the description of travelling and standing waves, such as mechanical waves (e.g. water waves, sound waves, and seismic waves) or electromagnetic waves (including light waves). It arises in fields like acoustics, electromagnetism, fluid mechanics, continuum mechanics, quantum mechanics, plasma physics, general relativity, geophysics, and many other disciplines \cite{evans2010partial,courant2008methods}. Here we consider the inhomogeneity as an external source, such as the electromagnetic wave equation. The equation has the form
\begin{equation}
    \frac{\partial^2 }{\partial t^2 } u(x,t) = \Delta u + b(x,t). 
\end{equation}
This is a special case of~\cref{eqn:app_PDE_hyper} where all the entries of $a$ are equal and $c = 0$, so can be fast-forwarded using~\cref{cor:app_PDE_hyper}. 
Since in the wave equation $c = 0$, all the eigenvalues of the corresponding coefficient matrix after spatial discretization and lifting are purely imaginary. 
Therefore~\cref{eqn:app_ODE_hyper_1st} corresponding to the homogeneous wave equation is just a Hamiltonian simulation problem, and the exponential decay related to the norm decay is avoided. 
Furthermore, the coefficient matrix has zero eigenvalues, which suggests, in the inhomogeneous case, the possibility that $\|\vec{u}(T)\|$ might be $\Theta(T)$ and thus the overall query complexity in~\cref{cor:app_PDE_hyper} can be $\mathcal{O}(1)$. 
However, noticing that the condition that $(0,\vec{b}(t))$ overlaps with eigenstates corresponding to the zero eigenvalues only implies the linear growth in time of the extended solution $(\vec{u}(t),\vec{v}(t))$, it is not clear yet whether the linear growth is contributed by $\vec{u}(t)$, so the overall query complexity to approximate $\ket{\vec{u}(T)}$ might still be linear in $T$. 

\paragraph{Klein-Gordon equation}

The Klein-Gordon equation is a relativistic wave equation that characterizes a field whose quanta are spinless particles. It is a quantized version of the relativistic energy–momentum relation. The Klein-Gordon equation plays the role of one of the fundamental equations of quantum field theory \cite{weinberg1995quantum,gross1999relativistic,greiner2000relativistic}.
The equation goes as 
\begin{equation}
    \frac{\partial^2 }{\partial t^2 } u(x,t) = \Delta u - m^2 u + b(x,t). 
\end{equation}
This is a special case of~\cref{eqn:app_ODE_hyper_1st} with all equal entries in $a$. 
Notice that, since $m > 0$, the eigenvalues of the coefficient matrix have strictly negative real part, so the solution $\vec{u}(t)$ cannot grow linearly. 
In this case, according to~\cref{cor:app_PDE_hyper_best_worst}, provided that the solution does not decay fast, the overall query complexity is always linear in time, and our fast-forwarded algorithm achieves speedup compared to existing generic ones only in $d$ and $n$.

\subsection{Application: higher-order PDEs}\label{sec:app_high_order_pde}

We have discussed how our algorithm can be applied to several parabolic and hyperbolic PDEs and achieve speedups over existing quantum algorithms. These two types of PDEs are second-order PDEs and only involve spatial derivatives up to second order. 
However, our algorithm can also be applied to PDEs with higher-order spatial derivatives. This is because the spatial derivative of any order (under periodic boundary conditions) has the Fourier basis functions as its eigenfunctions, and so its discretized version can be diagonalized by the QFT.  In this section, we discuss two higher-order examples of PDEs that cannot be characterized as either parabolic or hyperbolic: the Airy equation and the beam equation. For technical simplicity, we only consider the one-dimensional case. 

\paragraph{Airy equation}
The Airy equation is the linearized version of the Korteweg–De Vries (KdV) equation, which was initially introduced to model shallow water waves and found to have various applications in modeling long internal waves in oceans~\cite{SegurHammack1982}, ion acoustic waves in a plasma~\cite{konnoIchikawa1974}, and the Fermi–Pasta–Ulam–Tsingou problem in chaos theory~\cite{BenettinChristodoulidiPonno2013}. 
When the group velocity of the wave is at an extremum, the wave envelope is approximated by the linearized KdV equation, \emph{i.e.}, the Airy equation~\cite{Haberman2012applied}.
The equation has the form\footnote{The Airy equation also commonly refers to $\frac{d^2y}{dx^2}-xy = 0$. The connection between this equation and \cref{eq:airy} is that if $y(x)$ solves the Airy ODE then $u(x,t) = t^{-1/3}y(x/(3t)^{1/3})$ solves \cref{eq:airy} with $b\equiv 0$.}
\begin{equation}\label{eq:airy}
    \frac{\partial }{\partial t} u(x,t) + \frac{\partial^3 }{\partial x^3} u(x,t) = b(x,t). 
\end{equation}

We first discretize $\frac{\partial^3 }{\partial x^3}$ by the central difference method using $n$ grid points $x_j = j/n, j \in [n]$ and spatial step size $h = 1/n$. 
According to the central difference formula $f^{(3)}(x) \approx \frac{1}{h^3}( -f(x-2h)/2 + f(x-h) - f(x+h) + f(x+2h)/2 )$, we may define the discretized operator as $D_{h,3}$ of which all the non-zero entries are 
\begin{equation}
    D_{h,3}(j,j-2) = -\frac{1}{2h^3}, \quad D_{h,3}(j,j-1) = \frac{1}{h^3}, \quad D_{h,3}(j,j+1) = -\frac{1}{h^3}, \quad D_{h,3}(j,j+2) = \frac{1}{2h^3}. 
\end{equation}
Here, under periodic boundary condition, we interpret $D_{h,3}(j,k) = D_{h,3}(j,k\pm n)$ to let the index fall into $[n]$. 
The corresponding semi-discretized equation is given as 
\begin{equation}\label{eqn:app_airy_ode}
    \frac{d}{dt} \vec{u}(t) = -D_{h,3} \vec{u}(t) + \vec{b}(t). 
\end{equation}

The eigenvalues of $D_{h,3}$ are $\mu_k = - 4in^3 \sin(2k\pi/n)\sin^2(k\pi/n)$, and the corresponding eigenstate is $\frac{1}{\sqrt{n}}(1,\omega_n^k,\omega_n^{2k},\cdots,\omega_n^{(n-1)k})^{T}$ where $\omega_n = e^{2\pi i/n}$.
Therefore, the coefficient matrix $-D_{h,3}$ can be diagonalized using QFT, and its eigenvalues have closed-form expression. 
~\cref{thm:FF_bt} directly implies that the overall query complexity of our fast-forwarded algorithm applied to~\cref{eqn:app_airy_ode} is 
\begin{equation}
    \mathcal{O}\left( \frac{\sqrt{2\|\vec{u}(0)\|^2 + 2T^2 \|b\|^2_{\text{avg}} }}{\|\vec{u}(T)\|} \right). 
\end{equation}
Again, this is always independent of $n$, at most linear in $T$, and in the best case can be independent of $T$ as well. 
As a comparison, since $\|D_{h,3}\| = \Theta(n^3)$, query complexities of best existing algorithms (\cref{lem:DEsolver_TI} and~\cref{lem:DEsolver_TD}) have qubit scaling in $n$ and linear scaling in $T$. 
So our algorithm always achieves a speedup in $n$ and can potentially achieve a speedup in $T$ with further assumption on $\vec{b}(t)$.

\paragraph{Beam equation}

The beam equation (also known as the Euler–Bernoulli equation) describes the relationship between the beam's deflection and the applied load. It is widely used in engineering practice, especially civil, mechanical, and aeronautical engineering. Civil engineering structures often consist of an assembly or grid of beams with cross-sections having shapes such as T's and I's. A large number of machine parts also are beam-like structures: lever arms, shafts, etc. Several aeronautical structures such as wings and fuselages can also be treated as thin-walled beams. The beam models provide the designer with a simple tool to analyze the structures, and make it successfully to solve a variety of engineering problems \cite{witmer1991elementary,bauchau2009euler,carrera2011beam}. The equation has the form
\begin{equation}
    \frac{\partial^2 }{\partial t^2} u(x,t) + \frac{\partial^4 }{\partial x^4} u(x,t) = b(x,t). 
\end{equation}

We follow the approach for hyperbolic PDEs that we first spatially discretize the PDE to obtain a system of second-order ODEs, then transfer it to a system of first-order ones in order to apply our fast-forwarded algorithm. 
According to the central difference formula $f^{(4)}(x) \approx \frac{1}{h^4} (f(x-2h) - 4f(x-h) + 6f(x) - 4f(x+h) + f(x+2h))$, we may discretize the differential operator $\frac{\partial^4}{\partial x^4}$ by $D_{h,4}$ of which all the non-zero entries are 
\begin{equation}
    D_{h,3}(j,j\pm 2) = \frac{1}{h^4}, \quad D_{h,3}(j,j\pm 1) = -\frac{4}{h^4}, \quad D_{h,3}(j,j) = \frac{6}{h^4}. 
\end{equation}
The corresponding semi-discretized equation is given as 
\begin{equation}\label{eqn:app_beam_ode}
    \frac{d^2}{dt^2} \vec{u}(t) = -D_{h,4} \vec{u}(t) + \vec{b}(t). 
\end{equation}

One may compute that the eigenvalues of $D_{h,4}$ are $\mu_k = 16n^4 \sin^4(k\pi/n)$, and the corresponding eigenstate is $\frac{1}{\sqrt{n}}(1,\omega_n^k,\omega_n^{2k},\cdots,\omega_n^{(n-1)k})^{T}$ where $\omega_n = e^{2\pi i/n}$. 
So we may write $D_{h,4} = F_h^{-1} \Lambda_{4} F_h$ where $F_h$ represents the inverse QFT and $\Lambda_{4} = \text{diag}(16n^4 \sin^4(k\pi/n))$. 
Define $B_{4} =  F_h^{-1} \sqrt{\Lambda_{4}} F_h $ where $\sqrt{\Lambda_{4}} = \text{diag}( 4n^2 \sin^2(k\pi/n))$. 
Then $B_4^2 = D_{h,4}$. 
Consider the ODE system 
\begin{equation}\label{eqn:app_beam_ODE_1st}
\begin{split}
    \frac{d}{dt} \left(\begin{array}{c}
         \vec{u}(t) \\
         \vec{v}(t) 
    \end{array}
    \right) = \left(\begin{array}{cc}
        0 & iB_4 \\
        iB_4 & 0
    \end{array}\right) \left(\begin{array}{c}
         \vec{u}(t) \\
         \vec{v}(t) 
    \end{array}
    \right) + \left(\begin{array}{c}
         0 \\
         \vec{b}(t) 
    \end{array}
    \right), 
\end{split}
\end{equation}
Then $\vec{u}(t)$ is also the solution of~\cref{eqn:app_beam_ode}, and by definition $i B_4 \vec{v}(t) = \frac{d}{dt} \vec{u}(t)$. 
Notice that~\cref{eqn:app_beam_ODE_1st} has the same form of~\cref{eqn:app_ODE_hyper_1st}, so all the discussions in~\cref{sec:app_hyperbolic_pde} still hold for beam equation.

\end{document}